\documentclass[a4paper,UKenglish,cleveref, autoref, thm-restate]{lipics-v2021}

\pdfoutput=1 %
\nolinenumbers

\usepackage{graphicx} %
\usepackage{tikz}
\usepackage{framed}
\usetikzlibrary{automata}

\author{Bader Abu Radi}{School of Computer Science and Engineering, The Hebrew University, Jerusalem, Israel\footnote{This work was completed after the submission of the PhD thesis of this author.}}{}{https://orcid.org/0000-0001-8138-9406}{}%

\author{Rüdiger Ehlers}{Clausthal University of Technology, Clausthal-Zellerfeld, Germany}{}{https://orcid.org/0000-0002-8315-1431}{Funded by Volkswagen Foundation}

\authorrunning{Bader Abu Radi and Rüdiger Ehlers} %

\Copyright{Bader Abu Radi and Rüdiger Ehlers} %



\EventEditors{SubmittedPaper}
\EventNoEds{2}
\EventLongTitle{Submitted to a conference}
\EventShortTitle{CONFERENCE 2025}
\EventAcronym{CONFERENCE}
\EventYear{2025}
\EventDate{Month}
\EventLocation{Event Location}
\EventLogo{}
\SeriesVolume{??}
\ArticleNo{1}

\renewcommand{\copyrightline}{\relax}
\renewcommand{\doiHeading}{\relax}
\makeatletter
\newcommand{\@hideLIPIcs}{j}
\renewcommand{\@oddfoot}{\relax}%
\makeatother

\newcommand{\buchi}{B\"uchi }
\newcommand{\A}{{\cal A}}
\newcommand{\B}{{\cal B}}
\newcommand{\C}{{\cal C}}
\newcommand{\D}{{\cal D}}
\newcommand{\SC}{{\mathcal{S}}}

\newcommand{\zug}[1]{\langle #1 \rangle}
\newcommand\bigcomment[1]{}

\title{Characterizing the Polynomial-Time Minimizable $\omega$-Automata}

\begin{document}

\maketitle

\begin{abstract}
A central question in the theory of automata is which classes of automata can be minimized in polynomial time. We close the remaining gaps for deterministic and history-deterministic automata over infinite words by proving that deterministic co-Büchi automata with transition-based acceptance are NP-hard to minimize, as are history-deterministic Büchi automata with transition-based acceptance.
\end{abstract}

\section{Introduction}

Automata over infinite words ($\omega$-automata) have proven to be an indispensible tool for the formal analysis and synthesis of reactive systems. They are useful to describe the unwanted (or wanted) behavior of a system under concern.
Starting from a specification in a temporal logic such as linear temporal logic (LTL), a translation procedure can be applied to obtain an automaton that accepts the words that are models of the LTL property. 
By following this path, a human engineer can use the simplicity of temporal logic to express the required system behavior
while the resulting automaton over infinite words can be used as input to a verification or synthesis algorithm, which are seldom applicable directly to a specification in some logic. Several different types of $\omega$-automata are known in the literature that differ by their branching modes and acceptance conditions. 
For some applications, particularly in probabilistic model checking and reactive synthesis, \emph{deterministic} or \emph{history-deterministic} automata are needed, whereas for other applications, such as model checking of finite-state systems, \emph{non-deterministic} automata suffice. 

Automata are most useful for practical applications when they have a small number of states, as the computation time and memory requirements of algorithms taking automata as input grow when automata are large.
Unfortunately, when translating from LTL to deterministic automata with a suitable acceptance condition (such as \emph{parity acceptance}), a doubly-exponential blow-up cannot be avoided in the worst case.
This worst-case blow-up however only showcases the succinctness of the logic used, and many languages of interest have quite small automata. This observation raises the question of how to \emph{minimize} automata in case due to the generality of the translation procedures from LTL to automata, their results are unnecessarily large.

\begin{table}
\centering
\begin{tabular}{c||c|c||c|c||c|c}
\textbf{Branching:} &
\multicolumn{2}{c||}{\textbf{Non-deterministic}} & 
\multicolumn{2}{c||}{\textbf{History-deterministic}} &
\multicolumn{2}{c}{\textbf{Deterministic}} \\ \hline
\textbf{Acceptance:} & SB & TB & SB & TB & SB & TB \\ \hline \hline
\textbf{Parity acc.:} & \multicolumn{2}{c||}{\multirow{6}{*}{PSPACE \cite{DBLP:reference/mc/Kupferman18}}} & \multicolumn{1}{l|}{\multirow{3}{*}{NP \cite{DBLP:conf/fsttcs/Schewe20}}} & \multirow{2}{*}{\textbf{NP}} & \multirow{3}{*}{NP \cite{DBLP:conf/fsttcs/Schewe10}} & \multirow{3}{*}{\textbf{NP}} \\ \cline{1-1}
\textbf{Büchi acc.:} & \multicolumn{2}{l||}{} & \multicolumn{1}{l|}{} & & & \\ \cline{1-1} \cline{5-5}
\textbf{co-Büchi acc.:} & \multicolumn{2}{l||}{} & \multicolumn{1}{l|}{} & \multirow{3}{*}{PTIME \cite{DBLP:journals/lmcs/RadiK22}} & &\\ \cline{1-1} \cline{4-4} \cline{6-7}
\textbf{Weak acc.:} & \multicolumn{2}{l||}{} & PTIME & & \multicolumn{2}{c}{\multirow{2}{*}{PTIME \cite{DBLP:journals/ipl/Loding01}} } \\ \cline{1-1}
\textbf{Safety acc.:} & \multicolumn{2}{l||}{ } & \cite{DBLP:conf/icalp/RadiK23} & &\multicolumn{2}{l}{ } \\ 
\end{tabular}
\caption{Complexity overview of automaton minimization, where we consider all  combinations between a large number of acceptance conditions and state-based (SB) or transition-based (TB) acceptance. All non-PTIME problems are complete for the given complexity classes (for the decision version of the minimization problem, i.e., checking if an equivalent automaton with a given number of states exists). All new results from this paper are written bold-faced. For the new results, the reasons for containment of the respective minimization problem in NP can be found at the end of Section~\ref{sec:Overview}.
}
\label{tab:complexities}
\end{table}

Unfortunately, minimization turns out to be a computationally difficult problem in general. For non-deterministic automata over infinite words, even for the simple \emph{Büchi} acceptance condition and the case of a \emph{weak} automata, where in every strongly connected component, there are only accepting states or only non-accepting states, checking if an automaton has an equivalent one with a given number of states is PSPACE-complete \cite{DBLP:reference/mc/Kupferman18}.%

For deterministic automata, the minimization complexity is quite nuanced, ranging from polynomial-time (for determistic weak automata, \cite{DBLP:journals/ipl/Loding01}) to NP-complete for deterministic Büchi and co-Büchi automata with state-based acceptance \cite{DBLP:conf/fossacs/Schewe09}. The exact boundary between automaton types that can be minimized in polynomial-time and those that cannot is however not yet known. Most importantly, for deterministic Büchi and co-Büchi automata with transition-based acceptance, the minimization complexity is currently unknown.

The situation is similar in the context of history-deterministic automata, which extend deterministic acceptance by a restricted form of non-determinism. Here, for co-Büchi automata with transition-based acceptance, a polynomial-time minimization algorithm is known, and for co-Büchi and Büchi automata with state-based acceptance, it is known that minimizing these automata is NP-complete. Yet, for transition-based acceptance and the Büchi acceptance condition, the minimization complexity is still open.

With the many branching modes and acceptance conditions known in the literature as well as the possibility to select state-based or transition based acceptance, the question arises where exactly the boundary between automaton types with polynomial-time and NP-hard minimization problems lie. While several individual results have shed some light on the complexity landscape, the precise boundary has not been identified yet.

In this paper, we provide the missing pieces.
We show that history-deterministic Büchi automata and deterministic co-Büchi  automata are NP-hard to minimize even in the case of transition-based acceptance. Together with the cases implied by these results, the  complexity landscape of automaton minimization is now complete.
Table~\ref{tab:complexities} shows the resulting complexity overview. To maintain a reasonable scope, we focus on those acceptance conditions which in the context of games permit positional winning strategies for both players, and refer to \cite{DBLP:conf/csl/Casares22} for NP-hardness results for Rabin (and Streett) acceptance conditions and a discussion of Muller automata with transition-based acceptance. We also exclude generalized co-Büchi and generalized Büchi acceptance conditions, as for them there can be a conflict between minimizing the number of states and the size of the acceptance condition (a.k.a the index of the automaton). We refer to recent work by Casares et al.~\cite{DBLP:conf/csl/CasaresIKM025} for several results on the complexity of minimizing both history-deterministic and deterministic such automata with transition-based acceptance.

Apart from completing the complexity overview of automaton minimization, our results prove that it is really the combination of history-determinism, transition-based acceptance, and the co-Büchi acceptance condition that enabled the polynomial-time minimization procedure by Abu Radi and Kupferman \cite{DBLP:journals/lmcs/RadiK22} -- having only two of these three aspects is now proven to be insufficient and this concrete combination is the only sweet spot (with polynomial-time minimization if P$\neq$NP) beyond the weak automata. 
Our results hence also provide an additional motivation for the chain-of-co-Büchi automaton representation for arbitrary $\omega$-regular languages in \cite{DBLP:conf/fsttcs/EhlersS22}, which employs minimized history-deterministic co-Büchi automata with transition-based acceptance for each chain element.

\section{Preliminaries}
\label{prelim}

\subsection{Graphs}
\label{regular graphs prelim}
For a natural number $n\geq 1$, let $[n]$ denote the set $\{ 1, 2, \ldots, n\}$.
We consider directed graphs $G = \zug{V, E}$ with a finite nonempty set $V$ of vertices and a set $E \subseteq V \times V$ of edges. For simplicity, we assume that $G$ has no self-loops or parallel edges. 
A {\em path} $p$ in $G$ is a finite sequence of vertices $p =  v_1, v_2, \ldots, v_k$, where $E(v_i, v_{i+1})$ for all $i \in [k-1]$. We say that $p$ is {\em simple} if it visits every vertex of $G$ at most once.
The {\em length} of $p$ is defined as the number of edges it traverses, thus $k-1$. A {\em cycle} in $G$ is a path from a vertex to itself.
We say that $G$ is {\em strongly-connected} when there is a path from $u$ to $v$ for every two vertices $u$  and $v$ in $V$.
Note that when $G$ is strongly-connected, it only takes polynomial time to compute a path of polynomial length that visits every vertex of $G$ using standard traversal algorithms. %

For vertices $u, v\in V$, we define the {\em distance from $u$ to $v$}, denoted $d(u, v)$, as the minimal path length from $u$ to $v$ when such a path exists, and define it as $\infty$ otherwise.  
We say that $v$ is a {\em neighbor} of $u$ when $E(u, v)$. Thus, $v$ is a neighbor of $u$ when $d(u, v) = 1$.
Then, the {\em neighborhood} of $u$, denoted $\eta(u)$, is defined as  the set of vertices 
of distance at most $1$ from $u$; thus $\eta(u) = \{ v\in V: d(u, v) \leq 1\}$;
in particular $u\in \eta(u)$.

\subsection{Automata}
For a finite nonempty alphabet $\Sigma$, an infinite {\em word\/} $w = \sigma_1 \cdot \sigma_2 \cdots \in \Sigma^\omega$ is an infinite sequence of letters from $\Sigma$. 
	A {\em language\/} $L\subseteq \Sigma^\omega$ is a set of words. We denote the empty word by $\epsilon$, and the set of finite words over $\Sigma$ by $\Sigma^*$.
 For $i\geq 0$, we use $w[1, i]$ to denote the (possibly empty) prefix $\sigma_1\cdot \sigma_2 \cdots  \sigma_i$ of $w$, use $w[i+1]$ to denote the $(i+1)$'th letter of $w$, $\sigma_{i+1}$, and use $w[i+1, \infty]$ to denote the infinite suffix $\sigma_{i+1} \cdot  \sigma_{i+2} \cdots$.

  A \emph{nondeterministic automaton} over infinite words is a tuple $\A = \langle \Sigma, Q, q_0, \delta, \alpha  \rangle$, where $\Sigma$ is an alphabet, $Q$ is a finite set of \emph{states}, $q_0\in Q$ is an \emph{initial state}, $\delta: Q\times \Sigma \to 2^Q\setminus \{\emptyset\}$ is a \emph{transition function}, and $\alpha$ is an \emph{acceptance condition}, to be defined below. For states $q$ and $s$ and a letter $\sigma \in \Sigma$, we say that $s$ is a $\sigma$-successor of $q$ if $s \in \delta(q,\sigma)$.  
	Note that we assume that $\A$ has a  total transition function, that is, $\delta(q, \sigma)$ is non-empty for all states $q$ and letters $\sigma$. %
	If $|\delta(q, \sigma)| = 1$ for every state $q\in Q$ and letter $\sigma \in \Sigma$, then $\A$ is \emph{deterministic}.
    We define the \emph{size} of $\A$, denoted $|\A|$, as its number of states, thus, $|\A| = |Q|$.
	
	When $\A$ runs on an input word, it starts in the initial state and proceeds according to the transition function. Formally, a \emph{run}  of $\A$ on $w = \sigma_1 \cdot \sigma_2 \cdots \in \Sigma^\omega$ is an infinite sequence of states $r = r_0,r_1,r_2,\ldots \in Q^\omega$, such that $r_0 = q_0$, and for all $i \geq 0$, we have that $r_{i+1} \in \delta(r_i, \sigma_{i+1})$. 
	We sometimes consider finite runs on finite words. In particular, 
	we sometimes extend $\delta$ to sets of states and finite words. Then, $\delta: 2^Q\times \Sigma^* \to 2^Q$ is such that for every $S \in 2^Q$, finite word $u\in \Sigma^*$, and letter $\sigma\in \Sigma$, we have that $\delta(S, \epsilon) = S$, $\delta(S, \sigma) = \bigcup_{s\in S}\delta(s, \sigma)$, and $\delta(S, u \cdot \sigma) = \delta(\delta(S, u), \sigma)$. Thus, $\delta(S, u)$ is the set of states that $\A$ may reach when it reads $u$ from some state in $S$. 
    The transition function $\delta$ induces a transition relation $\Delta \subseteq Q\times \Sigma \times Q$, where for every two states $q,s\in Q$ and letter $\sigma\in \Sigma$, we have that $\langle q, \sigma, s \rangle \in \Delta$ iff $s\in \delta(q, \sigma)$. 
	We sometimes view the run $r = r_0,r_1,r_2,\ldots$ on $w = \sigma_1 \cdot \sigma_2 \cdots$ as an infinite sequence of successive transitions $\zug{r_0,\sigma_1,r_1}, \zug{r_1,\sigma_2,r_2},\ldots \in \Delta^\omega$.
	The acceptance condition $\alpha$ determines which runs are ``good''. We consider here \emph{transition-based} automata, in which $\alpha$ is a set of transitions. specifically, $\alpha\subseteq \Delta$. We use the terms {\em $\alpha$-transitions\/} and  {\em $\bar{\alpha}$-transitions\/} to refer to  transitions in $\alpha$ and in $\Delta \setminus \alpha$, respectively.

 We say that $\A$ is \emph{$\alpha$-homogenous} if for every state $q\in Q$ and letter $\sigma \in \Sigma$, either all the $\sigma$-labeled transitions from $q$ are $\alpha$-transitions, or they are all $\bar{\alpha}$-transitions. %
	For a run $r \in \Delta^\omega$, let ${\it inf}(r)\subseteq \Delta$ be the set of transitions that $r$ traverses infinitely often. Thus, 
	${\it inf}(r) = \{  \langle q, \sigma, s\rangle \in \Delta: q = r_i, \sigma = \sigma_{i+1} \text{ and } s = r_{i+1} \text{ for infinitely many $i$'s}   \}$. 
    We consider {\em \buchi} and {\em co-\buchi} automata.
     In {\em B\"uchi\/} automata, $r$ is \emph{accepting} iff ${\it inf}(r)\cap \alpha \neq \emptyset$, thus if $r$ traverses transitions in $\alpha$ infinitely many times. Dually,
 in {\em co-B\"uchi\/} automata, $r$ is \emph{accepting} iff ${\it inf}(r)\cap \alpha = \emptyset$. %
 We refer to \cite{DBLP:reference/mc/Kupferman18} for a discussion of parity automata, which are a strict generalization of both co-Büchi and Büchi automata and used only in Table~\ref{tab:complexities} in this paper.
 \color{black} A run that is not accepting is \emph{rejecting}.  A word $w$ is accepted by $\A$ if there is an accepting run of $\A$ on $w$. The language of $\A$, denoted by $L(\A)$, is the set of words that $\A$ accepts. Two automata are \emph{equivalent} if their languages are equivalent. 
We use four-letter acronyms in $ \{\text{t}\}\times \{\text{D}, \text{N}\} \times \{\text{B}, \text{C}\}\times \{\text{W}\}$ to denote the different automata classes. The first letter $t$ indicates that we consider automata with transition-based acceptance; the second letter stands for the branching mode of the automaton (deterministic or nondeterministic); the third for the acceptance condition type (\buchi or co-\buchi); and the last indicates that we consider automata on words.

For a state $q\in Q$ of an automaton $\A = \langle \Sigma, Q, q_0, \delta, \alpha \rangle$, we define $\A^q$ to be the automaton obtained from $\A$ by setting the initial state to be $q$. Thus, $\A^q = \langle \Sigma, Q, q, \delta, \alpha \rangle$. 
	We say that two states $q,s\in Q$ are \emph{equivalent}, denoted $q \sim_{\A} s$, if $L(\A^q) = L(\A^s)$. %

An automaton $\A$ is \emph{history-deterministic} (\emph{HD}, for short) if its nondeterminism can be resolved based on the past, thus on the prefix of the input word read so far. Formally, $\A$ is \emph{HD} if there exists a {\em strategy\/} $f:\Sigma^* \to Q$ such that the following hold: 
	\begin{enumerate}
		\item 
		The strategy $f$ is consistent with the transition function. That is, $f(\epsilon)=q_0$, and for every finite word $u \in \Sigma^*$ and letter $\sigma \in \Sigma$, we have that $\zug{f(u),\sigma,f(u \cdot \sigma)} \in \Delta$. 
		\item
		Following $f$ causes $\A$ to accept all the words in its language. That is, for every infinite word $w = \sigma_1 \cdot \sigma_2 \cdots \in \Sigma^\omega$, if $w \in L(\A)$, then the run $f(w[1, 0]), f(w[1, 1]), f(w[1, 2]), \ldots$, which we denote by $f(w)$, is %
		an accepting run of $\A$ on $w$. 
	\end{enumerate}
	We say that the strategy $f$ \emph{witnesses} $\A$'s HDness. 
	Note that every deterministic automaton is HD. 
 We say that an automaton $\A_d$ is a {\em deterministic pruning of $\A$} of when $\A_d$ 
 is a deterministic automaton obtained by pruning some of $\A$'s transitions.

	Consider a directed graph $G = \langle V, E\rangle$. A \emph{strongly connected set\/} in $G$ (SCS, for short) is a set $C\subseteq V$ such that for every two vertices $v, v'\in C$, there is a path from $v$ to $v'$. A SCS is \emph{maximal} if it is maximal w.r.t containment, that is, for every non-empty set $C'\subseteq V\setminus C$, it holds that $C\cup C'$ is not a SCS. The \emph{maximal strongly connected sets} are also termed \emph{strongly connected components} (SCCs, for short). %
	An automaton $\A = \langle \Sigma, Q, q_0, \delta, \alpha\rangle$ induces a directed graph $G_{\A} = \langle Q, E\rangle$, where $\langle q, q'\rangle\in E$ iff there is a letter $\sigma \in \Sigma$ such that $\langle q, \sigma, q'\rangle \in \Delta$. The SCSs and SCCs of $\A$ are those of $G_{\A}$.

 For a tNCW $\A$, we refer to $\overline{\alpha}$-transitions as safe (or \emph{accepting}) transitions, and $\alpha$-transitions as rejecting transitions.
A run $r$ of $\A$ is \emph{safe} if it traverses only safe transitions. 
The \emph{safe language} of $\A$, denoted $L_{\it safe}(\A)$, is the set of infinite words $w$ such that there is a safe run of $\A$ on $w$, and the safe language of a state $q$, denoted $L_{\it safe}(q)$, is the language $L_{\it safe}(\A^q)$.
	We refer to the maximal SCCs we get by removing $\A$'s $\alpha$-transitions as the \emph{safe components} of $\A$ (a.k.a \emph{accepting SCCs}); that is, the \emph{safe components} of $\A$ are the SCCs of the graph $G_{\A^{\bar{\alpha}}} = \langle Q, E^{\bar{\alpha}} \rangle$, where $\zug{q, q'}\in E^{\bar{\alpha}}$ iff there is a letter $\sigma\in \Sigma$ such that $\langle q, \sigma, q'\rangle \in \Delta \setminus \alpha$.
    We denote the set of safe components of $\A$ by $\SC(\A)$.  
    For a state $q\in Q$, we denote the safe component of $q$ in $\A$ by $\SC^\A(q)$. When $\A$ is clear from the context, we write $\SC(q)$. 
	Note that an accepting run of $\A$ eventually gets trapped in one of $\A$'s safe components. Then,  $\A$ is \emph{normal} if
there are no safe transitions connecting different safe components. That is,
for all states $q$ and $s$ of $\A$, if there is a safe run  from $q$ to $s$, then there is also a safe run from $s$ to $q$.

\begin{example}
{\rm 
Consider the tDCW $\A$ on the right. Dashed transitions are rejecting transitions.
The tDCW $\A$ recognizes the language of all infinite words with finitely many $a$'s or finitely many $b$'s. Then,  $\A$ has three safe components: $\SC(\A) = \{ \{q_0\}, \{q_1\},\{q_2\}\}$.}, and
$\A$ being deterministic implies that it is HD. 

\noindent
\begin{minipage}{3.3in}
\vspace*{3pt}
{\rm %
Note that the word $w = a^\omega$ is such that the run $r$ of $\A^{q_0}$ on $w$ is safe, and thus $w\in L_{\it safe} (\A^{q_0})$, yet $r$ leaves the safe component $\{q_0\}$. If, however, we make $\A$ normal by turning the transitions $\zug{q_0, a, q_1}$ and $\zug{q_1, b, q_2}$ to be rejecting, then we get that safe runs from a state $q$ are exactly the runs from $q$  that do not leave the safe component $\SC(q)$.\hfill \qed} 
\end{minipage} \ \hspace{.1in}
\begin{minipage}{1.7in}
\begin{tikzpicture}
\node[state,fill=black!10!white] (q0) at (0,0) {$q_0$};
\node[state,fill=black!10!white] (q1) at (2,0) {$q_1$};
\node[state,fill=black!10!white] (q2) at (4,0) {$q_2$};
\draw[->,semithick] (q0) edge[bend left=10] node[above] {$a$} (q1);
\draw[->,semithick] (q1) edge[bend left=10] node[above] {$b$} (q2);
\draw[->,semithick] (q0) edge[loop above] node[above] {$b$} (q0);
\draw[->,semithick] (q1) edge[loop above] node[above] {$a$} (q1);
\draw[->,semithick,dashed] (q2) edge[bend left=30] node[below] {$a,b$} (q0);
\draw[fill=black] (-0.7,0.5) circle (0.05cm);
\draw[->,semithick] (-0.7,0.5) -- (q0);
\end{tikzpicture}
\end{minipage}
\end{example}

The tNCWs we define in the following are normal, thus satisfying that safe runs from a state are exactly the runs from it  that do not leave its safe component.

For a class $\gamma$ of  automata,  e.g., $\gamma \in \{ \text{tDBW}, \text{HD-tNCW} \}$ we say that a $\gamma$ automaton $\A$ is {\em minimal} if for every
 equivalent $\gamma$ automaton $\B$, it holds that $|\A|\leq |\B|$. %
In this paper, we study the {\em HD-negativity problem for a co-\buchi language}, where we are given a minimal HD-tNCW $\A$, and we need to decide whether there is a tDCW for $L(\A)$ of size $|\A|$.
Thus, the problem is to decide whether history-determinism is not beneficial to the  state-space complexity of an automaton representation of the given language. 
Then, we study the {\em $\gamma$-minimization problem}, for $\gamma\in \{\text{tDBW, tDCW, HD-tNBW}\}$, where we are given a $\gamma$ automaton $\A$ and a bound $k\geq 1$, and we need to decide whether $\A$ has an equivalent $\gamma$ automaton with at most $k$ states.

In Appendix~\ref{sec:additionalPrelims}, we define the notions that are used only in proofs in the appendix.

\section{Overview}
\label{sec:Overview}
In this section, we provide an overview of how the main results of the paper are derived, namely the NP-hardness proofs of minimizing deterministic transition-based co-Büchi automata and history-deterministic transition-based Büchi automata.

Both results build on a common foundation.
We start by defining a variant of the Hamiltonian path problem that we prove to be NP-hard. We then show how to translate a graph for this restricted Hamiltonian path problem to a history-deterministic co-Büchi automaton. This co-Büchi automaton is shown to be already minimized, and it has an equally sized deterministic co-Büchi automaton if and only if a Hamiltonian path exists in the graph. We then show that it is always possible to build a determinstic co-Büchi automaton of size polynomial in the size of the graph for the same language as the said minimal history-deterministic co-Büchi automaton. Doing so proves NP-hardness of deterministic co-Büchi automaton minimization as we can use a minimizer for deterministic co-Büchi automata on an automaton built from a graph and compare the resulting automaton size against the history-deterministic automaton size in order to determine whether the graph has a Hamiltonian path.

To support the understandability of the encoding of a graph to a co-Büchi automaton, we split it into two steps. The first step comprises the main ideas, but yields an automaton with an alphabet of size exponential in the graph size (Section~\ref{exp red sec}). The automaton built in this step has  accepting strongly connected components for each vertex in the graph such that two components can be connected in a deterministic automaton without additional helper states if there is an edge between the corresponding vertices in the graph. We add another accepting strongly connected component that is specifically engineered so that it needs to be included in every loop across multiple accepting strongly connected components in the automaton. The construction ensures that any minimal automaton has a loop comprising  \emph{all} accepting strongly connected components in the automaton, and the fact that helper states are not needed if there is a corresponding edge in the graph ensures that the automaton has an equivalent deterministic one of a known size if and only if the original graph has a Hamiltonian path (from some fixed predefined vertex). We then show that the construction also works with an alphabet of polynomial size, but there are a substantial number of technical details. In particular, the languages of the strongly connected components need to be altered substantially to ensure that the history-deterministic automaton built from the graph can never be reduced in size, as otherwise the state number of the minimized deterministic co-Büchi automaton would not allow us to read off whether the Hamiltonian path problem has a solution. Reducing the alphabet size comes with an increase in the number of states, but only a polynomial increase if every vertex of the graph only has a number of neighbors that is bounded by a constant, which is why we define a variant on the Hamiltonian Path problem enforcing this constraint in the first place.

NP-hardness of minimizing history-deterministic Büchi automata is then shown to follow from the result for deterministic co-Büchi automata. We employ a recent result by Abu Radi et al.~\cite{ATVA2024} on the relationship of the sizes of history-deterministic co-Büchi and history-deterministic Büchi automata, where the latter is for the complement of the language of the former. We show that for the automata built from graphs for our modified Hamiltonian path problem, the sizes of these two coincide, implying that we can also use a transition-based history-deterministic Büchi automaton minimizer to detect Hamiltonian paths in the original graph. For deterministic Büchi automata, NP-hardness follows by duality from the respective result for deterministic co-Büchi automata.

We do not need to prove that minimizing history-deterministic or deterministic Büchi or parity automata are contained in NP in this paper. While this fact is used to complete the characterization of the complexities in Table~\ref{tab:complexities}, it follows from previous results. In particular, Theorem 4 in a paper by Schewe \cite{DBLP:conf/fsttcs/Schewe20} states that given a history-deterministic parity automaton, checking if another parity automaton is history-deterministic and has the same language is in NP. This allows a nondeterministic algorithm for minimization to simply ``guess'' a transition-based Büchi or parity automaton as solution, to translate it to state-based form (with a blow-up by a factor that corresponds to the number of colors in the automaton, which is bounded by the number of transitions), and to verify the correctness of the resulting automaton. For transition-based deterministic co-Büchi, Büchi, or parity automata, the same approach can be used for showing the containment of the respective minimization problem in NP. In this case, the resulting automaton needs to be syntactically restricted to be deterministic. As an alternative results for this case, a result by Baarir and Lutz~\cite{DBLP:conf/lpar/BaarirD15} can be employed. They showed how to reduce such a minimization problem to the satisfiability problem in a way that it can be solved with modern satisfiability solver tools.

\section{The HD-negativity Problem}
\label{exp red sec}

We start by defining a variant of the \emph{Hamiltonian path problem} that is NP-hard to solve despite restricting the set of graphs over which the problem is defined.
\begin{definition}
\label{def:GraphProperties}

We consider graphs $G = \langle V, E \rangle$ with a designated vertex $v_1$ satisfying the following additional constraints:
\begin{itemize}
\item The graph $G$ is \emph{$(\leq 3)$-regular}, i.e, each vertex in $V$ has at most 3 neighbors. 

\item The vertex $v_1$ has a single neighbor in $G$, and no incoming edges.

\item The sub-graph $G|_{\neg v_1}$ obtained by removing $v_1$ from $G$ and all transitions that leave it, is strongly-connected.

\end{itemize}
\end{definition}
We consider a slightly modified Hamiltonian path problem in the following, where given a graph $G$ with a designated vertex $v_1$,  we need to decide whether there is a simple path starting at $v_1$ and visiting every vertex of $G$.
The following theorem is proven in Appendix~\ref{HP is hard thm app}:
\begin{theorem}\label{HP is hard thm}
    The Hamiltonian path problem is NP-hard for graphs satisfying Def.~\ref{def:GraphProperties}.
\end{theorem}

We  show next that the HD-negativity problem is NP-hard.
Specifically, given a graph $G$ satisfying Def.~\ref{def:GraphProperties} with a designated vertex $v_1$, the following reduction returns a minimal HD-tNCW $\A^{G^{v_1}} = \zug{\Sigma, Q, q_0, \delta, \alpha}$ that has an equivalent tDCW of the same size if and only if $G$ has a Hamiltonian path starting at $v_1$.
The reduction however is such that $\A^{G^{v_1}}$ is defined over an alphabet $\Sigma$ that is exponential in $G$. In Section~\ref{hd neg sec}, we show how we can reduce the size of the alphabet to obtain a polynomial-time reduction, based on the properties of $G$ and $\A^{G^{v_1}}$.

\subsection{The Main Construction}

Consider an alphabet $\Sigma$, a special letter $a\notin \Sigma$, an infinite word  $w$ over $\Sigma$, and an infinite word $z$ over $\Sigma \cup \{ a\}$. We define the (possibly finite) word $z|_{\neg a}$ as the word over $\Sigma$ that is obtained from $z$  by removing  all occurrences of the letter $a$.
We say that the infinite word $w' \in (\Sigma \cup \{a\})^\omega$ is a \emph{mix} of the infinite word $w$ and the infinite word $a^\omega$ if either $(w')|_{\neg a} = w$, or $(w')|_{\neg a}$ is a prefix of $w$.
Thus, $w'$ is a mix of $w$ and $a^\omega$ when  $w'$ is obtained by injecting an arbitrary (possibly infinite) number of $a$'s into $w$ or into a prefix of $w$.

We now define the exponential reduction for the HD-negativity problem. The reduction is defined from the Hamiltonian path problem for graphs satisfying Def.~\ref{def:GraphProperties}. %
The reduction, given an input graph $G = \zug{V, E}$ and a designated vertex $v_1\in V$, returns a tNCW $\A^{G^{v_1}} = \zug{\Sigma, Q, q_0, \delta, \alpha}$ over the exponential alphabet $\Sigma= 2^{\mathsf{AP} } \cup \{\mathit{sleep}\}$, 
where $\mathsf{ AP}$ is the set of atomic propositions given by $\mathsf{AP} = V \uplus \{\hat v, \check v \mid v \in V\} \uplus \{ \mathit{extra} \}$. Thus, we have a special atomic proposition $\text{\it extra}$, and
for every vertex $v\in V$, we have three corresponding atomic propositions, $v, \hat v$, and $\check v$, to be explained below. Thus, a letter $\sigma \in \Sigma$, either equals the special letter $\text{\it sleep}$, or a subset of atomic propositions in $2^{\mathsf{AP}}$.
When $\sigma \in 2^{\mathsf{AP}}$, we also view it as an {\em assignment over} $\mathsf{AP}$ which is a function from $\mathsf{AP}$ to $\{ 0, 1\}$ where for all atomic propositions $a\in \mathsf{AP}$, we have that  $\sigma(a) = 1$ if and only if $a\in \sigma$. We say that $\sigma$ {\emph satisfies} $a$ when $\sigma(a) = 1$. In addition, we to refer $\sigma(a)$ as  the value of $a$ under the assignment $\sigma$.

We define next the rest of $\A^{G^{v_1}}$'s elements by first  defining its safe components, and then add rejecting transitions among them. For every vertex $v\in V$, we introduce a safe component $\SC(v)$ that traps infinite safe runs on words $w$ 
that are a mix of the word $(\mathit{sleep})^\omega$ and some word in
\begin{align*}
L_v = & \{\sigma_1 \cdot \sigma_2 \cdots \in (2^{\mathsf{AP}})^\omega \mid \text{for all $i\geq 1, u\in V$, it holds that } \\
& \quad\quad \quad\quad \quad\quad [E(v, u) \rightarrow (\sigma_i(u) = \sigma_{i+1}(\hat{u}))] \wedge [\sigma_i(v) = \sigma_{i+2}(\check v)]\}
\end{align*}

Thus, $L_v$ specifies consistent values in the neighborhood $\eta(v)$ across every three consecutive  letters, specifically, for all positions $i\geq 1$ and all neighbors $u$ of $v$, the assignment $\sigma_i$ satisfies $u$ if and only if the assignment $\sigma_{i+1}$ satisfies  $\hat{ u}$. In addition, the assignment $\sigma_{i+2}$ satisfies $\check v$ if and only if the assignment $\sigma_i$ satisfies $v$.
Note that every finite word $\sigma_1\cdot \sigma_2 \cdots \sigma_i \in (2^{\mathsf{AP}})^*$ that does not violate the consistency in the neighborhood $\eta(v)$ across every three letters can be extended to an infinite word $\sigma_1\cdot \sigma_2 \cdot \sigma_3 \cdots $ in $L_v$. Indeed, for all $j\geq i+1$, we can define $\sigma_j = \mathsf{AP} \cup \{ \hat v': v' \in \sigma_{j-1}\cap V\} \cup \{ \check v: v\in \sigma_{j-2}\}$.
In a way, the safe component $\SC(v)$ needs to remember the past-one value of every neighbor of $v$, and the past-two value of $v$.
Formally, we define $\SC(v)$ as follows. Let $\eta(v) = \{ v, u_1, u_2, \ldots, u_m \}$ denote the neighborhood of $v$. Then, $\SC(v)$ is defined on top of a state-space consisting  of vectors of length $m+3$ of the form $[p_1(u_1), p_1(u_2), \ldots, p_1(u_m), p_1(v), p_2(v), v]$, where for every vertex $v'$ in $\eta(v)$, $p_1(v')$ is a bit in $ \{0, 1\}$ that represents the value of the vertex $v'$  under the assignment by the last read letter. Similarly, $p_2(v)$ represents the value of $v$ assigned by the past-two letter, and the last  element in the vector contains $v$ to distinguish between vectors corresponding to distinct safe components.

In Appendix~\ref{exp red defs app}, we complete the definition of the safe component 
$\SC(v)$, and define safe transitions deterministically inside it in a way that allows us to trap infinite words $y$ from states $s$ in $\SC(v)$ when $y$ is a suffix of a word $w = x\cdot y$ that is a mix of $(\it sleep)^\omega$ and a word in $L_v$, $x$ contains at least two non-$\it sleep$ letters, and $s$ remembers the correct past values of vertices in $\eta(v)$ with respect to the prefix $x$ of $w$.
\bigcomment{
In Appendix~\ref{exp red defs app}, we complete the definition of the safe component 
$\SC(v)$, and define safe transitions deterministically inside it in a way that allows us to trap infinite words $y$ from states $s$ in $\SC(v)$ when $y$ is a suffix of a word $w = x\cdot y$ that is a mix of $(\it sleep)^\omega$ and \color{blue} $y$ is \color{red} a word in $L_v$, $x$ contains at least two non-$\it sleep$ letters, and $s$ remembers the correct past values of vertices in $\eta(v)$ with respect to the prefix $x$ of $w$.
}%
We also note that to implement the mixing with the infinite word $(\text{\it sleep})^\omega$, it suffices to add safe self-loops labeled with the special letter $\text{\it sleep}$ for every state in $\SC(v)$.

Consider an infinite word $w\in \Sigma^\omega$. We say that $w$ is \emph{extra-consistent} when for every infix $a_i \cdot a_{i+1}$ of $w|_{\neg \it sleep}$ consisting of two letters, 
it holds that $a_i(v_1) = a_{i+1}(\it extra)$.
Now in addition to the safe components $\{ \SC(v)\}_{v\in V}$, we define an additional synchronizing safe component $\SC_{\it Sync}$ to trap infinite safe runs on  words $w$ in 
\begin{align*}
L_\mathit{Sync} & =\{ w \in \Sigma^\omega \mid \text{$w$ is $\text{\it extra}$-consistent, and there exists }
b_1\cdot  b_2 \cdot b_3 \cdots \in ((\epsilon+0+00 \\&  \ \ \ \ \ \ +000)\cdot 1)^\omega \text{ such that }   \text{for all } i \geq 1, \text{it holds that } b_i \leftrightarrow (w[i] = \mathit{sleep}) \}
\end{align*}

Thus, $L_{\it Sync}$ specifies that the special letter $\text{\it sleep}$ appears every at most 4 letters; in addition, it specifies that the value of the special atomic proposition $\text{\it extra}$ always equals the previous past-one value of the vertex $v_1$. 
The safe component $\SC_{\it Sync}$ can be easily defined over the state-space $\{1, 2, 3, 4\}\times \{ 0, 1\}$, where the state $(t, k)$ in $\SC_{\it Sync}$ remembers that we have seen $\text{\it sleep}$ in the past-$t$ letter, and $k$ remembers the last value of $v_1$.
In Appendix~\ref{exp red defs app}, 
we complete the definition of the safe component 
$\SC_{\it Sync}$, and define safe transitions deterministically inside it 
in a way that allows us to trap infinite words $y$ from  states of the form $(1, k)$ in $\SC_{\it Sync}$ when $y$ is a suffix of a word $w = x\cdot y$ in $L_{\it Sync}$, $x$ contains at least one non-$\it sleep$ letter, and $(1, k)$ remembers that we have just seen $\it sleep$ and remembers the correct last value $k$ of $v_1$ as specified by the prefix $x$ of $w$. In addition, as  we detail in Appendix~\ref{exp red defs app},  every state in  $\SC_{\it Sync}$ has a safe outgoing $\it sleep$-transition.

We now end the definition of the tNCW $\A^{G^{v_1}}$ by choosing its initial state $q_0$ arbitrarily, and by defining rejecting transitions as follows.
Consider a state $q\in Q$, and a letter $\sigma \in \Sigma$. If there are no safe $\sigma$-labeled transitions going out from $q$, then we add the rejecting transitions $\{ q\} \times \{ \sigma\} \times Q$. Thus, if from $q$, we cannot traverse a safe $\sigma$-transition, then we branch nondeterministically to all states via rejecting transitions.
Essentially, the tNCW $\A^{G^{v_1}}$ consists of states specifying that the input word has a suffix in $ L_{\it Sync}$ or a suffix that is a mix of $(\it sleep)^\omega$ and some word in $(\bigcup_{v\in V} L_v)$, and they all have the same language. 
It is not hard to see that by definition $\A^{G^{v_1}}$ is already $\alpha$-homogenuous. In particular, since every state has an outgoing safe $\it sleep$-transition, $\A^{G^{v_1}}$ has no rejecting $\it sleep$-transition. 

In Appendix~\ref{useful properties app}, we prove that $\A^{G^{v_1}}$ is a minimal HD-tNCW. Intuitively, HDness follows from the fact that we branch to all states upon leaving a safe component, allowing an HD strategy to proceed to a safe component that traps the suffix to be read. The fact that $\A^{G^{v_1}}$ is minimal follows from the fact that all states have incomparable safe languages, except for states in $\SC_{\it Sync}$, in which case we show they must have distinct safe languages.
In addition, in Appendix~\ref{useful properties app}, we show that $\A^{G^{v_1}}$ satisfies some syntactical and semantical properties, allowing us to conclude the correctness of the reduction. 
Specifically, we have the following:

\begin{theorem}\label{corr thm}
    The graph $G = \zug{V, E}$ has a Hamiltonian path starting at $v_1$ if and only if the minimal HD-tNCW $\A^{G^{v_1}} = \zug{\Sigma, Q, q_0, \delta, \alpha}$ has an equivalent tDCW of size $|Q|$.
\end{theorem}

One direction of Theorem~\ref{corr thm} is easy to obtain. Specifically, it is not hard to show that a Hamiltonian path $p = v_1, v_2, \ldots, v_n$ starting at $v_1$ in $G$ induces a deterministic pruning $\A_d$ of $\A^{G^{v_1}}$ that recognizes the same language. 
Indeed,
as $\A^{G^{v_1}}$ is $\alpha$-homogenuous and safe transitions inside its safe components are defined deterministically, we have that $\A^{G^{v_1}}$ has nondeterministic choices only among its rejecting transitions. 
Then, to obtain $\A_d$, we essentially do the following.
For a state $q\in \SC(v_i)$ that has no outgoing safe $\sigma$-transitions for some letter $\sigma\neq \it sleep$, we keep only a single rejecting transition $\zug{q, \sigma, s}$ from $q$ where $s\in \SC(v_{i+1})$ is such that $s$ encodes the  past-two value of $v_{i+1}$ as specified by the past-one value of $v_{i+1}$ that is encoded in $q$. Indeed, as $E(v_i, v_{i+1})$, the previous past-one value of $v_{i+1}$ is already encoded in the state $q$. 
Also, as expected, the past-one values encoded in $s$ are induced by the  assignment $\sigma$. Note that by definition, this rejecting transition exists since we branch to all states upon leaving a safe component.
Similarly, when $q$ belongs to $\SC_{\it Sync}$, then we  prune rejecting transitions from $\SC_{\it Sync}$ except for those that lead to $\SC(v_1)$ with the past-two value of $v_1$ as encoded in $\SC_{\it Sync}$. Indeed, $\SC_{\it Sync}$ encodes the last value of $v_1$, and in fact this is why the $\it extra$-consistency property is specified. Finally, for $\SC(v_n)$, we only keep the transitions to $(1, k)$ in $\SC_{\it Sync}$ remembering that we have just seen $\it sleep$, and the past-one value of $v_1$ as specified by the last letter read.
To sum up, since $\A_d$ always proceeds to the states storing the correct past values,  and since it follows the Hamiltonian path $p$ closing a simple cycle that visits all the safe components of $\A^{G^{v_1}}$, we get that it is well-defined and recognizes the correct language. Indeed, by following $p$, all runs of the resulting automaton cycle through the safe components until reaching a safe component that (safely) accepts the word, provided that the word is actually in the language of the automaton.
For a formal proof, we refer the reader to Section~\ref{min is hard sec}, where we actually prove a stronger claim as needed for showing the NP-hardness of the automaton minimization problems considered in this paper. In particular, we show that for every  (not necessarily simple)  path $p$ that starts at $v_1$ and visits every vertex of $G$, we can build a tDCW $\A^{G^{v_1}_p}$ for $L(\A^{G^{v_1}})$ by taking a copy of $\SC(v)$ for every occurrence of $v$ in $p$, and only one copy of $\SC_{\it Sync}$ (see Proposition~\ref{tDCW is equiv pro}).
Thus, we can conclude with the following:

 \begin{proposition}[Direction 1 of Theorem~\ref{corr thm}]
    If the graph $G = \zug{V, E}$ has a Hamiltonian path starting at $v_1$, then the minimal HD-tNCW $\A^{G^{v_1}} = \zug{\Sigma, Q, q_0, \delta, \alpha}$ has an equivalent tDCW of size $|Q|$.
\end{proposition}

We now proceed with the other direction. 
This direction is more involved and uses key properties of  safe component $\SC_{\it Sync}$ and the way $\it sleep$ transitions are defined.
Assume that there is a  tDCW $\A_d$   for $L(\A^{G^{v_1}})$ of size $|Q|$. We need to show that $\A_d$ induces a Hamiltonian path starting at $v_1$ in $G$.
In Appendix~\ref{corr thm app}, we show that since both $\A^{G^{v_1}}$ and $\A_d$ are minimal HD-tNCWs, then we can assume that $\A_d$ is a deterministic pruning of $\A^{G^{v_1}}$. In particular, 
$\A_d$ and $\A^{G^{v_1}}$ have isomorphic safe components. In addition, all the states in $\A_d$ are equivalent to $\A^{G^{v_1}}$.

Recall that $G = \zug{V, E}$, and let $n = |V|$.  
We prove that $\A_d$ induces a Hamiltonian path by showing the existence of a sequence of finite $n$ words $z_1, z_2, \ldots, z_n$ over $2^{\mathsf{AP}}$ such that the run of $(\A_d)^{(1, 1)}$ on $z = z_1\cdot z_2 \cdots z_n$ is of the form $ r_z = (1, 1) \xrightarrow{z_1} s_1\xrightarrow{z_2} s_2 \cdots \xrightarrow{z_n} s_n$, 
where for all $i \in [n]$, it holds that the state $s_i$ belongs to a safe component $\SC(v_i)$ for some vertex $v_i\in V$, the run $r_z$ leaves the current safe component and enters the next safe component $\SC(v_{i})$ after reading the last letter of $z_{i}$, and  $(v_1, v_2, \ldots, v_n)$ is a Hamiltonian path in $G$. 
Thus,  the run $r_z$ starts at the state $(1, 1)$ in $\SC_{Sync}$, and switches between safe components by simulating a Hamiltonian path in $G$ starting at $v_1$.

In Appendix~\ref{corr thm app}, we
 define the word $z = z_1\cdot z_2 \cdot z_3 \cdots z_n$ iteratively. Essentially, the idea is that for the tDCW $\A_d$ to recognize the correct language, its rejecting transitions must lead to states storing the correct past values with the respect to the prefix read so far, and that happens only when rejecting transitions in $\A_d$ follow edges of the input graph $G$. 
Essentially, a transition that leaves a safe component $\SC(v_i)$ should move to a safe component corresponding to a neighbor of $v_i$, as these are the only vertices that $\SC(v_i)$ keeps track of.
A property of the word $z$ is that every prefix 
$z_1\cdot z_2\cdots z_i$ can  be extended to an infinite word in $L_u$ for all vertices $u$ whose corresponding safe components were not yet visited by the path built so far. The latter allows us prove that the path we gradually build is Hamiltonian. %
The formal argument uses the fact that the safe component $\SC_{\it Sync}$ is designed in a way that forces it to be visited by every non-safe cycle in $\A_d$ that is labeled with a finite word whose infinite repetition is $\it extra$-consistent (in fact, for this reason the mixing with $(\it sleep)^\omega$ is crucial in the specified language). In particular, if the form of the run $r_z$ was not as stated above, 
we could close a non-safe cycle starting at $\SC_{\it Sync}$ (for some finite word $x$ along the cycle) before some state of $\mathcal{S}(u)$ is reached for some vertex $u$ while having that $x^\omega$ is a mix of $(\it sleep)^\omega$ and some word in $ \in L_u$; in particular $x^\omega \in L(\A^{G^{v_1}})$, 
contradicting the fact that it is rejected by $\A_d$. 

 Thus, we can conclude with the following:

 \begin{proposition}[Direction 2 of Theorem~\ref{corr thm}]
    If the minimal HD-tNCW $\A^{G^{v_1}} = \zug{\Sigma, Q, q_0, \delta, \alpha}$ has an equivalent tDCW of size $|Q|$, then 
    the graph $G = \zug{V, E}$ has a Hamiltonian path starting at $v_1$. 
\end{proposition}

Note that since the graph $G$ is $(\leq 3)$-regular, we get that $\A^{G^{v_1}}$
has a size that is polynomial in the size of $G$.
We end this section by showing that the above exponential reduction implies
NP-hardness of the HD-negativity problem
with a modification of the construction so that it uses an alphabet of polynomial %
size. In this way, storing $\A^{G^{v_1}}$ requires space polynomial in the size of
$G$.

\subsection{Getting Rid of The Exponential Alphabet}

In Appendix~\ref{hd neg sec}, 
we show how to encode $\A^{G^{v_1}}$ into a  minimal HD-tNCW $e(\A^{G^{v_1}})$
whose description is at most polynomial in the input graph and that can be
computed in polynomial time from $G$ when applying it ``on-the-fly'' without ever
explicitly building $\A^{G^{v_1}}$.

The basic idea behind the modification of the reduction is to encode  letters in
$2^{\mathsf{AP}}$ as binary numbers.  Specifically, let $k = |\mathsf{AP}|$ denote
the number of the atomic propositions in $\mathsf{AP} = \{ a_1, a_2, \ldots,
a_k\}$. Then, we encode each letter $\sigma$ in $2^{\mathsf{AP}}$ as a binary
number  $\zug{\sigma}$ of length $k$, where the $i$'th bit of $\zug{\sigma}$ is 1
if and only if $a_i\in \sigma$.
The interesting part is that the encoding $e(\A^{G^{v_1}})$ of $\A^{G^{v_1}}$ can
simulate transitions of $\A^{G^{v_1}}$ by following paths of length $k$ labeled
with encoded letters using only polynomially many  new internal states.
The key insight is
that proposition values $a_i$ corresponding to vertices in $G$ that are
at least 3 edges away from some vertex $v$ can be ignored in the part of
$e(\A^{G^{v_1}})$ corresponding to $v$. Since $G$ is $(\leq 3)$-regular, this
means that the number of proposition values of relevance is constant, enabling the encoding to only have a polynomial blow-up.

Then, Appendix~\ref{hd neg sec} not only shows that we can construct the encoding
$e(\A^{G^{v_1}})$ in polynomial-time, but also proves the following Proposition:

\begin{proposition}\label{reduc applies prop}
     There is a Hamiltonian path in $G$ starting at $v_1$ iff the minimal HD-tNCW $e(\A^{G^{v_1}})$
has an equivalent tDCW of the same size.
\end{proposition}

Thus, we can conclude with the following:

\begin{theorem}
    The HD-negativity problem is NP-hard.
\end{theorem}

\section{Minimization of Transition-Based Automata is Computationally Hard}
\label{min is hard sec}
In this section, we show that the {\em $\gamma$-minimization problem}, for $\gamma\in \{\text{tDBW}, \allowbreak{} \text{tDCW}, \allowbreak{} \text{HD-tNBW}\}$, is NP-complete. 
The upper bound is quite obvious as language containment between HD automata can be checked in polynomial time for Büchi and co-Büchi acceptance conditions~\cite{DBLP:journals/iandc/HenzingerKR02,DBLP:conf/icalp/KuperbergS15}, and checking whether a given tNBW is HD can be done in polynomial time, too~\cite{DBLP:conf/fsttcs/BagnolK18}. Specifically, to show that the $\gamma$-minimization problem is in NP, we can, given a %
$\gamma$ automaton $\A$ and a bound $k\geq 1$, first guess a $\gamma$ automaton $\B$ of size at most $k$, and then check whether it is equivalent to $\A$. 
Note that when $\gamma = \text{HD-tNBW}$, then guessing $\B$ amounts to guessing a tNBW and then checking whether it is HD.
\bigcomment{%
The upper bound is quite obvious as language containment between HD automata can be checked in polynomial time~\cite{DBLP:journals/iandc/HenzingerKR02,DBLP:conf/icalp/KuperbergS15}, and checking whether a given Büchi or co-Büchi automaton is HD can be done in polynomial time, too~\cite{DBLP:conf/icalp/KuperbergS15,DBLP:conf/fsttcs/BagnolK18}. Specifically, to show that the $\gamma$-minimization problem is in NP, we can, given an HD $\gamma$ automaton $\A$ and a bound $k\geq 1$, first guess a $\gamma$ automaton $\B$ of size at most $k$, and then check whether it is HD and equivalent to $\A$. 
}%
Thus, we can conclude with the following:

\begin{proposition}\label{np prop}
    Consider $\gamma \in \{ \text{tDBW, tDCW, HD-tNBW}\}$. Then, the $\gamma$-minimization problem is in NP.
\end{proposition}

For a tDBW $\A$, we use $\widetilde{\A}$ to denote the tDCW obtained by viewing $\A$ as a co-\buchi automaton. The \buchi and
co-\buchi acceptance conditions are dual in the sense that $L(\A) = \overline{L(\widetilde{\A})}$. Similarly, for
a tDCW $\A$, we let $\widetilde{\A}$ denote the dual deterministic \buchi automaton, which complements $\A$.

We proceed with showing that minimization of tDBWs and HD-tNBWs (and thus by duality also minimization of tDCWs) is NP-hard by reducing from instances of the HD-negativity problem. Consider a graph $G = \zug{V, E}$ with a designated vertex $v_1$ satisfying Def.~\ref{def:GraphProperties}, 
consider the HD-tNCW $\A^{G^{v_1}}$ defined in Section~\ref{exp red sec}, and recall that it has no rejecting $\it sleep$ transitions and
it has a safe component $\SC(v)$ corresponding to every vertex $v\in V$; in addition, it has a safe component $\SC_{\it Sync}$ corresponding to the language $L_{\it Sync}$. %
We show next that we can always construct a tDCW with polynomialy many states for $L(\A^{G^{v_1}})$ regardless of whether the minimal HD-tNCW $\A^{G^{v_1}}$ has an equivalent tDCW of the same size. %
Recall that the sub-graph $G|_{\neg v_1}$ is strongly-connected, and let $p = v_1, v_2, v_3, \ldots, v_l$ be a path of polynomial length starting at $v_1$ and  visiting every vertex of $G$. Note that such a path $p$ can be computed in polynomial-time from the input graph $G$.
Then, a tDCW $\A^{G^{v_1}_p}$ for $L(\A^{G^{v_1}})$ that has $l+1$ safe components each of constant size is obtained from $G$ as follows.
We first add a safe component $\SC_{\it Sync} \times \{ l+1\}$ for the safe language $L_{\it Sync}$, which is a copy of the safe component $\SC_{\it Sync}$ with its states labeled with the index $l+1$; that is, $\zug{\zug{q, l+1}, \sigma,\zug{s, l+1}}$ is a safe transition in $\SC_{\it Sync} \times \{ l+1\}$ only when $\zug{q, \sigma, s}$ is a safe transition in $\SC_{\it Sync}$.
Then, for every $i\in [l]$, we define a unique safe component $\SC(v_i) \times \{i\}$ which is copy of the safe component $\SC(v_i)$. 
Thus, $\SC(v_i) \times \{i\}$ is just a copy of the safe component $\SC(v_i)$ with its states labeled with the index $i$. 
We define next rejecting transitions among safe components. Recall that the safe component $\SC(v_i)$ traps infinite words $w$ that are a mix of $(\it sleep)^\omega$ and some word in $L_{v_i}$. Specifically, $\SC(v_i)$ has two memory cells to keep track of the past-two value of $v_i$ and one memory cell for the past-one value of every other vertex in $\eta(v_i)$.
We add  rejecting transitions by letting $\A^{G^{v_1}_p}$ follow the path $p$ in $G$.  
Formally, consider a  letter $\sigma \neq \it sleep$, and a state $\zug{q, i}$ in $\SC(v_i)\times \{i\}$ that has no outgoing safe $\sigma$-labeled transition. We distinguish between two cases. 
First, if $i < l$, then we  add a deterministic rejecting transition $\zug{\zug{q, i}, \sigma, \zug{q', i+1}}$ where $q'$ corresponds to memory cells encoding the past-two value of $v_{i+1}$  as specified by the past-one value of $v_{i+1}$ in $q$, and encoding the past-one values of every vertex in $\eta(v_{i+1})$ as specified by the letter $\sigma$. 
Note that the state $q'$ can be determined  based only on the letter $\sigma$ and the state $q$ as the vertex $v_{i+1}$ is a neighbor of  $v_i$ and so its previous past-one value is encoded already in the state $q$.  We proceed to the other case where $i = l$.
In this case, we add a deterministic rejecting transition $\zug{ \zug{q, i}, \sigma, \zug{q', l+1}}$ where $q'$ is a state in $\SC_{\it Sync}$ that remembers that we have just seen $\it sleep$ and remembers the value of the vertex $v_1$ as specified by the letter $\sigma$. Note that $q'$ exists and can be determined only based on $\sigma$. Indeed, the safe component $\SC_{\it Sync}$ keeps track of the last value of $v_1$ and whether $\it sleep$ has appeared in the last four letters. %
It is left to define $\sigma$-labeled rejecting transitions from a state $\zug{q, l+1}$ in $\SC_{\it Sync} \times \{ l+1\}$, and as expected, 
we move to the state $\zug{q', 1}$ in $\SC(v_1)\times \{1\}$ encoding the past-two value of $v_1$  as specified by that past-one value of $v_1$ in $q$, and encoding the past-one values of every vertex in $\eta(v_1)$ as specified by $\sigma$.
Finally, the initial state of $\A^{G^{v_1}_p}$ is chosen arbitrarily.

Thus, the tDCW $\A^{G^{v_1}_p}$ is obtained from $\A^{G^{v_1}}$ by duplicating its safe components in a way that simulates the polynomial path $p$.
In addition, as the graph $G$ is $(\leq 3)$-regular, then each safe component is of constant size, and so the state-space of the automaton is at most polynomial in $G$.

The following lemma is immediate from the definitions and suggests that a run in $\A^{G^{v_1}_p}$ on a word having at least two non-$\it sleep$ transitions eventually always enters safe components via states encoding the correct past values: 

\begin{lemma}\label{tracks lemm}
     Consider a run $r$ in $\A^{G^{v_1}_p}$. After traversing %
     two non-$\it sleep$ transitions, the run $r$ always
     enters a safe component by moving to the 
     state  encoding the correct past values with respect to the last read letters.
\end{lemma}

\begin{proof}
    The definition of the transition function is such that when we read a non-$\it sleep$ letter $\sigma$ from any state, we move to a successor-state whose memory cells correspond to the correct past-one values as specified by $\sigma$. Therefore, after reading the first non-$\it sleep$ letter of the input word, we can have inconsistency only with the past-two values which will be resolved after reading the second non-$\it sleep$ input letter. Indeed, memories corresponding to the past-two values get updated according to the previous past-one values.
\end{proof}
The following proposition suggests that $\A^{G^{v_1}_p}$ recognizes the language $L(\A^{G^{v_1}})$.

\begin{proposition}\label{tDCW is equiv pro}
    Consider a graph $G = \zug{V, E}$ with a designated vertex $v_1$ satisfying Def.~\ref{def:GraphProperties},
    and consider a polynomial path $p = v_1, v_2, \ldots, v_l$ visiting every vertex of $G$. Then, the tDCW $\A^{G^{v_1}_p}$ recognizes the language $L(\A^{G^{v_1}})$.
\end{proposition}

\begin{proof}
We first show that $L(\A^{G^{v_1}}) \subseteq L(\A^{G^{v_1}_p})$.
Consider a word $w = \sigma_1\cdot \sigma_2 \cdots \in L(\A^{G^{v_1}})$, and consider the run $r = r_0, r_1, r_2, \ldots$ of the tDCW $\A^{G^{v_1}_p}$ on $w$. We need to show that $r$ eventually becomes safe, and thus $w\in L(\A^{G^{v_1}_p})$. 
If $w$ has a suffix in $(\it sleep)^\omega$, then we are done as every state has a safe run on $(\it sleep)^\omega$; indeed, the safe components of $\A^{G^{v_1}_p}$ are inherited from those of $\A^{G^{v_1}}$. 
We proceed to the case where $w$ has infinitely many non-$\it sleep$ letters.
By definition, the rejecting transitions of $\A^{G^{v_1}_p}$ induce a simple cycle $c = \SC(v_1)\times\{1\}, \SC(v_2)\times \{ 2\}, \ldots, \SC(v_l)\times \{l\}, \SC_{\it Sync} \times \{l+1\}$ that visits its safe components by following the path $p$ and then enters the safe component $\SC_{\it Sync} \times \{l+1\}$.  Therefore, upon traversing a rejecting transition, the run $r$ of $\A^{G^{v_1}_p}$ on $w$ moves to the next safe component as specified by the cycle $c$. 
We show that the run $r$  eventually gets stuck in some safe component and thus is accepting.
By the above, it is sufficient to show that there exists some $j\geq 0$ and a safe component $S\in \SC(\A^{G^{v_1}_p})$ such that if the run $r[j, \infty] = r_j, r_{j+1}, r_{j+2}$ on the suffix $w[j+1, \infty]$ %
enters $S$, it stays in there forever.
Recall that $w\in L(\A^{G^{v_1}})$; in particular, one of the following holds:

\begin{itemize}
    \item 
    There are $i\geq 1$ and $v\in V$ such that the suffix $w[i, \infty]$ of $w$ is a mix of $(\it sleep)^\omega$ and some word in $L_v$:  
    let $w[j+1, \infty]$ be a suffix of $w$ for some $j \geq i$ such that the infix $w[i, j]$ contains at least one non-$\it sleep$ letter.
     If the run $r[j, \infty] = r_j, r_{j+1}, \ldots$ on the suffix $w[j+1, \infty]$  enters a safe component $S$ that is a copy of $\SC(v)$ by traversing some rejecting non-$\it sleep$ transition, then, by applying Lemma~\ref{tracks lemm} on the run $r[i-1, \infty]$, we get that $r[j, \infty]$  enters $S$ by moving to a state encoding the correct values of the last two non-$\it sleep$ letters. Therefore, since $w[i, \infty]$ is a mix $(\it sleep)^\omega$ with some word in $L_v$, we get that $r[j, \infty]$ stays in $S$ once it enters it.

    \item There is $j\geq 0$ such that the suffix $w[j+1, \infty]$ of $w$ is in $L_{\it Sync}$: if the run $r[j, \infty] = r_j, r_{j+1}, \ldots$ on the suffix $w[j+1, \infty]$ enters the safe component $S = \SC_{\it Sync} \times \{ l+1\}$, then the definition of  rejecting transitions in $\A^{G^{v_1}_p}$ implies that $r[j, \infty]$  enters $S$ by traversing a non-$\it sleep$ transition leading to a state that remembers that we have just seen $\text{\it sleep}$ and remembers the correct past-one value of $v_1$ as specified by the last non-$\it sleep$ letter. Therefore, since $w[j+1, \infty]\in L_{\it Sync}$, we get that $r[j, \infty]$ stays in $S$ once it enters it.
      
\end{itemize}

The inclusion $  L(\A^{G^{v_1}_p})\subseteq L(\A^{G^{v_1}})$  is straight-forward as any word accepted by $\A^{G^{v_1}_p}$ has a run that eventually gets stuck in some safe component which is a copy of one of the safe components of $\A^{G^{v_1}}$ to begin with.
\end{proof}

Consider a graph $G = \zug{V, E}$ with a designated vertex $v_1$ satisfying Def.~\ref{def:GraphProperties}. 
Note that as the initial state of $\A^{G^{v_1}_p}$ is chosen arbitrarily, the latter proposition in fact implies that all the states in $\A^{G^{v_1}_p}$ are equivalent.
The equivalent automata  $\A^{G^{v_1}}$ and $\A^{G^{v_1}_p}$
are defined over the exponential alphabet $2^{\mathsf{AP}}\cup \{\it{sleep}\}$, and so 
the next step towards showing hardness of the minimization problem is to consider polynomial encodings of them. 
In Appendix~\ref{hd neg sec}, we have encoded the minimal HD-tNCW $\A^{G^{v_1}} = \zug{2^{\mathsf{AP}}\cup \{\text{\it sleep}\}, Q, q_0, \delta, \alpha}$ into a polynomial minimal HD-tNCW $e(\A^{G^{v_1}}) = \zug{\{ 0, 1, \text{ \it sleep}\} \cup \Sigma_\$, Q', q_0, \delta', \alpha' }$.
Essentially, the state-space $Q'$ includes $Q$ and contains polynomially many new internal states, %
every letter $\sigma$ in $2^{\mathsf{AP}}$ is encoded as a binary number $\zug{\sigma}$ of length $k = |\mathsf{AP}|$, and upon reading $\zug{\sigma}$ from a state $s\in Q$, the encoding $e(\A^{G^{v_1}})$ moves deterministically via a safe path to an internal state $s_x$ after reading a proper prefix $x$ of $\zug{\sigma}$.
Then, upon reading a bit $b$ from the state $s_x$, for some $x\in \{ 0, 1\}^{k-1}$; i.e., when reading the last bit $b$ of the encoding of some letter $\zug{\sigma}$, $e(\A^{G^{v_1}}) $ either proceeds via a safe transition to the only $\sigma$-successor $p$ of $s$ when $\zug{s, \sigma, p}$ is a safe transition in $\A^{G^{v_1}}$ , or proceeds nondeterministically to all states $p$ in $Q$ via rejecting transitions otherwise. 
Thus, the encoding $e(\A^{G^{v_1}})$ of $\A^{G^{v_1}}$  simulates non-$\it sleep$ transitions by following paths of length $k$ labeled with encoded letters. 

It is tempting to think that applying the same encoding scheme on the tDCW $\A^{G^{v_1}_p}$ induces equivalent encodings, that is $e(\A^{G^{v_1}}) = e(\A^{G^{v_1}_p})$. However%
, that does not hold since to make sure that $e(\A^{G^{v_1}})$ is a minimal HD-tNCW, we added a special unique letter $\$_d \in \Sigma_\$$  for every state $d$ in $Q'\setminus \{ q_{rej}\}$, where $q_{rej}$ is a rejecting sink in $e(\A^{G^{v_1}})$: for every internal state $s_x$, for some $|x|\leq k-1$, we added
the safe transition $\zug{s_x, \$_{s_x}, s}$, and added the rejecting transitions $\{s_x\} \times \{\$_{\neq s_x}\} \times Q$.
However, 
as the tDCW $\A^{G^{v_1}_p}$ is obtained by duplicating $\A^{G^{v_1}}$'s safe components, it is not surprising that 
the above encoding scheme can be adapted to the tDCW $\A^{G^{v_1}_p}$  by 
duplicating the safe components of $e(\A^{G^{v_1}})$,
and since $\A^{G^{v_1}_p}$ need not be minimal, both encodings can be defined over the same alphabet by
duplicating also special letters.
In Appendix~\ref{adapt enc app}, we adapt the encoding scheme to $\A^{G^{v_1}_p}$ by encoding $\A^{G^{v_1}_p}$ into a tDCW $e(\A^{G^{v_1}_p})$ that is induced naturally from the encoding $e(\A^{G^{v_1}})$ and the polynomial path $p$. Essentially, the encoding $e(\A^{G^{v_1}_p})$ is identical to $e(\A^{G^{v_1}})$ in the way it encodes safe components, the exception is how it handles rejecting transitions among them: as expected, for non-special letters $\sigma \in 2^{\mathsf{AP}}$, $e(\A^{G^{v_1}_p})$ follows (deterministically) the transition function of $\A^{G^{v_1}_p}$, and for rejecting transitions labeled with a special letter $\$_{s_x}$, it always proceeds
to the same state %
in the same safe component to preserve determinism.
In particular, we can conclude with the following, suggesting that the encoding $e(\A^{G^{v_1}_p})$ satisfies the desired properties (see proof in Appendix~\ref{equiv enc lemm app}):

\begin{lemma}\label{equiv enc lemm}
   Consider the automata $\A^{G^{v_1}}$ and  $\A^{G^{v_1}_p}$. Then, their encodings satisfy the following properties:
   \begin{enumerate}
        \item 
        The encoding $e(\A^{G^{v_1}_p})$ is a tDCW.   
        
       \item 
       The encodings are such that $L(e(\A^{G^{v_1}})) = L(e(\A^{G^{v_1}_p}))$.

   \end{enumerate}
   
\end{lemma}

We can now conclude NP-hardness for the minimization problem. For example, to show the hardness of  tDCW minimization, 
we can reduce the Hamiltonian Path problem for a graph
$G = \zug{V, E}$ with a designated vertex $v_1$ and that satisfies Def.~\ref{def:GraphProperties}
to the tDCW minimization problem with the tDCW $e(\A^{G^{v_1}_p})$ and the size bound $|e(\A^{G^{v_1}})|$. Indeed, Lemma~\ref{equiv enc lemm} and Proposition~\ref{reduc applies prop} imply that $e(\A^{G^{v_1}_p})$ has an equivalent tDCW of size at most $|e(\A^{G^{v_1}})|$ if and only if $G$ has a Hamiltonian path starting at $v_1$. 
As HD-tNBWs can be complemented into HD-tNCWs on the same state-space~\cite{ATVA2024}, the following proposition suggests that 
this reduction applies also to the minimization of HD-tNBWs when viewing $e(\A^{G^{v_1}_p})$ as a complemented HD-tNBW:

\begin{proposition}\label{np hard prop}
    Minimization of tDBWs and HD-tNBWs is NP-hard.
\end{proposition}

\begin{proof}
We start by describing a reduction from the Hamiltonian path problem for graphs satisfying Def.~\ref{def:GraphProperties}, 
proven to be NP-hard in Theorem~\ref{HP is hard thm}. 
Given a graph $G= \zug{V, E}$ with a designated vertex $v_1$, the
reduction computes a polynomial path $p$ that starts at $v_1$ and visits every vertex of $G$, computes the size of the minimal HD-tNCW $e(\A^{G^{v_1}})$, and returns the tDBW $\widetilde{e(\A^{G^{v_1}_p})}$ and the integer $|e(\A^{G^{v_1}})|$. 
Note that as tDBWs are a special case of HD-NBWs, then we can view $\widetilde{e(\A^{G^{v_1}_p})}$ as an HD-tNBW. We first prove that the reduction is valid for the minimization of HD-tNBWs.
Specifically, we prove that $G$ has a a Hamiltonian path starting at $v_1$ if and only if $\widetilde{e(\A^{G^{v_1}_p})}$ has an
equivalent HD-tNBW $\B$ of size at most $|e(\A^{G^{v_1}})|$. Then, we argue that the proof applies also for the minimization of tDBWs.

For the first direction, if $G$ has a Hamiltonian path starting at $v_1$, then by Proposition~\ref{reduc applies prop}, there is a tDCW $\A$ for $L(e(\A^{G^{v_1}}))$ of size $|e(\A^{G^{v_1}})|$. By duality and Lemma~\ref{equiv enc lemm}, we get that $\B = \widetilde{\A}$ is a tDBW for $\overline{L(e(\A^{G^{v_1}}))} = \overline{L(e(\A^{G^{v_1}_p}))}  =  L(\widetilde{e(\A^{G^{v_1}_p})})$ whose size is at most $|e(\A^{G^{v_1}})|$.
For the other direction, assume that there is an HD-tNBW $\B$ with  state-space $Q_\B$,  $L(\B) = L(\widetilde{e(\A^{G^{v_1}_p})})$, and $|Q_\B| \leq |e(\A^{G^{v_1}})|$.
By \cite{ATVA2024}, there is an HD-tNBW $\B'$  equivalent to $\B$ with the following properties: (1) $\B'$ can be obtained from $\B$ by pruning transitions, (2) there is a non-empty subset of states $Q_{\it det}\subseteq Q_\B$ such that for every state $q\in Q_{\it det}$ and  letter $a$, it holds that $\delta_{\B'}(q, a)$ is a singleton; In addition, (3) an HD-tNCW $\C$ for $\overline{L(\B)}$ can be defined on top of the state-space $Q_{\it det}$. Thus, $Q_{\it det}$ is a set of deterministic states in $\B'$ whose size bounds the size of a minimal HD-tNCW for $\overline{L(\B)}$ from above. 
By the assumption, Lemma~\ref{equiv enc lemm} and duality, it holds that $L(\C) = \overline{L(\B)} = L(e(\A^{G^{v_1}_p})) = L(e(\A^{G^{v_1}}))$. In addition, $|Q_\B| \leq |e(\A^{G^{v_1}})|$. On the one hand, as $e(\A^{G^{v_1}})$ is a minimal HD-tNCW, we get that $ |e(\A^{G^{v_1}})| \leq |\C|$. On the other hand, as $\C$ is defined on top of $Q_{\it det}$ and $Q_{\it det} \subseteq Q_\B$, we have that $|\C| \leq |\B|$. So we get in total that $|e(\A^{G^{v_1}})| \leq |\C| \leq |\B| \leq |e(\A^{G^{v_1}})|$. Hence, $Q_{\it det} = Q_\B$. In particular, $\B'$ is a tDBW for $L(\widetilde{e(\A^{G^{v_1}_p})})$ of size $|e(\A^{G^{v_1}})|$. By duality and item 2 of Lemma~\ref{equiv enc lemm}, we get that $\widetilde{\B'}$ is a tDCW for $L(e(\A^{G^{v_1}}))$ of size $|e(\A^{G^{v_1}})|$, and so we conclude by Proposition~\ref{reduc applies prop} that $G$ has a Hamiltonian path starting at $v_1$.

To conclude that the reduction is also valid for the minimization problem of tDBWs, 
note that in the first direction of the above proof, the automaton $\B$ is already a tDBW. 
Also, 
a tDBW $\B$ of size at most $|e(\A^{G^{v_1}})|$ that recognizes $L(\widetilde{e(\A^{G^{v_1}_p})})$ is an HD-tNBW, and
so the second direction applies also to the tDBW case.
\end{proof}

By duality of tDCWs and tDBWs, and Propositions~\ref{np prop} and \ref{np hard prop}, we can conclude with the following:

\begin{theorem}
    Consider $\gamma \in \{ \text{tDBW, tDCW, HD-tNBW}\}$. Then, the $\gamma$-minimization problem is NP-complete.
\end{theorem}

\section{Discussion / Conclusion}

In this paper, we showed that minimizing history-deterministic Büchi automata as well as deterministic Büchi and co-Büchi automata is NP-complete when the automata have transition-based acecptance.

While for automata with state-based acceptance, the complexity of automaton minimization with these acceptance condition-types was known, for transition-based acceptance, the complexities of these minimization problems were previously open. 
It should be noted that while the difference between state-based and transition-based acceptance may appear to be minor, the fact that for history-deterministic co-Büchi automata, transition-based acceptance lowers the minimization complexity from NP-complete to polynomial-time shows that distinguishing between state-based and transition-based acceptance is relevant when searching for an efficiently minimizable automaton model.
Furthermore, transition-based acceptance can lower the size of an automaton by a factor of at most 2 in Büchi and co-Büchi automata; in particular, a polynomial-time minimization of deterministic automata with transition-based acceptance implies a factor-2 approximation algorithm with polynomial computation time for minimization of state-based automata.
The NP-hardness proof for state-based deterministic co-Büchi automaton minimization by Schewe~\cite{DBLP:conf/fsttcs/Schewe10} is based on a reduction from the \emph{vertex-cover} problem that has a factor-2 approximation algorithm with polynomial computation time. With both factors being precisely the same, Schewe's result hints that minimization of deterministic transition-based automata might be tractable. 
However, by a careful observation, one obtains that Schewe's reduction outputs co-Büchi automata that recognize languages that 
do not benifit from the transition to the HD model in terms of state-space complexity when using transition-based acceptance, regardless of whether the input graph has a vertex cover of the specified size; in particular, the reduction has no implications for automata with transition-based acceptance.
Taking this one step further, the new results from this paper prove NP-hardness of deterministic co-Büchi automaton minimization for transition-based acceptance, shattering the remaining hope for an efficient automaton minimization procedure in the 
transition-based setting.

The key ingredient to the results obtained in this paper is that we did not look at determinism and history-determinism in isolation, but rather investigated the relationship between the state-space complexity of minimal history-deterministic and deterministic co-Büchi automata.
By exploiting the characterization of canonical minimal history-deterministic co-Büchi automata with transition-based acceptance by Abu Radi and Kupferman~\cite{DBLP:journals/lmcs/RadiK22}, we defined a class of history-deterministic automata that are already provably minimal, encode the graph of a (modified) Hamiltonian path problem, and for which determinizing the automaton without a blow-up in the state-space requires to solve the Hamiltonian path problem. 
In this way,
the study of history-deterministic automata led to understanding the complexity of a problem over deterministic automata.%

Given that polynomial-time minimization of $\omega$-automata is highly useful in \emph{compositional} reasoning over reactive systems %
and the complexity landscape depicted in Table~\ref{tab:complexities} proves that co-Büchi languages are the only ones that have a suitable automaton model, our results suggest that for compositional reasoning over $\omega$-regular languages, one has to look beyond the deterministic, history-deterministic, and non-deterministic automata. 
A candidate in this context are \emph{families of deterministic automata over finite words}~\cite{DBLP:journals/lmcs/AngluinBF18}, at the expense of their languages being defined over ultimately periodic words. Alternatively, $\omega$-regular languages can be decomposed  
into a \emph{chain-of-co-Büchi automaton representation}~\cite{DBLP:conf/fsttcs/EhlersS22}, which employs minimized history-deterministic co-Büchi automata with transition-based acceptance for each chain element, so that they can be minimized with the polynomial-time minimization approach for this automaton class.
We hope that as a consequence of our results showing that for polynomial-time minimization, one has to trade deterministic automata against some other model, interest in using these alternative $\omega$-regular language representations for compositional reasoning increases in the future.

\bibliography{bib}

\begin{thebibliography}{10}

\bibitem{DBLP:journals/lmcs/RadiK22}
Bader {Abu Radi} and Orna Kupferman.
\newblock Minimization and canonization of {GFG} transition-based automata.
\newblock {\em Log. Methods Comput. Sci.}, 18(3), 2022.
\newblock \href {https://doi.org/10.46298/LMCS-18(3:16)2022}
  {\path{doi:10.46298/LMCS-18(3:16)2022}}.

\bibitem{DBLP:conf/icalp/RadiK23}
Bader {Abu Radi} and Orna Kupferman.
\newblock On semantically-deterministic automata.
\newblock In Kousha Etessami, Uriel Feige, and Gabriele Puppis, editors, {\em
  50th International Colloquium on Automata, Languages, and Programming
  ({ICALP})}, volume 261 of {\em LIPIcs}, pages 109:1--109:20. Schloss Dagstuhl
  - Leibniz-Zentrum f{\"{u}}r Informatik, 2023.
\newblock \href {https://doi.org/10.4230/LIPICS.ICALP.2023.109}
  {\path{doi:10.4230/LIPICS.ICALP.2023.109}}.

\bibitem{ATVA2024}
Bader {Abu Radi}, Orna Kupferman, and Ofer Leshkowitz.
\newblock Easy complementation of history-deterministic b{\"{u}}chi automata.
\newblock In S.~Akshay, Aina Niemetz, and Sriram Sankaranarayanan, editors,
  {\em 22nd International Symposium on Automated Technology for Verification
  and Analysis ({ATVA})}, volume 15054 of {\em Lecture Notes in Computer
  Science}, pages 67--88. Springer, 2024.
\newblock \href {https://doi.org/10.1007/978-3-031-78709-6\_4}
  {\path{doi:10.1007/978-3-031-78709-6\_4}}.

\bibitem{DBLP:journals/lmcs/AngluinBF18}
Dana Angluin, Udi Boker, and Dana Fisman.
\newblock Families of {DFAs} as acceptors of {\(\omega\)}-regular languages.
\newblock {\em Log. Methods Comput. Sci.}, 14(1), 2018.
\newblock \href {https://doi.org/10.23638/LMCS-14(1:15)2018}
  {\path{doi:10.23638/LMCS-14(1:15)2018}}.

\bibitem{DBLP:conf/lpar/BaarirD15}
Souheib Baarir and Alexandre Duret{-}Lutz.
\newblock Sat-based minimization of deterministic {$\omega$}-automata.
\newblock In Martin Davis, Ansgar Fehnker, Annabelle McIver, and Andrei
  Voronkov, editors, {\em 20th International Conference on Logic for
  Programming, Artificial Intelligence, and Reasoning ({LPAR-20})}, volume 9450
  of {\em Lecture Notes in Computer Science}, pages 79--87. Springer, 2015.
\newblock \href {https://doi.org/10.1007/978-3-662-48899-7\_6}
  {\path{doi:10.1007/978-3-662-48899-7\_6}}.

\bibitem{DBLP:conf/fsttcs/BagnolK18}
Marc Bagnol and Denis Kuperberg.
\newblock B{\"{u}}chi good-for-games automata are efficiently recognizable.
\newblock In Sumit Ganguly and Paritosh~K. Pandya, editors, {\em 38th {IARCS}
  Annual Conference on Foundations of Software Technology and Theoretical
  Computer Science ({FSTTCS})}, volume 122 of {\em LIPIcs}, pages 16:1--16:14.
  Schloss Dagstuhl - Leibniz-Zentrum f{\"{u}}r Informatik, 2018.
\newblock \href {https://doi.org/10.4230/LIPICS.FSTTCS.2018.16}
  {\path{doi:10.4230/LIPICS.FSTTCS.2018.16}}.

\bibitem{DBLP:conf/csl/Casares22}
Antonio Casares.
\newblock On the minimisation of transition-based {R}abin automata and the
  chromatic memory requirements of {M}uller conditions.
\newblock In Florin Manea and Alex Simpson, editors, {\em 30th {EACSL} Annual
  Conference on Computer Science Logic ({CSL})}, volume 216 of {\em LIPIcs},
  pages 12:1--12:17. Schloss Dagstuhl - Leibniz-Zentrum f{\"{u}}r Informatik,
  2022.
\newblock \href {https://doi.org/10.4230/LIPICS.CSL.2022.12}
  {\path{doi:10.4230/LIPICS.CSL.2022.12}}.

\bibitem{DBLP:conf/csl/CasaresIKM025}
Antonio Casares, Olivier Idir, Denis Kuperberg, Corto Mascle, and Aditya
  Prakash.
\newblock On the minimisation of deterministic and history-deterministic
  generalised (co)b{\"{u}}chi automata.
\newblock In J{\"{o}}rg Endrullis and Sylvain Schmitz, editors, {\em 33rd
  {EACSL} Annual Conference on Computer Science Logic ({CSL})}, volume 326 of
  {\em LIPIcs}, pages 22:1--22:18. Schloss Dagstuhl - Leibniz-Zentrum f{\"{u}}r
  Informatik, 2025.
\newblock \href {https://doi.org/10.4230/LIPICS.CSL.2025.22}
  {\path{doi:10.4230/LIPICS.CSL.2025.22}}.

\bibitem{DBLP:conf/fsttcs/EhlersS22}
R{\"{u}}diger Ehlers and Sven Schewe.
\newblock Natural colors of infinite words.
\newblock In Anuj Dawar and Venkatesan Guruswami, editors, {\em 42nd {IARCS}
  Annual Conference on Foundations of Software Technology and Theoretical
  Computer Science ({FSTTCS})}, volume 250 of {\em LIPIcs}, pages 36:1--36:17.
  Schloss Dagstuhl - Leibniz-Zentrum f{\"{u}}r Informatik, 2022.
\newblock \href {https://doi.org/10.4230/LIPICS.FSTTCS.2022.36}
  {\path{doi:10.4230/LIPICS.FSTTCS.2022.36}}.

\bibitem{DBLP:journals/siamcomp/GareyJT76}
M.~R. Garey, David~S. Johnson, and Robert~Endre Tarjan.
\newblock The planar hamiltonian circuit problem is {NP}-complete.
\newblock {\em {SIAM} J. Comput.}, 5(4):704--714, 1976.
\newblock \href {https://doi.org/10.1137/0205049} {\path{doi:10.1137/0205049}}.

\bibitem{DBLP:journals/iandc/HenzingerKR02}
Thomas~A. Henzinger, Orna Kupferman, and Sriram~K. Rajamani.
\newblock Fair simulation.
\newblock {\em Inf. Comput.}, 173(1):64--81, 2002.
\newblock \href {https://doi.org/10.1006/INCO.2001.3085}
  {\path{doi:10.1006/INCO.2001.3085}}.

\bibitem{DBLP:conf/icalp/KuperbergS15}
Denis Kuperberg and Michal Skrzypczak.
\newblock On determinisation of good-for-games automata.
\newblock In Magn{\'{u}}s~M. Halld{\'{o}}rsson, Kazuo Iwama, Naoki Kobayashi,
  and Bettina Speckmann, editors, {\em 42nd International Colloquium on
  Automata, Languages, and Programming ({ICALP})}, volume 9135 of {\em Lecture
  Notes in Computer Science}, pages 299--310. Springer, 2015.
\newblock \href {https://doi.org/10.1007/978-3-662-47666-6\_24}
  {\path{doi:10.1007/978-3-662-47666-6\_24}}.

\bibitem{DBLP:reference/mc/Kupferman18}
Orna Kupferman.
\newblock Automata theory and model checking.
\newblock In Edmund~M. Clarke, Thomas~A. Henzinger, Helmut Veith, and Roderick
  Bloem, editors, {\em Handbook of Model Checking}, pages 107--151. Springer,
  2018.
\newblock \href {https://doi.org/10.1007/978-3-319-10575-8\_4}
  {\path{doi:10.1007/978-3-319-10575-8\_4}}.

\bibitem{DBLP:journals/ipl/Loding01}
Christof L{\"{o}}ding.
\newblock Efficient minimization of deterministic weak omega-automata.
\newblock {\em Inf. Process. Lett.}, 79(3):105--109, 2001.
\newblock \href {https://doi.org/10.1016/S0020-0190(00)00183-6}
  {\path{doi:10.1016/S0020-0190(00)00183-6}}.

\bibitem{DBLP:conf/fossacs/Schewe09}
Sven Schewe.
\newblock Tighter bounds for the determinisation of {B}{\"{u}}chi automata.
\newblock In Luca de~Alfaro, editor, {\em 12th International Conference on the
  Foundations of Software Science and Computational Structures ({FOSSACS})},
  volume 5504 of {\em Lecture Notes in Computer Science}, pages 167--181.
  Springer, 2009.
\newblock \href {https://doi.org/10.1007/978-3-642-00596-1_13}
  {\path{doi:10.1007/978-3-642-00596-1_13}}.

\bibitem{DBLP:conf/fsttcs/Schewe10}
Sven Schewe.
\newblock Beyond hyper-minimisation---minimising {DBA}s and {DPA}s is
  {NP}-complete.
\newblock In Kamal Lodaya and Meena Mahajan, editors, {\em {IARCS} Annual
  Conference on Foundations of Software Technology and Theoretical Computer
  Science ({FSTTCS})}, volume~8 of {\em LIPIcs}, pages 400--411. Schloss
  Dagstuhl - Leibniz-Zentrum f{\"{u}}r Informatik, 2010.
\newblock \href {https://doi.org/10.4230/LIPICS.FSTTCS.2010.400}
  {\path{doi:10.4230/LIPICS.FSTTCS.2010.400}}.

\bibitem{DBLP:conf/fsttcs/Schewe20}
Sven Schewe.
\newblock Minimising good-for-games automata is np-complete.
\newblock In Nitin Saxena and Sunil Simon, editors, {\em 40th {IARCS} Annual
  Conference on Foundations of Software Technology and Theoretical Computer
  Science ({FSTTCS})}, volume 182 of {\em LIPIcs}, pages 56:1--56:13. Schloss
  Dagstuhl - Leibniz-Zentrum f{\"{u}}r Informatik, 2020.
\newblock \href {https://doi.org/10.4230/LIPICS.FSTTCS.2020.56}
  {\path{doi:10.4230/LIPICS.FSTTCS.2020.56}}.

\end{thebibliography}

\clearpage
\appendix

\section{Notes on the complexity overview in Table~\ref{tab:complexities}}

In Table~\ref{tab:complexities}, we did not refer to the original publications if they are older than 30 years, but rather referred to more accessible sources. 
This includes referring to \cite{DBLP:reference/mc/Kupferman18} for the PSPACE-completeness of the minimization problem for the non-deterministic versions of all automaton types considered and mentioning \cite{DBLP:journals/ipl/Loding01} for the minimization of deterministic safety automata, even though for the state-based case, this already follows from the classical DFA minimization algorithm.

The corresponding claim in \cite{DBLP:reference/mc/Kupferman18} (Theorem 17) only states that universality checking of non-deterministic Büchi automata with state-based acceptance is PSPACE-complete. However, the PSPACE-completeness of the minimization problem for all cases considered here follow immediately from the arguments from that proof.

To see this, note that PSPACE-containment of the minimization problems is simple to show: An algorithm can non-deterministically guess a smaller automaton and check it for equivalence, which is in PSPACE for non-deterministic Büchi automata and all automata types that can be translated to such with only polynomial blow-up (as the ones considered here). Finally, Savitch' Theorem can be applied.

For the PSPACE-hardness of the automaton types considered here, it can be observed that the proof of Theorem 17 in \cite{DBLP:reference/mc/Kupferman18} actually reduces from the problem of universality checking of a weak automaton in which there is only one absorbing accepting state. As in weak automata, the sizes of transition-based and state-based automata are exactly the same, the result follows for all acceptance conditions except for the safety acceptance condition.

For the safety acceptance condition, some additional details are needed: the argument in \cite{DBLP:reference/mc/Kupferman18} reduces the non-acceptance of a polynomially-spaced Turing machine to the universality problem of an automaton. The word to be accepted is encoded into the input of the automaton. The construction can be altered so that once an accepting state of the Turing machine is reached, the run of the automaton ends, so that the corresponding run can only be accepted if the automaton can detect an illegal computation. The automaton is then a safety automaton.

\section{Additional Preliminaries}
\label{sec:additionalPrelims}

For the detailed proofs in the following, we employ some notation that was not introduced in the preliminaries section of the main document because it was not needed there. We state the additional preliminaries here.

\subsection{Automata}
Given a finite word $w \in \Sigma^*$, as also say that $w^\omega$ is the \emph{omega} of $w$.

An automaton $\A$ is \emph{semantically deterministic} if different nondeterministic choices lead to equivalent states. Thus, for every state $q\in Q$ and letter $\sigma \in \Sigma$, all the $\sigma$-successors of $q$ are equivalent.

The following proposition follows easily from the definitions and suggests that from equivalent states in a semantically deterministic automaton, we can move only to equivalent successor states upon reading the same letter. In particular, for all finite words $x\in \Sigma^*$, all the states in $\delta(q_0, x)$ are equivalent: 

\begin{proposition}\label{SD prop}
    Consider a semantically deterministic automaton $\A = \zug{\Sigma, Q, q_0, \delta, \alpha}$, and transitions $\zug{q, \sigma, q'}$ and $\zug{s, \sigma, s'}$ of $\A$. If $q\sim_\A s$, then $q' \sim_\A s'$. 
\end{proposition}

For an automaton $\A$, we say that a state $q$ of $\A$ is \emph{HD} if $\A^q$ is HD. 
We say that $q$ is {\em reachable in $\A$} if $q\in  \delta(q_0, x)$ for some finite word $x\in \Sigma^*$.
Recall that two states $q,s\in Q$ are equivalent ($q \sim_{\A} s$) if $L(\A^q) = L(\A^s)$. 
Then, $q$ and $s$ are \emph{strongly-equivalent}, denoted $q \approx _{\A} s$, if $q \sim_{\A} s$ and $L_{\it safe}(\A^q) = L_{\it safe}(\A^s)$. 
When $\A$ is clear from the context, we omit it from the notations, thus write $L_{\it safe}(q)$,
$q\sim s$, $q\approx s$, etc.
Then, we say that a tNCW $\A$ is \emph{safe deterministic} if by removing its rejecting transitions, we remove all nondeterministic choices. Thus, for every state $q\in Q$ and letter $\sigma\in \Sigma$, it holds that there is at most one safe $\sigma$-labeled transition from $q$.

We now combine several properties defined above and say that an HD-tNCW $\A$ is {\em nice\/} if all the states in $\A$ are reachable and HD, and $\A$ is normal, safe deterministic, and semantically deterministic.

    \begin{example}
{\rm 
Consider the tDCW $\A$ on the right. Dashed transitions are rejecting transitions.
The tDCW $\A$ recognizes the language of all infinite words with finitely many $a$'s or finitely many $b$'s. Then,  $\A$ has three safe components: $\SC(\A) = \{ \{q_0\}, \{q_1\},\{q_2\}\}$.}
All the states of $\A$ are reachable, and $\A$ being deterministic implies that it is already safe-deterministic, semantically-deterministic, and all its states are HD. 

\noindent
\begin{minipage}{3.3in}
\vspace*{3pt}
{\rm %
Note that the word $w = a^\omega$ is such that the run $r$ of $\A^{q_0}$ on $w$ is safe, and thus $w\in L_{\it safe} (\A^{q_0})$, yet $r$ leaves the safe component $\{q_0\}$. If, however, we make $\A$ normal by turning the transitions $\zug{q_0, a, q_1}$ and $\zug{q_1, b, q_2}$ to be rejecting, then we get that $\A$ is nice, in particular, safe runs from a state $q$ are exactly the runs from $q$  that do not leave the safe component $\SC(q)$.\hfill \qed} 
\end{minipage} \ \hspace{.1in}
\begin{minipage}{1.7in}
\begin{tikzpicture}
\node[state,fill=black!10!white] (q0) at (0,0) {$q_0$};
\node[state,fill=black!10!white] (q1) at (2,0) {$q_1$};
\node[state,fill=black!10!white] (q2) at (4,0) {$q_2$};
\draw[->,semithick] (q0) edge[bend left=10] node[above] {$a$} (q1);
\draw[->,semithick] (q1) edge[bend left=10] node[above] {$b$} (q2);
\draw[->,semithick] (q0) edge[loop above] node[above] {$b$} (q0);
\draw[->,semithick] (q1) edge[loop above] node[above] {$a$} (q1);
\draw[->,semithick,dashed] (q2) edge[bend left=30] node[below] {$a,b$} (q0);
\draw[fill=black] (-0.7,0.5) circle (0.05cm);
\draw[->,semithick] (-0.7,0.5) -- (q0);
\end{tikzpicture}
\end{minipage}
\end{example}

The following theorem suggests that nice minimal equivalent HD-tNCWs have isomorphic safe components~\cite{DBLP:journals/lmcs/RadiK22}, in particular, their structures differ only in the way rejecting transitions are defined across safe components:

\begin{theorem}\label{iso thm}
    Consider nice HD-tNCWs $\A = \zug{\Sigma, Q, q_0, \delta, \alpha}$ and $\B=\zug{\Sigma, Q', q'_0, \delta', \alpha'}$. If $\A$ and $\B$ are both minimal and equivalent, then there is a bijection $\gamma: Q\to Q'$ with the following properties: 
    \begin{itemize}
        \item 
        For all states $q, s\in Q$ and letter $\sigma\in \Sigma$, it holds that $\zug{q, \sigma, s}$ is a safe transition of $\A$ if and only if   $\zug{\gamma(q), \sigma, \gamma(s)}$ is a safe transition of $\B$. Thus, the bijection $\gamma$ induces an isomorphism between the safe components of $\A$ and $\B$.

        \item
        For all states $q\in Q $, it holds that $q\approx \gamma(q)$. Thus, $\gamma$ maps a state $q$ to a state that not only has the same safe language as $q$, but also the same co-\buchi language.
    \end{itemize}
\end{theorem}

We say that a tNCW $\A = \zug{\Sigma, Q, q_0, \delta, \alpha}$ is $\alpha$-saturated if for all rejecting transitions $\zug{q, \sigma, s}$ of $\A$, and states $p$ with $p\sim s$, it holds that $\zug{q, \sigma, p}$ is also a transition of $\A$. 
The following theorem is, in a sense, dual to the fact that we can assume that every HD-tNCW is nice~\cite{DBLP:journals/lmcs/RadiK22},
and it is used in the several proofs that show that
the tNCWs we define in this paper are already HD:

\begin{theorem}\label{are GFG thm}    
    Consider a tNCW $\A = \zug{\Sigma, Q, q_0, \delta, \alpha}$. If $\A$ is normal, safe-deterministic, semantically-deterministic, $\alpha$-homogeneous, and $\alpha$-saturated, then it is HD.
\end{theorem}

\begin{proof}
    Let $<$ be a total order on the set of states $Q$.
Consider a finite word $x = \sigma_1\cdot \sigma_2 \cdots \sigma_n$ and a run $r = r_0, r_1, \ldots, r_n$ of $\A$ on $x$.
Let $i_{r, x}$ denote the minimal index in $\{ 0, 1, \ldots, n\}$ such that $r_i, r_{i+1}, \ldots, r_n$ is a safe run on the suffix $x[i+1, n]$.
For a finite word $x$ of length $n$, and runs $r = r_0, r_1, \ldots, r_n$ and $r' = r'_0, r'_1, \ldots, r'_n$ of $\A$ on $x$, we say that {\em $r$ is older than $r'$ with respect to $x$} if either $i_{r, x} < i_{r', x}$, or $i_{r, x} = i_{r', x}$ and $r_{i_{r, x}} < r'_{i_{r', x}}$. 
Thus, older runs are runs that have been safe on a longer suffix of $x$. We let $r_{\it old} (x)$ denote the oldest run %
of $\A$ on $x$.

We show that $\A$ is HD by defining a function $f: \Sigma^* \to Q$ that is consistent with $\delta$ and is such that the induced run $f(w)$ of $\A$ on $w$ is an accepting run for all $w\in L(\A)$.
We first define $f(\epsilon) = q_0$. Then, for a finite word $u\in \Sigma^*$ and letter $\sigma\in \Sigma$, we define $f(u\cdot \sigma)$ as follows.
If $\zug{f(u), \sigma, p}$ is a safe transition in $\A$, then $\A$ being safe-deterministic and $\alpha$-homogeneous implies that $p$ is the only $\sigma$-successor of $f(u)$, and we define $f(u\cdot \sigma) = p$.
Otherwise, namely, if all $\sigma$-transitions from $f(u)$ in $\A$ are rejecting, we define $f(u\cdot \sigma) = s$, where $s$
is the state that the run $r_{\it old}(u\cdot \sigma)$ ends in. %
Thus, $f$ tries to follow the deterministic safe components of $\A$ when possible, and when it needs to guess, it follows the oldest run of $\A$ with respect to the prefix of the input word read so far.
We show next that $f$  is well-defined. In fact, we show a stronger claim, namely, we show by induction on $|u|$ that for all finite words $u\in \Sigma^*$ and letters $\sigma\in \Sigma$, it holds that $f(u)$ is equivalent to all the states in $\delta(q_0,u)$, and $\zug{f(u), \sigma, f(u\cdot \sigma)}\in \Delta$. 
For the base case, we have that $f(\epsilon) = q_0$, and indeed $q_0$ is equivalent to all the states in $\delta(q_0, \epsilon) = \{q_0\}$. Then, for a finite word $u$ such that $f(u)$ is equivalent to all the states in $\delta(q_0,u)$, and a letter $\sigma$, we distinguish between two cases. The first case is when $f(u)$ has a safe $\sigma$-labeled transition $\zug{f(u), \sigma, p}$, and in this case, we defined $f$ such that $f(u\cdot \sigma) = p$. The other case is when all the $\sigma$-labeled transitions from $f(u)$ are rejecting.   
Recall that $f(u)$ is equivalent to all the states in $\delta(q_0, u)$. Then, 
$\A$ being semantically deterministic, we get by Proposition~\ref{SD prop} that all the states in $\delta(q_0, u\cdot \sigma)$ are equivalent to the $\sigma$ successors of $f(u)$. In particular, the state $q$ that the run 
$r_{\it old}(u\cdot \sigma)$ ends in %
is equivalent to the $\sigma$-successors of $f(u)$.
 Therefore, as the $\sigma$-labeled transitions from $f(u)$ are rejecting, and $\A$ is $\alpha$-saturated, it follows that $\zug{f(u), \sigma, q}$ is a transition in $\Delta$, and $f$ is defined to follow that transition. That is $f(u\cdot \sigma) = q$. Hence, in both cases, we have that $\zug{f(u), \sigma, f(u\cdot \sigma)} \in \Delta$. In particular, by Proposition~\ref{SD prop},  we have that $f(u\cdot \sigma)$ is equivalent to all the states in $\delta(q_0, u\cdot \sigma)$, and the induction is concluded.

To conclude the proof, consider a word $w\in L(\A)$, and let $r = r_0, r_1, r_2, \ldots$ be an accepting run of $\A$ on $w$. %
Thus, there is $i \geq 0$ such that $r[i, \infty] = r_i, r_{i+1}, \ldots$ is a safe run on the suffix $w[i+1, \infty]$.
We claim that the run $f(w)$ traverses finitely many rejecting transitions upon reading the suffix $w[i+1, \infty]$, and thus is an accepting run of $\A$ on $w$:
fix a finite run $p = p_0, p_1, \ldots, p_i$ of $\A$ on the prefix $w[1, i]$. We have one of two possible scenarios: 
\begin{enumerate}
    \item 
    Either the finite run $p$ can be extended to a run on $w$ by appending an infinite safe run $p_i, p_{i+1}, \ldots$ on the suffix $w[i+1, \infty]$. In this case, as $\A$ is safe-deterministic and alpha-homogeneous, we get that $p_i, p_{i+1}, \ldots$ is the only run from $p_i$ on the suffix $w[i+1, \infty]$, and so there is a unique extension of $p$. 

    \item
    Or, the finite run $p$  cannot be extended to a run on $w$ by appending a safe run $p_i, p_{i+1}, \ldots$ on the suffix $w[i+1, \infty]$. In this case, consider the minimal $j \geq i$ such that $p_i \xrightarrow{w[i+1, j]} p_j$ is a safe run and the state $p_{j}$ has no safe $w[j+1]$-labeled transition. Then, as $r[i, \infty]$ is a safe run on the suffix $w[i+1, \infty]$, it follows that for all $k > j$, the run $r[0, k] = r_0, r_1, \ldots, r_k$ on $w[1, k]$ is older than all runs that extend the finite run $p$ to a run on $w[1, k]$. Thus, the run $r$ eventually becomes always older than all runs that extend the finite run $p$.
\end{enumerate}
By the argument above, it follows that if the run $p = p_0, p_1, \ldots, p_i$ on $w[1, i]$ is the oldest run with respect to $w[1, i]$, then it either can be extended to an infinite run on $w$ so that its suffix on $w[i+1, \infty]$ is safe, and thus its unique extension remains the oldest run with respect to every prefix $w[1, j]$ for $j\geq i$, or at some point, upon reading a letter $w[j+1]$, %
we get that the run $r[0, k]$ on $w[1, k]$ for all $k \geq  j+1$, is older than all the runs that extend $p$ to a run on $w[1, k]$. 
Therefore, as there are finitely many candidate runs $p = p_0, p_1, \ldots, p_i$ on the prefix $w[1, i]$, and as one of these candidates can be extended to the accepting run $r$, it follows that the oldest run, while reading the suffix $w[i+1, \infty]$, can switch only finitely many times to a run extending some of the candidate runs $p$ by appending a safe run from $p_i$. %
In particular, as $r[i, \infty]$ is safe, and as the strategy $f$ switches to follow the oldest run upon reading a rejecting $w[j+1]$-labeled transition from the state $f(w[1, j])$ for $j \geq i$, it follows that $f$ eventually follows the same oldest run that is safe on the suffix $w[i+1, \infty]$; In particular $f(w)$ is accepting, and we're done.
\end{proof}

\section{Proof of Theorem \ref{HP is hard thm}}
\label{HP is hard thm app}
 A graph $G = \zug{V, E}$ is {\em 3-regular} when every vertex of $G$ has exactly 3 neighbors.
    For proving Theorem~\ref{HP is hard thm}, we employ an existing result on the NP-hardness of checking if a $3$-regular graph has a Hamiltonian cycle is NP-hard~\cite{DBLP:journals/siamcomp/GareyJT76}.
    We first note that we can assume that the latter problem is NP-hard already for strongly-connected 3-regular graphs. Indeed, the reduction that witnesses NP-hardness of the problem can, before returning a 3-regular graph $G$, check the connectivity of $G$  and output instead of $G$ a constant strongly-connected 3-regular graph $G'$ that has no Hamiltonian-cycle when $G$ is not strongly-connected.

    We reduce from the Hamiltonian-cycle problem as follows. Given a strongly-connected 3-regular graph $G = \zug{V, E}$, the reduction returns a
    graph $G'$ with a designated vertex $v_1$ that is obtained from $G$ by taking an arbitrary vertex $v\in V$ and splitting it into two vertices $(v, in)$ and $(v, out)$, where all edges that enter $v$ in $G$ now enter $(v, in)$ in $G'$,  all edges that leave $v$ in $G$ now leave $(v, out)$ in $G'$,  we add an edge from $(v, in)$ to $(v, out)$, and 
    finally,  we add a new vertex $v_1$, and an edge from it to $(v, out)$.
    Thus, the graph $G'$ is obtained from $G$ by splitting the vertex $v$ according to its incoming and outgoing edges, adding an edge from $(v, in)$ to $(v, out)$, and adding a new vertex $v_1$ with the edge $(v_1, (v, out))$. 
    Formally, let $v$ be an arbitrary vertex in $G$. Then,
    we define $G' = \zug{V', E'}$ with a designated vertex $v_1$, where: 
    \begin{itemize}
        \item 
        The vertex $v_1$ is not in $V$, and $V' = (V\setminus \{ v\}) \cup \{(v, in), (v, out), v_1 \}$.

        \item 
        The edges $E'\subseteq V'\times V'$ are defined as
        $E' = (E \setminus ( (\{v\}\times V) \cup (V \times \{v\})) \cup \{ (u, (v, in)): E(u, v) \} \cup \{ ( (v, out), u): E(v, u) \} \cup \{  (v_1, (v, out))\}$.
    \end{itemize}

    Before proving the correctness of the reduction, let us verify that $G'$ satisfies the desired properties.
    First, it is not hard to verify that $G'$ is $(\leq 3)$-regular.
    Indeed, the vertex $(v, out)$ has 3 neighbors as the vertex $v$ in $G$, the vertex $(1, in)$ has a single neighbor which is $(v, out)$, all the vertices $u$ in $V\setminus \{ v\}$ have exactly 3 neighbors (their neighbors in $G$), and finally the vertex $v_1$ has  $(v, out)$ as its only neighbor. Then, note that $v_1$ has no incoming edges.
    In addition, the fact that $G$ is strongly-connected implies that the sub-graph $G'|_{\neg v_1}$ is strongly-connected as well as any path $p$ that visits all the vertices of $G$ induces a path $p'$ that visits all the vertices in $G'|_{\neg v_1}$ where every visit to $v$ in $p$ is replaced by the edge $((v, in) , (v, out))$ in $p'$.

    We conclude the proof by showing the correctness of the reduction.
    To begin with, it is not hard to see that a Hamiltonian cycle $p$ from $v$ to $v$ in $G$ induces a Hamiltonian path from $v_1$ to $(v, in)$ in $G'$ by first traversing the edge $(v_1, (v, out))$, then following the same sequence of vertices that $p$ follows to visit all the vertices in $V\setminus \{ v\}$, and then end in $(v, in)$.
    For the other direction, if $p$ is a Hamiltonian path starting at $v_1$ in $G'$, then after traversing the edge $(v_1, (v, out))$, $p$ must visit first every vertex in $V\setminus \{v\}$ before visiting  $(v, in)$. Indeed, $(v, in)$ leads to $(v, out)$ without visiting any vertex in $V\setminus \{v\}$ in between.
    Therefore, as $G'|_{\neg v_1}$ preserves original edges of $G$, it follows that $p$ induces a Hamiltonian cycle from $v$ to $v$ in $G$, and we are done.

\section{Missing Details from Section~\ref{exp red sec}}
\label{exp red sec app}

\subsection{Missing Details in the Definition of $\A^{G^{v_1}}$}
\label{exp red defs app}
Consider a vertex $v\in V$. We start with the formal definition of the safe component $\SC(v)$.  
Let $\eta(v) = \{ v, u_1, u_2, \ldots, u_m \}$ denote the neighborhood of $v$. Then, $\SC(v)$ is defined on top of a state-space consisting  of vectors of length $m+3$ of the form $[p_1(u_1), p_1(u_2), \ldots, p_1(u_m), p_1(v), \allowbreak{} p_2(v), v]$, where for every vertex $v'$ in $\eta(v)$, $p_1(v')$ is a bit in $ \{0, 1\}$ that represents the value of the vertex $v'$  under the assignment by the last read letter. Similarly, $p_2(v)$ represents the value of $v$ assigned by the past-two letter, and the last coordinate in the vector is labeled by $v$ to distinguish between vectors corresponding to distinct safe components. 

We define next safe transitions inside $\SC(v)$. For a letter $\sigma \in 2^{\mathsf{AP}}$, and a vector  $vec = [p_1(u_1), p_1(u_2), \ldots, p_1(u_m), p_1(v), p_2(v), v]$, we say that $\sigma$ is {\emph consistent} with the vector
$vec$ when $\sigma(\hat u_j) = p_1(u_j)$ for all $j\in [m]$, and $\sigma(\check v) = p_2(v)$. Thus, the letter $\sigma$ is consistent with the vector $vec$ when the values specified by $\sigma$ are consistent with the past values encoded in $vec$, as required by $L_v$.
Then, for every letter $\sigma\in 2^{\mathsf{AP}}$, and vector $vec = [p_1(u_1), p_1(u_2), \ldots, p_1(u_m), p_1(v), p_2(v), v]$, we add (deterministically) a safe transition $vec \xrightarrow{\sigma}  [\sigma(u_1), \sigma(u_2), \ldots, \sigma(u_m), \sigma(v), p_1(v), v]$ only when $\sigma$ is consistent with the vector $vec$. 
Note that we update the next past-one values as specified by the letter $\sigma$, and we update the next past-two value of $v$ according to its recent past-one value.
Thus, if a run enters a vector $vec$ in $\SC(v)$ after reading a finite word $x\in (2^{\mathsf{AP}})^*$ containing at least two letters, and $vec$ encodes the correct past values as specified by $x$, then for every infinite word $y \in (2^{\mathsf{AP}})^\omega$ with $x\cdot y\in L_v$,  
a safe run for $y$ exists from $vec$.

\sloppypar Next, we add safe self-loops labeled with the special letter $\text{\it sleep}$ for every vector $vec$ to implement the mixing with $(\text{\it sleep})^\omega$. 
Finally, note that $\SC(v)$ forms a safe component, that is, it is strongly-connected. Indeed, it is not hard to verify that every two vectors $vec^1  = [p^1_1(u_1), p^1_1(u_2), \ldots, p^1_1(u_m), p^1_1(v), p^1_2(v), v]$ and $vec^2  = [p^2_1(u_1), p^2_1(u_2), \ldots, p^2_1(u_m), p^2_1(v), \allowbreak{} p^2_2(v), \allowbreak{} v]$ are such that 

$$ vec^1 \xrightarrow{\{ v: p^2_2(v) = 1 \} \cup \{ \hat u_j: j\in [m] \ \wedge \  p^1_1(u_j) = 1 \} \cup \{ \check v: p^1_2(v) = 1\}} [0, 0, \ldots, 0, p^2_2(v), p^1_1(v), v]$$
 and 
$$
 [0, 0, \ldots, 0, p^2_2(v), p^1_1(v), v] \xrightarrow{\{ u_j: j\in [m] \ \wedge \ p^2_1(u_j) = 1 \} \cup \{ v: p^2_1(v) = 1\}  \cup \{ \check v: p^1_1(v) = 1\}} vec^2
$$
are safe transitions in $\SC(v)$ inducing a safe run from  $vec^1$ to  $vec^2$.

We proceed  with the formal definition of the safe component $\SC_{\it Sync}$. 
The safe component $\SC_{\it Sync}$ is formally defined over the state-space $\{1, 2, 3, 4\}\times \{ 0, 1\}$. The state $(t, k)$ in $\SC_{\it Sync}$ remembers that we have seen $\text{\it sleep}$ in the past-$t$ letter, and $k$ remembers the last value of $v_1$. In particular, the state $(1, 1)$ remembers that we have just read the special letter $\text{\it sleep}$ and the last value of $v_1$ is $1$, while the state $(2, 0)$ remembers the last time ${\it sleep}$ appeared was in the past-two letter and the past-one value of $v_1$ is 0.
 Then, safe transitions inside $\SC_{\it Sync}$ are defined as follows. Consider a state $(t, k)$, and a letter $\sigma\in 2^{\mathsf{AP}}$. We add the safe transition $\zug{(t, k), \it{sleep}, (1, k)}$. In addition,  if $t < 4$ and $\sigma(\text{\it extra}) = k$, then we add the safe transition $\zug{(t, k), \sigma, (t+1, \sigma(v_1)) }$.
Note that the right-coordinate always stores the last value of $v_1$ as specified by the last non-$\text{\it sleep}$ letter, and there are no safe non-$\it sleep$ transitions from states whose first state tuple element  is 4. Also, when $t< 4$, then safe $\sigma$-transitions are defined only when the ${\it extra}$-consistent property was not violated.
In particular, if a run enters the state $(1, k)$ in $\SC_{\it Sync}$ after reading a finite word $x\in \Sigma^*$, where $x$ contains at least one non-$\it sleep$ letter and $k$ encodes the correct past-one value of $v_1$ as specified by $x$, then $(1, k)$ has a safe run on any infinite word $y\in \Sigma^\omega$ with $x\cdot y\in L_{\it Sync}$.

Finally, note also that the safe component $\SC_{\it Sync}$ is strongly-connected as the following is a safe cycle visiting all its states:
\begin{multline*}
 (1, 0) \xrightarrow{\emptyset} (2, 0) \xrightarrow{\emptyset} (3, 0) \xrightarrow{\emptyset} (4, 0) \xrightarrow{\text{\it sleep} } (1, 0) \xrightarrow{\{v_1\}} (2, 1) \xrightarrow{\{ v_1, \text {\it extra}\}} (3, 1) \\ \xrightarrow{\{ v_1, \text{\it extra}\}} (4, 1)  \xrightarrow{\text{\it sleep}} (1, 1) \xrightarrow{\{ \text{\it extra}\}} (2, 0) \xrightarrow{\text{\it sleep}} (1, 0)
 \end{multline*}

\subsection{Useful Properties of $\A^{G^{v_1}}$}
\label{useful properties app}

We prove in this section some useful properties of the tNCW $\A^{G^{v_1}}$ that are used only in the appendices. 
The following Lemma suggests that $\A^{G^{v_1}}$ recognizes the expected language and satisfies some immediate syntactic properties:

\begin{lemma}\label{useful prop lem}
    \begin{enumerate}
        \item 
        The tNCW $\A^{G^{v_1}}$ is normal, safe-deterministic, and $\alpha$-homogenuous.

        \item 
        Consider a state $q\in Q$ of $\A^{G^{v_1}}$. Then, it holds that $L(q)$ is the language of infinite words $w$ having a suffix 
        in $L_{\it Sync}$ or a suffix that  is a mix of $(\it sleep)^\omega$ and some word in $(\bigcup_{v\in V} L_v)$. In particular, all the states of $\A^{G^{v_1}}$ are language-equivalent.
    \end{enumerate}
\end{lemma}

\begin{proof}
    \begin{enumerate}
    \item
    Immediate from the definition of $\A^{G^{v_1}}$.

    \item 
    We start with the easy direction, namely, showing that every infinite word $w = \sigma_1 \cdot \sigma_2 \cdots $ in $L(q)$ has a suffix in $L_{\it Sync}$ or a suffix that is a mix of $(\it sleep)^\omega$ and some word in $(\bigcup_{v\in V} L_v)$. 
    Let $r = r_0, r_1, r_2, \ldots$ be an accepting run of $q$ on $w$. 
    As $r$ is accepting, there is some $i\geq 0$ such that the run $r[i, \infty] = r_i, r_{i+1}, r_{i+2},  \ldots$ is a safe run on the suffix $w[i+1, \infty]$. Either $r_i\in \SC_{\it Sync}$, or  $r_i \in \SC(v)$ for some vertex $v\in V$. In both cases, the fact that the run $r[i, \infty]$ is safe implies that it moves to a state encoding the correct past-one values upon reading a non-$\it sleep$ letter. Therefore, since $r[i, \infty]$ is safe, we have that 
      the word $w[i+1, \infty]$ is in $ L_{\it Sync}$ when $r_i \in \SC_{\it Sync}$, and that   the word 
    $w[i+1, \infty]|_{\neg \it sleep}$ is  either infinite and in $L_v$ or can be extended to a word in $L_v$ when $r_i\in \SC(v)$. 
    Indeed, in the former case, the fact that $r[i, \infty]$ is safe implies that $w[i+1, \infty]$ contains $\it sleep$ in every at most 4 letters and  does not violate the $\it extra$-consistent property, and in the latter case, the fact that $r[i, \infty]$ is safe implies that the past-two values get updated according to the previous past-one values, and thus $r[i, \infty]$ being safe also witnesses that there is consistency in propositional values in $\eta(v)$ across every three consecutive non-$\it sleep$ letters in $w[i+1, \infty]$. 
    
    We proceed with the other direction. Consider an infinite word $w$ having a suffix $w[i, \infty]$ that is either in $ L_{\it Sync}$ or is a mix of $(\it sleep)^\omega$ and some word in $(\bigcup_{v\in V} L_v)$. We need to show that $q$ has an accepting run on $w$. If $w$ has a suffix in $(\it sleep)^\omega$, then we're done as every state has a safe run on $(\it sleep)^\omega$; in particular, any run of $q$ on $w$ can eventually branch to an accepting run. 
    We proceed to the case where $w$ has infinitely many non-$\it sleep$ letters. Let 
     $r = r_0, r_1, r_2, \ldots$ be an arbitrary run of $q$ on $w$. If $r$ is accepting, then $w\in L(q)$ and we're done. 
    Otherwise,
    if the run $r$ is not accepting, we show that we can modify it by letting it branch to an accepting run of $q$ on $w$. %
    Consider some $j$ such that $j\geq i$ and the infix $w[i, j]$ contains at least two non-$\it sleep$ letters, and $\zug{r_{j-1}, w[j], r_{j}}$ is a rejecting transition of $r$.
    We distinguish between two cases. Either $w[i, \infty]$ is a mix of $(\it sleep)^\omega$ and some word in $L_v$ for some vertex $v\in V$, or $w[i, \infty] \in L_{\it Sync}$. In the former case, we modify the run $r$ by proceeding to the state $s\in \SC(v)$ instead of $r_{j}$, where $s$ is the state that encodes the correct 
    past values  of the vertices in $\eta(v)$ with respect to the infix  $w[i, j]$.  In the latter case, we modify the run $r$ by proceeding to the state $s\in \SC_{\it Sync}$ instead of $r_{j}$, where $s = (1, k)$ is the state that encodes  the %
    past value $k$ of $v_1$ as specified by last non-$\it sleep$ letter in $w[i, j]$, and encodes that we have just seen $\it sleep$. In both cases, we moved to a state $s$ encoding the correct past values of atomic propositions  as specified by the last two non-$\it sleep$ letters in  $w[i, j]$, and thus since  %
    $s$ belongs to a safe component that can trap the suffix to be read, 
    we get that  $s$ has a safe run on $w[j+1, \infty]$. 
    Thus, $r$ can be modified to follow an accepting run of $q$ on $w$, as required.
    Note that by the definition of rejecting transitions, $r$ can indeed move to any state in $Q$ upon traversing a rejecting transition, and so the latter runs exist.
 \qedhere
\end{enumerate}
\end{proof}

     We proceed with the following proposition suggesting that $\A^{G^{v_1}}$ is already a nice minimal HD-tNCW:

\begin{proposition}\label{A is nice prop}
  The tNCW  $\A^{G^{v_1}}$ is $\alpha$-saturated, $\alpha$-homogenuous, nice, minimal, and HD. %
\end{proposition}

\begin{proof}

As we argue below, $\A^{G^{v_1}}$ is $\alpha$-saturated. Thus, it suffices by Lemma~\ref{useful prop lem} to show that $\A^{G^{v_1}}$ is nice, minimal, and HD.
We start by showing that $\A^{G^{v_1}}$ is nice. First, the fact that all the states of $\A^{G^{v_1}}$ are reachable is immediate from the fact that rejecting transitions lead to all states in $Q$. In particular, $q_0$ can reach any other state $q$ by leaving the safe component $\SC(q_0)$ via traversing a rejecting transition that leads to $q$. Then, the fact that $\A^{G^{v_1}}$ is semantically-deterministic is immediate from item 2 of Lemma~\ref{useful prop lem} stating that all its states are equivalent. Next, by item 1 of Lemma~\ref{useful prop lem}, we have that $\A^{G^{v_1}}$ is normal and safe-deterministic. To conclude that $\A^{G^{v_1}}$ is nice, it is left to show that every state $q\in Q$ is  HD, and this follows by Theorem~\ref{are GFG thm}. Indeed, item 1 of Lemma~\ref{useful prop lem} implies that $(\A^{G^{v_1}})^q$ is $\alpha$-homogenuous, and since all  states in $Q$ are equivalent and rejecting transitions lead to all states, we have that $(\A^{G^{v_1}})^q$ is $\alpha$-saturated as well. 

    To show that $\A^{G^{v_1}}$ is minimal, it is sufficient by~\cite{DBLP:journals/lmcs/RadiK22} to show that  every distinct states $q$ and $s$ in $Q$ are such that (1) $L_{\it safe}(q) \neq L_{\it safe}(s)$, and (2) if $L_{\it safe}(q) \subseteq L_{\it safe}(s)$, then $q$ and $s$ belong to the same safe component.
    We distinguish between cases: 
    \begin{itemize}
        \item There are distinct vertices $v$ and $v'$ in $V$ such that $q\in \SC(v)$ and $s\in \SC(v')$:
        we show that there is an infinite word $L_{\it safe}(q) \setminus L_{\it safe}(s)$. As the latter holds for every distinct states  $q$ and $s$ that belong to distinct safe components, and are not in $\SC_{\it Sync}$, we conclude that the safe languages of $q$ and $s$ are incomparable. 
        Let $\eta(v) = \{ v, u_1, u_2, \ldots, u_m\}$ denote the neighborhood of $v$,  $q = [p_1(u_1), p_1(u_2), \ldots, p_1(u_m), p_1(v), p_2(v), v]$, and let $p_2(v')$ denote the past-two value of $v'$ as encoded in the state $s$. Then, the letter $\sigma = \{ \hat u_j: \ j\in [m] \wedge p_1(u_j) = 1  \} \cup \{ \check v: p_2(v) = 1\} \cup \{ \check v': p_2(v') = 0\}$ is such that $q$ has a safe outgoing $\sigma$-transition, yet all $\sigma$ transitions from $s$ are not safe. 
        As $\A^{G^{v_1}}$ is normal, then $\sigma$ can be extended to an infinite word in $L_{\it safe}(q) \setminus L_{\it safe}(s)$.

        \item There is a vertex $v$ in $V$ such that both $q$ and $s$ belong to $\SC(v)$: as in the previous item, we show that the safe langauges of $q$ and $s$ are incomparable.
        Again, as $\A^{G^{v_1}}$ is normal, it is sufficient to show the existence of a finite word $x$ such that $q$ has a safe run on $x$, yet all runs of $s$ on $x$ are not safe.
        Let $\eta(v) = \{ v, u_1, u_2, \ldots, u_m\}$ denote the neighborhood of $v$, let $q = [p_1(u_1), p_1(u_2), \ldots, p_1(u_m), p_1(v), p_2(v), v]$ and $s = [p'_1(u_1), p'_1(u_2), \ldots, p'_1(u_m), p'_1(v), p'_2(v), v]$. Consider the letter $\sigma = \{ \hat u_j: \ j\in [m] \wedge p_1(u_j) = 1  \} \cup \{ \check v: p_2(v) = 1\} $. Clearly, $q$ has an outgoing safe $\sigma$-transition. If all the $\sigma$-transitions from $s$ are rejecting, then by choosing $x = \sigma$, we're done. Otherwise, if $s$ has an outgoing safe $\sigma$-transition, then it must be the case that $s = [p_1(u_1), p_1(u_2), \ldots, p_1(u_m), p'_1(v), p_2(v), v]$, where $p_1(v) \neq p'_1(v)$. Thus, $q$ and $s$ disagree only in the past-one value of $v$. Consider the safe transitions $q \xrightarrow{\sigma}  [0, 0, \ldots, 0, p_1(v), v] $ $s \xrightarrow{\sigma}  [0, 0, \ldots, 0, p'_1(v), v]$, and note that the letter $\sigma' = \{ \check v: p_1(v) = 1\}$ is such that former transition can be extended to a safe run on $\sigma\cdot \sigma'$, yet the latter transition cannot. Hence, as $\A^{G^{v_1}}$ is safe-deterministic, it follows that all the runs of $s$ on $\sigma\cdot \sigma'$ are not safe, and so by choosing $x = \sigma \cdot \sigma'$, we're done.
    
        \item 
        $q\in \SC(v)$ for some vertex $v\in V$ and $s\in \SC_{\it Sync}$:  we  first claim that $L_{\it safe}(q)$ is not contained in $L_{\it safe}(s)$, and this is immediate from the fact that we can take an infinite run in $\SC(v)$ from $q$ that is labeled with a word that has no $\it sleep$ letters. Next,  we claim that $L_{\it safe}(s)$ is not contained in $L_{\it safe}(q)$ either. Indeed,  let $s = (t, k)$, then the word $x = \text{\it sleep} \cdot (\{ \it extra:  k = 1 \}\cup \{ \check v:  p_2(v) = 0\})$, where $p_2(v)$ is the past-two value of $v$ as encoded in $q$ is such that there is a safe run from $s$ on $x$, yet all the runs from $q$ on $x$ are not safe. As $\A^{G^{v_1}}$ is normal, then $x$ can be extended to an infinite word in $L_{\it safe}(s) \setminus L_{\it safe}(q)$.
        Thus, we have established that $L_{\it safe}(q)$ and $L_{\it safe}(s)$ are incomparable.

        \item 
        Both states $q$ and $s$ are in $\SC_{\it Sync}$:
        let $q = (t, k)$ and $s = (t', k')$. 
        We show that $L_{\it safe}(q) \neq L_{\it safe}(s)$. 
        As $q$ and $s$ are distinct, then we have two cases. 
        The first case is when $k \neq k'$. In this case, by considering the word $ x = \text{\it sleep }\cdot \{\it extra\}$, we get that there is a safe run  from only one of the states $q$ and $s$ on $x$. 
        The other case is when $k = k' = b$, yet w.l.o.g $t < t'$. In this case, consider the word $x = \{ extra: b = 1\} \cdot (\emptyset)^{4-t'} $. 
        Since $|x| > 4 - t'$, we get that 
        by reading $x$ from $s$, we leave $\SC_{\it Sync}$, yet since $4-t > 4-t' = |x| - 1$, we get that $4-t \geq |x|$, and thus by reading $x$ from $q$ we do not leave $\SC_{\it Sync}$.
     \qedhere   
    \end{itemize}

\end{proof}

\subsection{Missing Details in the Proof of Theorem~\ref{corr thm}}
\label{corr thm app}

In this section, we complete the missing details in the second direction of the proof of Theorem~\ref{corr thm}.
This direction is more involved and uses key properties of the safe component $\SC_{\it Sync}$.
Assume that there is a tDCW $\A_d$  for $L(\A^{G^{v_1}})$ of size $|Q|$. We need to show that $\A_d$ induces a Hamiltonian path starting at $v_1$ in $G$.
First, note that as $\A^{G^{v_1}}$ is a nice minimal HD-tNCW, and $\A_d$ is deterministic (and thus HD), then it is a minimal HD-tNCW as well. In addition, we can assume that $\A_d$ is nice as we can make it normal by modifying safe transitions connecting distinct safe components to be rejecting.
Therefore, by Theorem~\ref{iso thm}, we have that $\A^{G^{v_1}}$ and $\A_d$ have isomorphic safe components. Thus, as $\A^{G^{v_1}}$ is $\alpha$-saturated, and all its states are equivalent,  we can then assume that $\A_d = \zug{\Sigma, Q, q_0, \delta', \alpha'}$ is a deterministic pruning of $\A^{G^{v_1}}$. Indeed, by Theorem~\ref{iso thm}, we have that all the states of $\A_d$ are equivalent as well, and so the choice of its initial can be arbitrary.

Recall that $G = \zug{V, E}$, and let $n = |V|$.  
We prove that $\A_d$ induces a Hamiltonian path by showing the existence of a sequence of $n$ finite words $z_1, z_2, \ldots, z_n$ over $\Sigma$ such that the run of $(\A_d)^{(1, 1)}$ on $z = z_1\cdot z_2 \cdots z_n$ is of the form $ r_z = (1, 1) \xrightarrow{z_1} s_1\xrightarrow{z_2} s_2 \cdots \xrightarrow{z_n} s_n$,
where for all $i \in [n]$, it holds that the state $s_i$ belongs to a safe component $\SC(v_i)$ for some vertex $v_i\in V$, we enter the safe component $\SC(v_{i})$ after reading the last letter of $z_i$, and  $(v_1, v_2, \ldots, v_n)$ is a Hamiltonian path in $G$.
Thus,  the run $r_z$ starts at the state $(1, 1)$ in $\SC_{Sync}$, and switches between safe components by simulating a Hamiltonian path in $G$ starting at $v_1$.
\bigcomment{
\color{blue}As all states in $\mathcal{A}_d$ are language-equivalent, assuming that the run starts in state $(1,1)$ is without loss of generality.\color{black}
}%

Before defining the word $z = z_1\cdot z_2 \cdot z_3 \cdots z_n$ iteratively,  
we proceed with some definitions and lemmas  stating useful properties of the tDCW $\A_d$. The following lemma formalizes the fact that $\SC_{\it Sync}$ is a synchronizing safe component
suggesting that non-safe cycles in $\A_d$, labeled with finite words whose omega is $\it extra$-consistent,  must enter $\SC_{\it Sync}$:

\begin{lemma}\label{rej loop lemm}
    Consider a finite run $r = r_0, r_1, \ldots, r_m$ over a finite word $x = \sigma_1\cdot \sigma_2 \cdot \sigma_m$ in $\A_d$. 
    If $x^\omega$ is $\it extra$-consistent, 
    and $r$ is a non-safe cycle, then $r$ enters the safe component $\SC_{\it Sync}$. 
\end{lemma}

\begin{proof}
   Assume towards contradiction that $r$ does not enter $\SC_{\it Sync}$.
    In particular, every state that $r$ visits has a self-loop labeled with $\it sleep$. Now consider the non-safe cycle over the word $ y= \sigma_1 \cdot \it sleep \cdot \sigma_2 \cdots \it sleep \cdot \sigma_m \cdot \it sleep$ that is obtained from $r$ by traversing the self-loop labeled with $\it sleep$ after each transition of $r$. On the one hand, the latter cycle is non-safe, yet it is labeled with a word $y$ such that $y^\omega$ is $\it extra$-consistent, and $y$ contains $\it sleep$ regularly. %
    In particular, $y^\omega \in  L_{\it Sync} \subseteq L(\A_d)$. On the other hand, since all the states of $\A_d$ are equivalent, and $y$ is labeled on a non-safe cycle, we get that $y^\omega \notin L({\A_d})$, and we have reached a contradiction.
\end{proof}

Consider a transition $t = s_0 \xrightarrow{\sigma} s_1$ in $\A_d$ for some letter $\sigma\in 2^{\mathsf{AP}}$. We say that the transition $t$ is \emph{$\it extra$-consistent} when
    the state $s_0 $ is in  $\SC(v_1)$ and the past-one value of $v_1$ encoded in $s_0$ equals $\sigma(\it extra)$. 
    Note that since $v_1$ has no incoming edges in the graph $G$, then, among all the states in $\bigcup_{v\in V}\SC(v)$, only the states in $\SC(v_1)$ keep track of the past-one value of $v_1$.
The following lemma states sufficient conditions to when a finite cycle of safe components must visit $\SC_{\it Sync}$. It is key in the construction of the word $z = z_1\cdot z_2 \cdots z_n$ %
and is based on Lemma~\ref{rej loop lemm}:

\begin{lemma}\label{enters sync lemm}
    Consider a rejecting transition $t = \zug{s_0, \sigma_1, s_1}$ in $\A_d$, and a finite run of the form $r = s_0 \xrightarrow{\sigma_1} s_1 \xrightarrow{\sigma_2\cdot \sigma_3 \cdots \sigma_k} s_k$ where $\SC(s_0) = \SC(s_k)$, the word $\sigma_1\cdot \sigma_2 \cdots \sigma_k$ is over $2^{\mathsf{AP}}$, and $\{v_1, \it extra\} \subseteq \sigma_i$ for all $2\leq i \leq k$. Thus, $r$ starts and ends in the same safe component, and the (possibly empty) word $\sigma_2\cdot \sigma_3 \cdots \sigma_k $ is such that $v_1$ and $\it extra$ belong to every letter in it. If the transition $t = \zug{s_0, \sigma_1, s_1}$ satisfies one of the following conditions: 
    \begin{enumerate}
        \item  The transition $t = \zug{s_0, \sigma_1, s_1}$ is $\it extra$-consistent and $v_1\in \sigma_1$, or
        \item $\SC(s_0) = \SC(v)$ for some vertex $v\neq v_1$, and $\{v_1, \it extra\} \subseteq \sigma_1$.
    \end{enumerate}

    then the run $r$ enters the safe component $\SC_{\it Sync}$.
\end{lemma}

\begin{proof}
As the transition $t = \zug{s_0, \sigma_1, s_1}$ is rejecting, then, by Lemma~\ref{rej loop lemm}, it is sufficient to show the existence of a finite nonempty word $y$ such that $(\sigma_1\cdots \sigma_k \cdot y)^\omega$ is $\it extra$-consistent and  $r_c = s_0 \xrightarrow{\sigma_1} s_1 \xrightarrow{\sigma_2\cdot \sigma_3 \cdots \sigma_k} s_k \xrightarrow{y} s_0$ is a cycle in $\A_d$ that ends by traversing the safe run $ s_k \xrightarrow{y} s_0$ inside $\SC(s_0)$. %
To show the existence of $y$, we distinguish between two cases:
\begin{enumerate}
    \item 
    The transition $t = \zug{s_0, \sigma_1, s_1}$ is $\it extra$-consistent, and $v_1\in \sigma_1$:
    as $t$ is $\it extra$-consistent, then $s_0\in \SC(v_1)$. 
    Let $p^{s_0}_1(v_1)$ denote the encoded past-one value of $v_1$ in the state $s_0$. As $\SC(s_0) = \SC(s_k) = \SC(v_1)$, then 
there is a finite safe run $r'$ from $s_k$ to $s_0$ in $\SC(v_1)$ labeled with some finite nonempty word $y = y_1\cdot y_2 \cdots y_m$  over $2^{\mathsf{AP}}$. Indeed, the safe component $\SC(v_1)$ is non-trivial and contains $\it sleep$ transitions only as self-loops. As $r'$ is safe, labeled with a nonempty word, and ends in $s_0$, it follows that $y_m(v_1) = p^{s_0}_1(v_1)$. What we have so far is the following non-safe cycle $r_c = s_0 \xrightarrow{\sigma_1} s_1 \xrightarrow{\sigma_2\cdot \sigma_3 \cdots \sigma_k} s_k \xrightarrow{y_1\cdot y_2\cdots y_m} s_0$, where $y_m(v_1)  = p^{s_0}_1(v_1) = \sigma_1(\it extra) $. Indeed, the latter equation follows from the fact that the transition $\zug{s_0, \sigma_1, s_1}$ is $\it extra$-consistent.
Therefore, since the sub-run $r' = s_k \xrightarrow{y_1\cdot y_2\cdots y_m} s_0$ is a safe run in $\SC(v_1)$, and $\SC(v_1)$ does not track the value of $\it extra$, we can then modify the $\it extra$ values in $y$ without affecting transitions of the  sub-run $r'$ being safe, yet while guaranteeing  that $\sigma_k (v_1) = y_1(\it extra) $, and $y_i (v_1) = y_{i+1}(\it extra)$ for all $i \in [m-1]$. In particular, it is not hard to verify that $(\sigma_1\cdots \sigma_k \cdot y)^\omega$ is $\it extra$-consistent and we're done.

    \item 
    $\SC(s_0) = \SC(v)$ for some vertex $v\neq v_1$, and $\{v_1, \it extra\} \subseteq \sigma_1$: 
   as $v\neq v_1$, and $v_1$ has no incoming edges in $G$, we have that  $\SC(v)$ does not track past values of $v_1$. In particular, if $r' = s_k \xrightarrow{y_1\cdot y_2\cdots y_m} s_0$ is a safe run from $s_k$ to $s_0$ inside $\SC(v)$ over a nonempty word in $(2^{\mathsf{AP}})^*$, then we can assume w.l.o.g that every letter of $y$ contains $v_1$ and $\it extra$. Then, it is not hard to verify that the non-safe cycle  $r_c = s_0 \xrightarrow{\sigma_1} s_1 \xrightarrow{\sigma_2\cdot \sigma_3 \cdots \sigma_k} s_k \xrightarrow{y_1\cdot y_2\cdots y_m} s_0$ is such that $(\sigma_1\cdot \sigma_2 \cdots \sigma_k \cdot y)^\omega$ is $\it extra$-consistent. Indeed, every letter in the latter word contains $v_1$ and $\it extra$, and we're done. \qedhere
\end{enumerate}

\end{proof}

The following Lemma suggests essentially that for a state $q\in \SC(v)$ for some vertex $v\in V$, there is a finite word and a safe run from $q$ on it, yet the run of 
any state distinct from $q$ on the same word leads to the safe component $\SC_{\it Sync}$. It is crucial in showing that the run $r_z$ in $\A_d$ from which we read off the Hamiltonian path enters states that remember the correct past values with respect to the input read so far: 
\bigcomment{
The following Lemma suggests essentially that a state $q\in \SC(v)$ for some vertex $v\in V$ can force, while staying in $\SC(v)$, any distinct state to follow a run that leads to the safe component $\SC_{\it Sync}$ while reading the same word, and is crucial in showing that the run $r_z$ in $\A_d$ from which we read off the Hamiltonian path enters states that remember the correct past values with respect to the input read so far:
}%

\begin{lemma}\label{trap lemm}
        Consider a state $q\in \SC(v)$ for some vertex $v\in V$, and consider a state $s\neq q$ that  does not belong to $\SC_{\it Sync}$. 
        If there is a rejecting transition in $\A_d$ of the form $t = \zug{s^{-1}, \sigma, s}$, where $s^{-1}\in \SC(v)$ and either
        \begin{enumerate}
            \item 
            $v = v_1$, $v_1\in \sigma$  and the transition $t = \zug{s^{-1}, \sigma, s}$ is $\it extra$-consistent, or

            \item 
            $v\neq v_1$ and  $\{v_1, \it extra\} \subseteq \sigma$,
        \end{enumerate}
        then there is a nonempty finite word $w_q$ such that the run of $\A_d^q$ on $w_q$ is safe, yet the run of $\A_d^s$ on $w_q$ enters $\SC_{\it Sync}$ for the first time upon reading the last letter of $w_q$.
\end{lemma}

\begin{proof}
We first show the existence of a nonempty finite word $x_1$ such that every letter of $x_1$ contains $v_1$ and $\it extra$,  the run of $q$ on $x_1$ is safe, yet the run of $s$ on $x_1$ traverses a rejecting transition upon reading the last letter of $x_1$. Let $\eta(v) = \{ v, u_1, u_2, \ldots, u_m\}$ denote the neighborhood of $v$, and let
     $q = [p_1(u_1), p_1(u_2), \ldots, p_1(u_m), p_1(v), p_2(v), v]$. 
    Let $v'\in V$ be such that $s$ belongs to $\SC(v')$.
    We distinguish between two cases:
    \begin{itemize}
        \item $\SC(v) \neq \SC(v')$: as $\SC(v')$ is the only safe component that keeps track of the past-two value of $v'$, then the letter  $\sigma_1 = V\cup \{ \it extra\}\cup \{ \hat u_j: j\in [m] \ \wedge \ p_1(u_j) = 1 \} \cup \{ \check v: p_2(v) = 1\} \cup \{ \check v': p'_2(v') = 0\} $, where $p'_2(v')$ is the past-two value of $v'$ as encoded in $s$, is such that there is a safe transition from $q$ labeled with $\sigma_1$, yet no safe transition labeled with $\sigma_1$ from $s$. Also, $\{v_1, \it extra\} \subseteq \sigma_1$, and so in this case, we can choose $x_1 = \sigma_1$.

        \item $\SC(v) = \SC(v')$: 
        consider the letter $\sigma_1 = V\cup \{ \it extra\} \cup \{ \hat u_j: j \in [m] \wedge  p_1(u_j) = 1 \} \cup \{ \check v: p_2(v) = 1\}$, and note that the transition $q \xrightarrow{ \sigma_1} [1, 1, \ldots, 1, p_1(v), v]$ is safe. If there is no safe $\sigma_1$-transition from $s$, then we can choose $x_1 = \sigma_1$. Otherwise, if $s$ has an outgoing safe $\sigma_1$-transition, then since $s\neq q$, we have that $s$ must equal the vector  $[p_1(u_1), p_1(u_2), \ldots, p_1(u_m), p'_1(v), p_2(v),v]$ where $p_1(v) \neq p'_1(v)$. Thus, $q$ and $s$ differ only in the encoded past-one value of $v$. In this case, we have the safe transition $s \xrightarrow{\sigma_1} [1 , 1, \ldots, 1, p'_1(v), v]$.  Then, the word $x_1 = \sigma_1 \cdot (V\cup \{\it extra\} \cup \{ \hat u: u\in V\cap\sigma_1\} \cup \{ \check v: p_1(v) = 1 \}) $ satisfies the desired properties, namely,  every letter of $x_1$ contains $v_1$ and $\it extra$,  the run of $q$ on $x_1$ is safe, yet the run of $s$ on $x_1$ traverses a rejecting transition upon reading the last letter of $x_1$.
    \end{itemize}
    
    Now let $r_q = q_0 \xrightarrow{x_1} q_1$ and $r_s = s_0 \xrightarrow{x_1} s_1$ denote the runs of $q$ and $s$ on $x_1$, respectively. 
    If the run $r_s$ already entered $\SC_{\it Sync}$, then by choosing $w_q = x_1$, we're done.
    Otherwise, if $s_1\notin \SC_{\it Sync}$, then we claim that the case where $\SC(s_1) = \SC(v)$ is impossible.   Indeed, in this case,
    we can consider the run $r = s^{-1} \xrightarrow{\sigma} s \xrightarrow{x_1} s_1$    from $\SC(v)$ to $\SC(v)$. Then, recall that every letter of $x_1$ contains $v_1$ and $\it extra$, and note that either  $v =  v_1$, $v_1\in \sigma$, and the transition $\zug{s^{-1}, \sigma, s}$ is $\it extra$-consistent, or $v\neq v_1$ and the transition $\zug{s^{-1}, \sigma, s}$ is such that $\{v_1, \it extra \} \subseteq \sigma$. In both cases, Lemma~\ref{enters sync lemm} implies that the run $r$ visits $\SC_{\it Sync}$, contradicting the fact that $r$ only visits states in $\SC(s^{-1}) = \SC(v)$ and $\SC(s) = \SC(v')$.

     What we've shown so far is that if $r_s$ did not enter $\SC_{\it Sync}$, then $s_1 \in \SC(u)$ for some vertex $u \neq v$. In particular, since $q_1\in \SC(v)$, we have that $s_1\neq q_1$.
      Therefore, we can repeat the same considerations iteratively again from $q_1$ and $s_1$ to get two runs $r_q = q_0\xrightarrow{x_1}q_1\xrightarrow{x_2} q_2$ and $r_s = s_0\xrightarrow{x_1}s_1\xrightarrow{x_2} s_2$, where the former run is a safe run from $q$, and the latter run is from $s$ and traverses a rejecting transition only upon reading the last letter of every infix $x_i$.  In addition, every letter in the infixes $x_1$ and $x_2$ contains $v_1$ and $\it extra$.
      If $\SC(s_2)= \SC_{\it Sync}$, then we are done as we can choose $w_q =x_1\cdot x_2$.
      Otherwise, similarly to the above argument, it cannot be the case that $\SC(s_2) = \SC(v) $ as that implies that $r_s$ must have entered $\SC_{\it Sync}$ before.
    Finally, assume towards contradiction that a repetitive application of the previous considerations  results in infinite runs $r_q = q_0\xrightarrow{x_1}q_1\xrightarrow{x_2} q_2 \xrightarrow{x_3} \cdots $ and $r_s = s_0\xrightarrow{x_1}s_1\xrightarrow{x_2} s_2\xrightarrow{x_3} \cdots$
    where 
    all the  states $s_i$ do not belong to $\SC(v)$ nor to $\SC_{\it Sync}$; thus the run $r_s$  does not visit $  \SC_{\it Sync}$. 
    Then, as the state-space of $\A_d$ is finite, we get eventually that there are $i < j$ such that $s_i = s_j$. Hence,  the sub-run $s_{i} \xrightarrow{x_{i+1}} \cdots \xrightarrow{ x_j} s_j$ is a non-safe cycle from $s_i$ to $s_i$ over the word $w = x_{i+1}\cdot x_{i+2} \cdots x_j$. As $i<j$, then $w$ is nonempty. Now since the construction of the runs is such that every letter in $w$ contains $v_1$ and $\it extra$, we get that $w^\omega$ is $\it extra$-consistent, so by Lemma~\ref{rej loop lemm}, we have the latter cycle enters $\SC_{\it Sync}$, and we have reached a contradiction. 
\end{proof}

We are now ready to construct the word $z = z_1\cdot z_2 \cdots z_n$ and show that we can read off a Hamiltonain-path from the run $r_z$ of $\A_d^{(1, 1)}$ on it.
The following lemma forms the base case of our iterative construction of the word $z = z_1\cdot z_2 \cdots z_{n}$:

\begin{lemma}\label{base lemm}
    Consider the  finite word $z_1 = (\mathsf{AP})^4$. Then, for all vertices $v\neq v_1$, it holds that $z_1$ can be extended to an infinite word in $L_v$. In addition, the run $r$ of $\A_d^{(1, 1)}$ on $z_1$ leaves the safe component $\SC_{\it Sync}$ after reading the last letter of $z_1$,  and ends in the state $s_1 \in \SC(v_1)$ that remembers the correct past values of every vertex in $\eta(v_1)$ with respect to $z_1$. %
\end{lemma}

\begin{proof}
    First, since all atomic propositions belong to every letter in $z_1$, then it follows that for all vertices $v\neq v_1$, the word $z_1$ can be extended to an infinite word in $L_v$. Indeed, $z_1$ does not violate consistency of propositional values in the neighborhood $\eta(v)$.
Then, by the definition of the safe component $\SC_{\it Sync}$, we have that the run $r$ of $\A_d^{(1, 1)}$ on $z_1$ is of the from
$  r =   (1, 1)  \xrightarrow{\mathsf{AP}} (2, 1) \xrightarrow{\mathsf{AP}} (3, 1) \xrightarrow{\mathsf{AP}} (4, 1) \xrightarrow{\mathsf{AP}} r_4$.
Indeed, since $z_1$ does not violate the $\it extra$-consistency property, and  $(1, 1)$ remembers that the last value of $v_1$ is $1$ as $\mathsf{AP}(\it extra)$, then the run $r$ remains safe while reading the prefix $z_1[1, 3]$. Therefore, as  $z_1$ is of length 4 and does not contain the ${\it sleep}$-letter,  the run  $r$ traverses a rejecting transition for the first time when reading the last letter of $z_1$.
We claim that the state $r_4$ is a state in $\SC(v_1)$ that encodes the correct past values of every vertex in the neighborhood $\eta(v_1) = \{ v_1, v_1'\}$. Thus, $r_4 = [1, 1, 1, v_1]$. 
Specifically, in all other cases, we can point to a contradiction:
\begin{enumerate}
    \item 
    The state $r_4 = (t, k)$ is in $\SC_{\it Sync}$:
    consider the run $r'$ of $\A_d^{(1, 1)}$ on $z_1\cdot \it sleep$.
    The run 
    $ r' = (1, 1)  \xrightarrow{\mathsf{AP}} (2, 1) \xrightarrow{\mathsf{AP}} (3, 1) \xrightarrow{\mathsf{AP}} (4, 1) \xrightarrow{\mathsf{AP}} (t, k) \xrightarrow{\text{\it sleep}} (1, k)$ is a non-safe run from the state $(1, 1)$ to the state $(1, k)$. 
    We distinguish between two cases. Either $k=1$ or $k = 0$.
    If $k= 1$, we get that $r'$ is a non-safe cycle from $(1, 1)$ labeled with the word $x =  (\mathsf{AP})^4\cdot \text{\it sleep}$ in $\A_d$. 
    As $x^\omega$ is a mix of $(\it sleep)^\omega$ and the word $(\mathsf{AP})^\omega$, which is in $L_{v_1}$, we get that $x^\omega \in L(\A^{{G^{{v_1}}}})$. Hence,  as all the states of $\A_d$ are equivalent to $\A^{G^{v_1}}$, we conclude that  $ x^\omega \in L(\A^{(1 , 1)}_d)$, contradicting that the latter cycle is non-safe.
    Similarly, if $k = 0$, then we can modify the run $r'$ by changing the state $(1, 1)$ to $(1, 0)$ and the label of the first transition to $ \mathsf{AP}\setminus \{ \it extra\}$ to get the following non-safe cycle $ r'= (1, 0)  \xrightarrow{\mathsf{AP}\setminus \{ \it extra\}} (2, 1) \xrightarrow{\mathsf{AP}} (3, 1) \xrightarrow{\mathsf{AP}} (4, 1) \xrightarrow{\mathsf{AP}} (t, k) \xrightarrow{\text{\it sleep}} (1, k)$  from $(1, 0)$ labeled with the word $x = (\mathsf{AP}\setminus \{ \it extra\}) \cdot (\mathsf{AP})^3  \cdot \text{\it sleep}$ satisfying %
    that $x^\omega$ is a mix of $(\it sleep) ^\omega$ and the word $ ((\mathsf{AP}\setminus \{ \it extra\}) \cdot (\mathsf{AP})^3)^\omega\in L_{v_1}$. %
    Indeed, by changing the state $(1, 1)$ to $(1, 0)$, we only need to flip the value of $\it extra$ across the first transition so that the  modified run $r'$ coincides with the non-modified $r'$ after traversing its first safe transition. %
    
    \item 
    The state $r_4$ is a state in $\SC(v)$ for some vertex $v\neq v_1$: let $\eta(v) = \{ v, u_1, u_2, \ldots, u_m\}$ denote the neighborhood of $v$, let $r_4 = [p_1(u_1), p_1(u_2), \ldots, p_1(u_m), p_1(v), p_2(v), v]$, and 
    let $z_1 = \sigma_1\cdot \sigma_2 \cdot \sigma_3\cdot  \sigma_4 = (\mathsf{AP})^4$. 
    Now if $p_2(v) = 1$, then 
    we modify the letter $\sigma_3$ in $z_1$ to $\sigma_3  = \mathsf{ AP}\setminus \{ v\}$. Thus, we modified $\sigma_3$ so that the run $r$ enters a state $r_4$ that does not remember the correct past-two value of $v$. 
    Note that the modified run $r$ of $\A_d^{(1, 1)}$ on the modified word $z_1$ agrees with the non-modified run $r$ in the first three transitions although we modified the letter $\sigma_3$. Indeed, the values of vertices distinct from $v_1$ do not affect safe transitions taken inside $\SC_{\it Sync}$.
    Consider now extending the word $z_1$ and the modified run $r$ by reading the letter  %
    $\sigma_5 = V \cup \{ \it extra\} \cup \{\hat v': v' \in V \cap \sigma_4\} \cup \{\check v' \mid v' \in V\cap \sigma_3\}$ \color{black} from the state $r_4$, and note that while reading $\sigma_5$, the run  traverses a rejecting transition and moves to the state $r_5$. Indeed, the run $r$ we have so far is of the form %
    $$ r = (1, 1)  \xrightarrow{\mathsf{AP}} (2, 1) \xrightarrow{\mathsf{AP}} (3, 1) \xrightarrow{\sigma_3} (4, 1) \xrightarrow{\mathsf{AP}} [p_1(u_1), p_1(u_2), \ldots, p_1(u_m), p_1(v), p_2(v), v] $$ $$  \xrightarrow{V \cup \{\it extra\}\cup \{\hat v': v' \in V\cap \sigma_4\} \cup \{\check v'\ \mid v'\ \in V\cap \sigma_3\}} r_5$$
    where $p_2(v) \neq \sigma_3(v)$. %
    In fact, we claim that 
    $r_5 \notin \SC(v) = \SC(r_4)$. That is, upon reading the letter $\sigma_5$ from $r_4$, the safe component changes. 
    Indeed, otherwise, if $r_5\in \SC(v)$, then since $v\neq v_1$ and $v_1, \it extra\in \sigma_5$, then we can apply Lemma~\ref{enters sync lemm} on the rejecting transition 
    $t = \zug{r_4, \sigma_5, r_5}$ and the run $r_4 \xrightarrow{\sigma_5} r_5$ that starts and ends in $\SC(v)$, to obtain that $r_4 \xrightarrow{\sigma_5} r_5$ enters $\SC_{\it Sync}$, contradicting the fact that both $r_4$ and $r_5$ do not belong to $\SC_{\it Sync}$.

    Note that the modified $z_1 = \sigma_1\cdot \sigma_2 \cdots \sigma_5$ can be extended to an infinite word in $L_{v}$. Indeed, the language $L_{v}$ does not specify consistency in the past-one value of $v$ and so the letter $\sigma_4$ need not be modified. 
    To reach a contradiction, we show below that we can extend the modified $z_1$ to a finite word $x = z_1 \cdot w$ such that %
    there is a safe cycle in $\SC(v)$ labeled with $x$, 
    yet the run  $r_x$ of  $\A_d^{(1, 1)}$ on $x$ ends in the safe component $\SC_{\it Sync}$. %
    Then, we can proceed as in item 1 to point to a contradiction %
    by considering the finite word $x\cdot \it sleep$, and the run  of $\A_d^{(1 ,1)}$ on it. 
    Specifically, 
    the latter run either ends in $(1, 1)$ closing a non-safe cycle labeled with $x\cdot \it sleep$, contradicting the fact that $x\cdot\it sleep$ is labeled along some safe cycle in $\SC(v)$, or it ends in the state $(1, 0)$ in which case we can modify $r_0 = (1, 0)$ and $\sigma_1  = \mathsf{AP}\setminus \{\it extra\}$. Note that as $\SC(v)$ does not track $\it extra$ values, the latter modification does not affect safe cycles in $\SC(v)$ that are labeled with $x\cdot \it sleep$.

    To conclude the proof, we show how to extend the modified $z_1$ to a finite word $x = z_1 \cdot w$ with the desired properties. %
    Note that the state $ vec =  [1, 1, \ldots, 1, 1, v] \in \SC(v)$ has a safe run on the modified $z_1 = \sigma_1\cdot \sigma_2 \cdots \sigma_5$ that ends in some state $q\in \SC(v)$. 
    Indeed, by reading $ \sigma_1\cdot \sigma_2 = (\mathsf{AP})^2$ from $vec$ we stay in $vec$ and traverse only safe transitions. At this point, $vec$ stores the correct past values, and thus can trap the rest of the modified $z_1$ in $\SC(v)$. 
    Now either $r_5 \in \SC_{\it Sync}$, or $r_5 \in \SC(v')$ for some vertex $v'\neq v$. In 
    both cases, there is a finite word $w_q$ such that the run of $q$ on $w_q$ is safe and ends in some state $q'\in \SC(v)$, yet the run of $r_5$ on $w_q$ visits $\SC_{\it Sync}$ after reading the last letter of $w_q$. %
    Indeed, in the former case, we can take $w_q = \epsilon$, and in the latter case, 
    we can  apply Lemma~\ref{trap lemm} on the state $q$ and the rejecting transition $t = \zug{r_4, \sigma_5, r_5}$. Indeed, all the conditions of the Lemma are met, specifically, $q,r_4 \in \SC(v)$, $v\neq v_1$,  $\{v_1, \it extra \} \subseteq \sigma_5$,  $r_5\neq q$, and $r_5\notin \SC_{\it Sync}$.
    Now as $\A_d$ is normal, we can consider a finite safe run of the form $q'\xrightarrow{w'_q} vec$ where in addition
    $w'_q$ is such that the run of $r_5$ on $w_q\cdot w'_q$  stays in the safe component $\SC_{\it Sync}$ once it enters it. Note that the additional constraint on $w'_q$ can be fulfilled as every state in $\SC(v)$  has a safe loop labeled with $\it sleep$, and $\SC(v)$ does not track the $\it extra$ values. Thus, to enforce the run of $r_5$ on $w_q\cdot w'_q$ to stay in the safe component $\SC_{\it Sync}$, we can 
    consider injecting $\it sleep$ letters in $w'_q$ regularly and modify the $\it extra$ values in every non-$\it sleep$ letter in $w'_q$ as necessary. 
    So far, we have the run $vec \xrightarrow{z_1} q \xrightarrow{w_q} q' \xrightarrow{w'_q} vec$ and the run $(1, 1) \xrightarrow{z_1} r_5 \xrightarrow{w_q} s \xrightarrow{ w'_q} s'$, where the former run is a safe cycle inside $\SC(v)$, and the latter run  enters $\SC_{\it Sync}$ for the first time in the state $s$ and stays there.
    So by  taking $w = w_q\cdot w'_q$ and $x = z_1 \cdot w$, we're done.

    \item
    The state $r_4$ is a state in $\SC(v_1)$, yet does not encode the correct past values 
    with respect to $z_1$:  
we can follow the proof sketch of item 2 but with a careful modification to the $\it extra$ values in some letters  while (possibly) introducing an additional safe transition along the extended run $r$. Specifically, let $z_1 = \sigma_1 \cdot \sigma_2 \cdots \sigma_4  = (\mathsf{AP})^4$. Then,
    consider to not modify the letter $\sigma_3$, but only to extend $z_1$ and the run $r$ of $\A_d^{(1, 1)}$ on it by reading the letter $\sigma_5$ from the state $r_4$, where $\sigma_5$ is defined next.
    As the state $r_4$ does not remember that the past values of the vertices in $\eta(v_1) = \{ v_1, v'_1\}$ are all ones, %
    then it  remembers that the past-one value of $v_1$ is either 0 or 1. In the former case, we define $\sigma_5 = \mathsf{AP}\setminus \{ \it extra\}$, %
    and in the latter case, we define $\sigma_5 = \mathsf{AP}$. %
    Thus, $\sigma_5(\it extra)$ equals the past-one value of $v_1$ as encoded in $r_4$. In addition, $v_1\in \sigma_5$.
    So we have in total the following extended run
    $r =    (1, 1)  \xrightarrow{\mathsf{AP}} (2, 1) \xrightarrow{\mathsf{AP}} (3, 1) \xrightarrow{\mathsf{AP}} (4, 1) \xrightarrow{\mathsf{AP}} r_4 \xrightarrow{\sigma_5} r_5 $.
    The fact that $r_4\in \SC(v_1)$, and the definition of $\sigma_5$ imply that  the transition $\zug{r_4, \sigma_5, r_5}$ is $\it extra$-consistent.
    We argue next that we can assume w.l.o.g that the transition $ r_4 \xrightarrow{\sigma_5} r_5$ is rejecting. Indeed, otherwise, if the latter transition is safe, then as $r_4$ does not encode all the past values as ones,  as $\SC(v_1)$ %
    requires consistency only in the past-one value of $v'_1$ and the past-two value of $v_1$, and as $\sigma_5$ satisfies all proposition values in $\mathsf{AP}\setminus \{ \it extra\}$, it follows that $r_4$ encodes the past-one value of $v_1$ as $0$, and therefore $r_5 = [1, 1, 0, v_1]$ is the vector that assigns 1's to the past-one values of the two vertices in $\eta(v_1)$, and assigns 0 to the past-two value of $v_1$. %
    In this case, as $\A_d$ is normal,
    we can force a traversal of an $\it extra$-consistent rejecting transition by  defining the letter $\sigma_6 = \mathsf{AP}$, %
    and reading $\sigma_6$ from $r_5$.  
    Note that also here $  \sigma_6(\it extra) $ equals the past-one value of $v_1$ as encoded in $r_5$, and $v_1\in \sigma_6$.

    By the above, we may assume that the extended run $r =    (1, 1)  \xrightarrow{\mathsf{AP}} (2, 1) \xrightarrow{\mathsf{AP}} (3, 1) \xrightarrow{\mathsf{AP}} (4, 1) \xrightarrow{\mathsf{AP}} r_4 \xrightarrow{\sigma_5} r_5 $  
     traverses the rejecting $\it extra$-consistent transition $\zug{r_4, \sigma_5, r_5}$ moving to the state $r_5$.  
    Now to be able to carry out the proof as in item 2, we need to show that $r_5$ does not belong to $\SC(v_1)$, and this is immediate from Lemma~\ref{enters sync lemm}. Indeed, the rejecting transition $\zug{r_4, \sigma_5, r_5}$ is not only $\it extra$-consistent but also is such that $v_1\in \sigma_5$, and so if $r_5\in \SC(v_1)$, then the latter transition must enter $\SC_{\it Sync}$, contradicting the fact that it visits only states in $\SC(v_1)$.
    Note that also here, since $\SC(v_1)$ does not track $\it extra$ values, we have that $z_1 = \sigma_1\cdot \sigma_2 \cdots \sigma_5$ can be extended to an infinite word in $L_{v_1}$. Then, let $q$ be a state such that there is a safe run of the form $vec \xrightarrow{z_1} q$ inside $\SC(v_1)$.
    Now  if $r_5$ is not already in $\SC_{\it Sync}$, then as $v_1\in \sigma_5$,  we can apply also here Lemma~\ref{trap lemm} on the state $q$ and the rejecting $\it extra$-consistent transition $t = \zug{r_4, \sigma_5, r_5}$ to obtain a finite word $w_q$ where $q$ has a safe run on $w_q$ and the run of $r_5$ on $w_q$ enters $\SC_{\it Sync}$ upon reading the last letter of $w_q$. Hence, as $\SC(v_1)$ does not track $\it extra$ values, we can obtain, as in item 2, a finite 
    word $x$ labeled on some safe cycle inside $\SC(v_1)$, yet the run of $\A_d^{(1, 1)}$ on $x$ ends in $\SC_{\it Sync}$. \qedhere
\end{enumerate}
\end{proof}

The following lemma forms the inductive step of our iterative construction of the word $z = z_1\cdot z_2\cdots z_n$:

\begin{lemma}\label{step lemm}
    Consider $1\leq i\leq n$ and a word $z = z_1 \cdot  z_2\cdot z_3 \cdots z_i \in (2^{\mathsf{AP}})^*$ \color{black} starting with $z_1 = (\mathsf{AP})^4$, where  the run  $r_z = (1, 1)\xrightarrow{z_1} s_1 \xrightarrow{z_2} s_2 \cdots \xrightarrow{z_i} s_i$ of $\A_d^{(1, 1)}$ on $z$ is such that:

    \begin{enumerate}
       \item  $\SC(s_1) = \SC(v_1)$, \color{black}
       for all $j\in \{ 2, 3, \ldots, i\}$,  $\SC(s_j) = \SC(v_j)$ for some vertex $v_j$, and $p=(v_1, v_2, \ldots, v_i)$ is a simple path in $G$.

        \item 
        For all $j \in [i]$, the run $r_z$ enters $\SC(v_j)$ once upon reading the last letter of $z_j$ by moving to the  state $s_j$ that encodes the correct past values 
        of every vertex in $\eta(v_j)$
        with respect to  the prefix $z_1\cdot z_2\cdots z_j$. In addition, $r_z$ stays in the same safe component while reading the infix $z_j[1, |z_j|-1]$.

        \item 
        For every vertex $u \in V\setminus \{ v_1, v_2, \ldots, v_i\}$, the word $z$ can be extended to an infinite word in $L_u$. In addition, every letter in $z$ contains $v_1$ and $\it extra$.
    \end{enumerate}

    If $p$ is not a Hamiltonian-path, then 
    there is a word $z_{i+1}$ such that  every letter in $z_{i+1}$ contains $v_1$ and $\it extra$, and the run $r$ of $\A_d^{s_i}$ on $z_{i+1}$  leaves $\SC(s_i)$ and enters a state $s_{i+1}$ upon reading the last letter of $z_{i+1}$, where %
    $s_{i+1}$ is in $\SC(v_{i+1})$ for some vertex $v_{i+1}$ not visited by $p$,  $E(v_i, v_{i+1})$, and $z \cdot z_{i+1}$ can be extended to a word in $L_{v'}$ for all $v'\in V\setminus\{ v_1, v_2, \ldots, v_{i+1}\}$.
    In addition, $s_{i+1}$ encodes the correct past values %
    with respect to $z_1\cdot z_2\cdots z_i \cdot z_{i+1}$.
\end{lemma}

\begin{proof}
Assume that the simple path $p = (v_1, v_2, \ldots, v_i)$ is not Hamiltonian, and let $z=z_1 \cdot z_2\cdot z_3 \cdots z_i = \sigma_1\cdot \sigma_2 \cdots \sigma_l$.
    We define $z_{i+1}$ as the two-letter word $\sigma_{l+1} \cdot \sigma_{l+2}$, where  $\sigma_{l+1} =   V\cup \{ \it extra\} \cup \{ \hat v \mid v \in V \cap \sigma_l \} \cup \{ \check v \mid v \in V \cap \in \sigma_{l-1}\}$, and $\sigma_{l+2} = V \cup \{ \it extra\} \cup \{ \hat v \mid v \in V \cap \sigma_{l+1} \} \cup \{ \check v \mid v \in V, (v \neq v_i \wedge v \in \sigma_{l}) \vee ( v = v_i \wedge v \notin \sigma_{l})\}$. 
    As $s_i$ encodes the correct past values of vertices in $\eta(v_i)$ w.r.t $z$, then the transition $\zug{s_i, \sigma_{l+1}, s}$ is a safe transition inside $\SC(v_i)$; in particular, the state $s$ encodes the past-two value of $v_i$ as $\sigma_l (v_i)$. Therefore, by reading $\sigma_{l+2}$ from $s$, we traverse a rejecting transition $\zug{s, \sigma_{l+2}, s_{i+1}}$.
Note that it is still the case that for all $u\in V\setminus \{ v_1, v_2, \ldots ,v_i\}$, the word $z_1 \cdot z_2 \cdots z_{i} \cdot z_{i+1}$ can be extended to a word in $L_u$. Indeed, the safe component $\SC(u)$ does not track the past-two value of any vertex distinct from $u$. In addition, by definition, every letter in $z_{i+1}$ contains $v_1$ and $\it extra$.
So far, we have extended the run $r_z$ to the following run:
$$r_z =  (1, 1)\xrightarrow{z_1} s_1 \xrightarrow{z_2} s_2 \cdots \xrightarrow{z_i} s_i \xrightarrow{z_{i+1}} s_{i+1} = $$ $$ = (1, 1)\xrightarrow{z_1} s_1 \xrightarrow{z_2} s_2 \cdots \xrightarrow{z_i} s_i \xrightarrow{\sigma_{l+1}} s \xrightarrow{\sigma_{l+2}} s_{i+1}$$ where  the transition $\zug{s_i, \sigma_{l+1}, s}$ is safe, and the transition $\zug{s, \sigma_{l+2}, s_{i+1}}$ is rejecting.
To conclude the proof, we show that there exists a vertex $v_{i+1}\notin \{ v_1, v_2, \ldots, v_i\}$, and the state $s_{i+1} \in \SC(v_{i+1})$  encodes the correct past values %
with respect to $z_1\cdot z_2 \cdots z_i \cdot z_{i+1}$, and $E(v_i, v_{i+1})$.
We proceed with assuming by contradiction that the latter does not hold,  distinguish between cases, and point to a contradiction in each case:

\begin{enumerate}

      \item 
    There is a vertex $v_{i+1}\in V\setminus \{ v_1, v_2, \ldots, v_i\}$ such that $\neg E(v_i, v_{i+1})$ and $s_{i+1}\in \SC(v_{i+1})$:

    First, recall the extended run $r_z =   (1, 1)\xrightarrow{z_1} s_1 \xrightarrow{z_2} s_2 \cdots \xrightarrow{z_i} s_i \xrightarrow{\sigma_{l+1}} s \xrightarrow{\sigma_{l+2}} s_{i+1}$ 
 that is defined over the  word $z_1\cdot z_2 \cdots z_i \cdot z_{i+1} = \sigma_1\cdot \sigma_2 \cdots \cdots \sigma_l \cdot \sigma_{l+1} \cdot \sigma_{l+2} $ starting with $z_1=(\mathsf{AP})^4$. 
 As $\neg E(v_i , v_{i+1})$, then the safe component $ \SC(s_i) = \SC(v_i)$ does not track past values of $v_{i+1}$.
Therefore,  modifying the belonging of $v_{i+1}$ in the letter $\sigma_{l+1}$ to the opposite of the past-two value of $v_{i+1}$ as encoded in the state $s_{i+1}$ is such that the run $r_z$ on the modified word agrees with the original run on its form, and enters a state $s_{i+1}$ that does not remember the correct past-two value of $v_{i+1}$.
    Indeed, since $s_i \xrightarrow{\sigma_{l+1}} s$ is a safe transition inside $\SC(v_i)$, and $\SC(v_i)$ does not track the past values of $v_{i+1}$, then the modified run $r_z$ on the modified word $z_1 \cdot z_2 \cdots z_{i+1}$ agrees with the non-modified run $r_z$ on the safe transition $s_i \xrightarrow{\sigma_{l+1}} s$. 
    Consider extending the modified run $r_z$ by reading, from the state $s_{i+1}$, the letter $\sigma_{l+3} = V \cup \{\it extra\} \cup \{\hat v: v \in V\cap \sigma_{l+2}\} \cup \{\check v \mid v \in V\cap \sigma_{l+1}\}$  for the modified $\sigma_{l+1}$, and note that the run $r_z$ traverses a rejecting transition $\zug{s_{i+1}, \sigma_{l+3}, s_{i+2}}$. 
    So far, we have the following run $$ r_z =   (1, 1)\xrightarrow{z_1} s_1 \xrightarrow{z_2} s_2 \cdots \xrightarrow{z_i} s_i \xrightarrow{\sigma_{l+1}} s \xrightarrow{\sigma_{l+2}} s_{i+1} \xrightarrow{\sigma_{l+3}} s_{i+2}$$
    
    Recall that the non-modified word $z_1\cdot z_2 \cdots z_i \cdot z_{i+1}$ can be extended to a word in $L_{v_{i+1}}$. 
    Then, since 
    the language $L_{v_{i+1}}$ does not specify consistency in the past-one value of $v_{i+1}$, there is no need to modify the letter $\sigma_{l+2}$, and so
    it is still the case that the  word $z_1\cdot z_2 \cdots z_i \cdot \sigma_{l+1} \cdot \sigma_{l+2}\cdot \sigma_{l+3}$ for the modified $\sigma_{l+1}$ can be extended to a word in $L_{v_{i+1}}$.  
    In addition, the rejecting transition $\zug{s_{i+1}, \sigma_{l+3}, s_{i+2}}$ is such that $\{v_1, \it extra\} \subseteq \sigma_{l+3}$.
    Therefore, since $v_{i+1}\neq v_1$, then as in item 2 of Lemma~\ref{base lemm}, we can reach a contradiction using Lemmas~\ref{enters sync lemm} and \ref{trap lemm} to
     show that we can extend the modified word $z' = z_1\cdot z_2 \cdots z_i \cdot \sigma_{l+1} \cdot \sigma_{l+2}\cdot \sigma_{l+3}$ to a finite word $x= z'\cdot w$ such that $x$ is labeled on some safe cycle in $\SC(v_{i+1})$, yet the run $r_x$ of $\A^{(1, 1)}_d$ on $x$ ends in the safe component $\SC_{\it Sync}$.

    \item 
    There is a vertex $v_{i+1}\in V\setminus \{ v_1, v_2, \ldots, v_i\}$ such that $E(v_i, v_{i+1})$ and $s_{i+1}\in \SC(v_{i+1})$, yet the state $s_{i+1}$ does not encode the correct past values %
    with respect to $z_1\cdot z_2 \cdots z_{i} \cdot z_{i+1}$: 

    First, recall the extended run $r_z =   (1, 1)\xrightarrow{z_1} s_1 \xrightarrow{z_2} s_2 \cdots \xrightarrow{z_i} s_i \xrightarrow{\sigma_{l+1}} s \xrightarrow{\sigma_{l+2}} s_{i+1}$. 
    We define the letter $\sigma_{l+3} = V \cup \{ \it extra\}
    \cup\{\hat v: v \in V\cap \sigma_{l+2}\} \cup \{\check v \mid v \in V\cap\sigma_{l+1}\} = V\cup \{ \it extra\} \cup \{ \hat v, \check v: v\in V\}$, and note that $\{v_1, \it extra\} \subseteq \sigma_{l+3}$. %
    Then, consider the following extended run:
    $$r_z =   (1, 1)\xrightarrow{z_1} s_1 \xrightarrow{z_2} s_2 \cdots \xrightarrow{z_i} s_i \xrightarrow{\sigma_{l+1}} s\xrightarrow{\sigma_{l+2}} s_{i+1} \xrightarrow{\sigma_{l+3}} s_{i+2}$$ 
    
    We argue that we can assume w.l.o.g that the transition $ s_{i+1} \xrightarrow{\sigma_{l+3}} s_{i+2}$ is rejecting. Indeed, otherwise, if the latter transition is safe, then since $\sigma_{l+3}(\hat v) = \sigma_{l+3}(\check v_{i+1}) =  1$ for every vertex $v$ that is a neighbor of $v_{i+1}$, %
    and since $s_{i+1}$ does not encode the correct past values across $\eta(v_{i+1})$ as 1's, it follows that $s_{i+1}$ encodes all past values as 1's except for the past-one value of $v_{i+1}$ which is encoded as $0$, and therefore $s_{i+2} = [1, 1, 1, \ldots, 0, v_{i+1}]$ is the vector that assigns 1's to the past-one values of every vertex in $\eta(v_{i+1})$, and assigns 0 to the past-two value of $v_{i+1}$.  Indeed, $\sigma_{l+3}$ satisfies all atomic propositions in $V$.
    In this case, as $\A_d$ is normal,
    we can force a traversal of a rejecting transition by  defining the letter $\sigma_{l+4} = V\cup \{\it extra\} \cup \{\hat v: v \in V\cap\sigma_{l+3}\} \cup \{\check v \mid v \in V\cap \sigma_{l+2}\}$, and reading $\sigma_{l+4}$ from the state $s_{i+2}$. Indeed, $\sigma_{l+4}(\check v_{i+1}) = 1$. %
    Note that also the letter $\sigma_{l+4}$ is such that $  v_1, \it extra \in \sigma_{l+4}$.

    By the above, we may assume that the extended run $r_z =   (1, 1)\xrightarrow{z_1} s_1 \xrightarrow{z_2} s_2 \cdots \xrightarrow{z_i} s_i \xrightarrow{\sigma_{l+1}} s\xrightarrow{\sigma_{l+2}} s_{i+1} \xrightarrow{\sigma_{l+3}} s_{i+2}$   
     traverses the rejecting transition $s_{i+1} \xrightarrow{\sigma_{l+3}} s_{i+2}$ moving to the state $s_{i+2}$. %
     In addition,  the definition of $\sigma_{l+3}$ is such that  $z' = z_1 \cdot z_2 \cdots z_{i} \cdot z_{i+1} \cdot \sigma_{l+3}$ can still be extended to a word in $L_{v_{i+1}}$. 
     Therefore, since $v_{i+1}\neq v_1$ and $\SC(v_{i+1})$ does not track $\it extra$ values, 
    then, as in item 2 of Lemma~\ref{base lemm}, we can reach a contradiction using Lemmas~\ref{enters sync lemm} and \ref{trap lemm} to
     show that we can extend $z'$ to a finite word $x= z'\cdot w$ such that $x$ is labeled on some safe cycle in $\SC(v_{i+1})$, yet the run $r_x$ of $\A^{(1, 1)}_d$ on $x$ ends in the safe component $\SC_{\it Sync}$.

       \item 
    The state $s_{i+1}$ belongs to the safe component $\SC_{\it Sync}$: 
    Since the simple path $p = (v_1, v_2, \ldots, v_i)$ is not Hamiltonian, then there is a vertex $v \notin \{ v_1, v_2, \ldots, v_i\}$.  
    Consider the extended run $r_z =   (1, 1)\xrightarrow{z_1} s_1 \xrightarrow{z_2} s_2 \cdots \xrightarrow{z_i} s_i \xrightarrow{z_{i+1}} s_{i+1}$, and recall that $z_1 = (\mathsf{AP})^4$, and  $z' = z_1\cdot z_2 \cdots z_i \cdots z_{i+1}$ can be extended to an infinite word in $L_v$.
    Then, let $vec \xrightarrow{z'} q$ be a safe run inside $\SC(v)$, and consider extending the latter run to a safe cycle $c$ in $\SC(v)$ by appending a safe run of the form $q \xrightarrow{w} vec$ where $w$ ends with $\it sleep$,  $w$ contains $\it sleep$ regularly, and $w$ is such that the run of $\A_d^{s_{i+1}}$ on it does not leave $\SC_{\it Sync}$. Note that $w$ exists as every state in $\SC(v)$ has a safe-loop labeled with $\it sleep$, and thus we can inject $\it sleep$ along $w$ regularly. In addition, since we can modify the $\it extra$ values across  safe transitions in $\SC(v)$ without affecting their safety, then we can make sure to modify the $\it extra$ values in $w$ so that the run of $\A_d^{s_{i+1}}$ on $w$ stays in $\SC_{\it Sync}$.

    So far we have a safe cycle $vec \xrightarrow{z'} q \xrightarrow{w} vec$ in $\SC(v)$, yet the run $r_{z'\cdot w} =  (1, 1)\xrightarrow{z_1} s_1 \xrightarrow{z_2} s_2 \cdots \xrightarrow{z_i} s_i \xrightarrow{z_{i+1}} s_{i+1} \xrightarrow{w} s$ of $(1,1)$ on $z'\cdot w$ is a non-safe run that ends in a state $s$ in $\SC_{\it Sync}$.  
    We can now reach a contradiction similarly to item 1 of Lemma~\ref{base lemm}. Specifically, since $w$ ends with $\it sleep$, then $s = (1, k)$.  Now if $k=1$, then $r_{z'\cdot w}$ is a non-safe cycle labeled with $z'\cdot w$, and we have reached a contradiction to the fact that $z' \cdot w$ is labeled on some safe cycle in $\SC(v)$. Then, if $k=0$,  we can modify the first state in $r_{z'\cdot w}$ to $(1, 0)$, and modify its labeled word $z' \cdot w$ to $ (\mathsf{AP}\setminus \{ \it extra\}) \cdot (\mathsf{AP})^3 \cdot z_2\cdot z_3 \cdots z_i \cdot z_{i+1} \cdot w$ to get that the modified run $r_{z'\cdot w}$ is a non-safe cycle from $(1, 0)$ labeled with the modified $z'\cdot w$. Note that since $\SC(v)$ does not track $\it extra$-values, then it is still the case that $vec \xrightarrow{z'} q \xrightarrow{w} vec$ is a safe run inside $\SC(v)$ for the modified $z'\cdot w$,  and we have reached a contradiction also in this case.

       \item 
    The state $s_{i+1}$ belongs to the safe component $\SC(s_j)$ for some $j\in [i]$: 
    the run $r_z =  (1, 1)\xrightarrow{z_1} s_1 \xrightarrow{z_2} s_2 \cdots \xrightarrow{z_i} s_i \xrightarrow{z_{i+1}} s_{i+1}$ is such that $s_j \xrightarrow{z_{j+1}} \cdots \xrightarrow{z_i} s_i \xrightarrow{z_{i+1}} s_{i+1}$ is a non-safe run that does not visit $\SC_{\it Sync}$, and starts and ends in the same safe component. To reach a contradiction, we show that the latter run can be extended to a non-safe  cycle $c = s_j \xrightarrow{z_{j+1}} \cdots \xrightarrow{z_i} s_i \xrightarrow{z_{i+1}} s_{i+1} \xrightarrow{y} s_j$  from $s_j$ to $s_j$, that does not visit $\SC_{\it Sync}$, by  appending a safe run of the form $r_{\it safe} = s_{i+1} \xrightarrow{y} s_j$, where every letter in $y$ contains $v_1$ and $\it extra$. Indeed, this is sufficient as by the construction of $z_1\cdot z_2\cdots z_{i+1}$, the  non-safe cycle $c$ is labeled with a word $z_{j+1} \cdot z_{j+2} \cdots z_{i+1} \cdot y$ that contains $v_1$ and $\it extra$ in every letter, and thus its omega is $\it extra$-consistent, contradicting  Lemma~\ref{rej loop lemm} which states that $c$ must enter $\SC_{\it Sync}$.
    
    It is left to show that the run $r_{\it safe}$ exists. 
If the safe component $\SC(s_j)$ is distinct from $ \SC(v_1)$, then as $v_1$ has no incoming edges in $G$, we have that $\SC(s_j)$  does not track the past values of $v_1$ and $\it extra$ and so we can assume w.l.o.g that every safe path inside $\SC(s_j)$ is labeled with a word containing $v_1$ and $\it extra$ in every letter. 
Then,  if $\SC(s_{j}) = \SC(v_1)$,  let $v_1'$ denote the single neighbor of $v_1$ in $G$. In this case, since the path $p = (v_1, v_2, \ldots, v_i)$ is simple, we have that $s_j = s_1$.
Let $s_{i+1} = [p_1(v_1'), p_1(v_1), p_2(v_1), v_1]$.
As $z_1 = (\mathsf{AP})^4$, then we have by Lemma~\ref{base lemm} that $s_j = s_1 = [1, 1, 1, v_1]$. 
Therefore, we can consider the safe run
\begin{multline*}
r_{\it safe} = \underbrace{[p_1(v_1'), p_1(v_1), p_2(v_1), v_1]}_{s_{i+1}} \xrightarrow{V\cup \{ \it extra\} \cup \{ \hat v'_1:\  p_1(v_1') = 1\}  \cup \{\check v_1: \ p_2(v_1)=1  \}} \\ [1, 1, p_1(v_1), v_1] \xrightarrow{V \cup \{ \it extra\} \cup \{ \hat v'_1 \} \cup \{\check v_1 : \ p_1(v_1) = 1\}} \underbrace{[1, 1, 1, v_1]}_{s_1 }
\end{multline*} 
 all of whose letters contain $v_1$ and $\it extra$.
\qedhere
\end{enumerate}

\end{proof}

By Lemmas~\ref{base lemm} and \ref{step lemm}, we can conclude the other direction of the reduction: 

\begin{proposition}
    If the HD-tNCW $\A^{G^{v_1}} = \zug{\Sigma, Q, q_0, \delta, \alpha}$ has an equivalent tDCW of size $|Q|$, then 
    the graph $G = \zug{V, E}$ has a Hamiltonian path starting at $v_1$. 
\end{proposition}

\section{The HD-negativity Problem is NP-hard: Encoding The Reduction From Section~\ref{exp red sec}}
\label{hd neg sec}

Consider the reduction defined in Section~\ref{exp red sec}, that for a given directed graph $G =  \zug{V, E}$ with a designated vertex $v_1$ satisfying Def.~\ref{def:GraphProperties}, returns the nice minimal $\alpha$-homogeneous HD-tNCW $\A^{G^{v_1}} = \zug{2^\mathsf{AP} \cup \{\it sleep\}, Q, q_0, \delta, \alpha}$. %
In this section, we show how this reduction implies NP-hardness of the HD-negativity problem.
This is involved, however, as the automaton $\A^{G^{v_1}}$ has an
alphabet whose size is exponential in $G$; in particular, 
 an explicit representation of $\A^{G^{v_1}}$ 
is not computable in polynomial time and storing it requires exponential space.
We fix this issue by encoding $\A^{G^{v_1}}$ into a nice minimal HD-tNCW $e(\A^{G^{v_1}})$ whose description is at most polynomial in the input graph and that can be computed in polynomial time from $G$ when applying it ``on-the-fly'' without ever explicitly building $\A^{G^{v_1}}$.

The starting point of the encoding $e(\A^{G^{v_1}})$ is to encode letters in $2^{\mathsf{AP}}$ as binary numbers.  Specifically, let $k = |\mathsf{AP}|$ denote the number of the atomic propositions in $\mathsf{AP} = \{ a_1, a_2, \ldots, a_k\}$. Then, we encode each letter $\sigma$ in $2^{\mathsf{AP}}$ as a binary number  $\zug{\sigma}$ of length $k$, where the $i$'th bit of $\zug{\sigma}$ is 1 if and only if $a_i\in \sigma$. It is straight-forward to see that the encoding defines a bijection between the binary numbers of length $k$, namely elements in $\{ 0, 1\}^k$, and the letters in $2^{\mathsf{AP}}$. 
For the letter $\it {sleep}$, we define $\zug{\it{sleep}} = \it{sleep}$; that is, the $\it sleep$ letter remains as is.
For a letter $\sigma\in 2^{\mathsf{AP}} \cup \{\it {sleep}\}$, we let $dec(\zug{\sigma})$ denote the decoding of $\zug{\sigma}$, that is, $dec(\zug{\sigma}) = \sigma$. 
For a word $w = \sigma_1 \cdot \sigma_2 \cdot \sigma_3 \cdots \in (2^{\mathsf{AP}} \cup \{\it {sleep}\})^\omega$, we let $\zug{w}$ denote the infinite word obtained by encoding every letter of $w$; thus, $ \zug{w} = \zug{\sigma_1}  \cdot \zug{\sigma_2} \cdot \zug{\sigma_3} \cdots $ and $dec(\zug{w}) = dec(\zug{\sigma_1})  \cdot dec(\zug{\sigma_2}) \cdot dec(\zug{\sigma_3}) \cdots  = w$.

The idea is to let $\A^{G^{v_1}}$  read the encoding $\zug{\sigma}$ instead of the letter $\sigma$. However, replacing every letter $\sigma$ with its encoding involves two technical difficulties: (1) by encoding transitions from a state $q\in Q$, the state $q$ should be able to read all encodings of letters in $2^{\mathsf{AP}}$. Therefore, in the worst case, we need to add $2^k$ internal states that we visit while reading an encoding of a letter. Also, (2) while  reading an encoding $\zug{\sigma}$ from a state $q\in Q$, we might leave $\SC(q)$ right after reading a proper prefix of the encoding, and thus lose the correspondence between the safe components of the original automaton and those of the encoded one, and in fact this makes the encoding trickier since it is not clear why the resulting automaton should be HD --  how does an HD-strategy of the encoded automaton ``know'' which successor state to move to before actually reading the whole encoding? The former challenge is resolved by using the fact that  the input graph satisfies Def.~\ref{def:GraphProperties}, allowing us to reduce the state-space of the automaton by merging internal states, and the latter challenge is resolved by adding safe cycles that loop  back to $q$ via a new letter $\$$ that is special for the respective current internal state. %
Thus, we force leaving a safe component to occur only after reading the last bit of the encoding. In particular, the decision to which $\sigma$-successor $s$ of $q$ to move to happens once we have read the whole encoding of a letter $\sigma$. 
In a sense, by allowing the automaton $e(\A^{G^{v_1}})$ to leave the current safe component only after reading the whole character encoding, we maintain the correspondence between  safe components of $\A^{G^{v_1}}$ and $e(\A^{G^{v_1}})$, and prevent any ambiguity in deciding to which successor state to proceed to.
In addition, the added safe cycles enforce the states that are reachable from $q$ upon reading an encoding of a letter in $2^{\mathsf{AP}}$, except for the last bit, to be unique for $q$. 
In fact, the way we define these safe cycles guarantees that the encoding $e(\A^{G^{v_1}})$ is already a nice minimal HD-tNCW.

Consider $\A^{G^{v_1}} = \zug{2^{\mathsf{AP}} \cup \{\textit{sleep}\}, Q, q_0, \delta, \alpha}$. Before describing the construction of $e(\A^{G^{v_1}})$, we proceed with some useful notation. 
Recall the proof of the first direction in Theorem~\ref{corr thm}. There, we showed how a Hamiltonian path $p$ starting at $v_1$ in $G$ induces a deterministic pruning $\A_d$ of $\A^{G^{v_1}}$ that recognizes the same language. Essentially, for a state $q\in Q$ and letter $\sigma \in 2^{\mathsf{AP}} \cup \{ \it sleep\}$ where all $\sigma$-transitions from $q$ are rejecting, we pruned out all  $\sigma$ transitions from $q$ except for one rejecting transition $t = \zug{q, \sigma, s}$. Similarly, to the definition of safe transitions in $\A^{G^{v_1}}$, the destination $s$ of the rejecting transition $t$ was induced deterministically based on the state $q$ and the values of some atomic propositions under the assignment $\sigma$. 
The goal is to define $e(\A^{G^{v_1}})$ in a way that deterministic prunings in it correspond to those of $\A^{G^{v_1}}$, and thus $e(\A^{G^{v_1}})$ should be able to track the proposition values of relevance in each safe component, motivating the following definition.
For a state $q\in Q$, and atomic proposition $a_i$ in $\mathsf{AP}$, we say that $a_i$ {\em is essential} for $q$ if one of the following holds:
\begin{enumerate}
    \item 
    $q \in \SC(v)$ for some vertex $v\in V$ and $a_i = v'$ for some $v'$ vertex with $d(v, v') \leq 2$, or $a_i = \hat v'$ for some vertex $v'$ with $d(v, v') = 1$, or $a_i = \check v$.
    Thus, $a_i$ corresponds to the values of vertices of distance at most $2$ from $v$, or the hat values of neighbors of $v$, or the check value of $v$.
     
    \item 
    $q \in \SC(v)$ for some vertex $v\neq v_1$ and $a_i = v_1$.

    \item 
    $q\in \SC_{\it Sync}$ and  $a_i = \it {extra}$, or $a_i\in V$ is a vertex with $d(v_1, a_i)\leq 1$.  
    Thus, $a_i$ corresponds to the values of vertices in the neighborhood of $v_1$, or the value of $\it extra$.
\end{enumerate}

\begin{remark}\label{essen rem}
    Consider a state $q\in Q$, and an atomic proposition $a_i\in \mathsf{AP}$. When $a_i$ is essential to $q$, then the assignment of $a_i$ is essential not only for determining safe transitions inside $\A^{G^{v_1}}$, but also for determining rejecting transitions going out from $q$ in a deterministic pruning $\A_d$ of $\A^{G^{v_1}}$ that recognizes $L(\A^{G^{v_1}})$ (if it exists). Thus, the definition of essential atomic propositions considers the cases of $\A_d$ switching safe components. 
    Specifically,
    item 1 in the definition specifies that in the worst case, 
    when switching from $\SC(v)$ to $\SC(u)$ where $E(v, u)$, 
    we need to consider the values of vertices that are neighbors of $u$ as essential since $\A_d$ needs to move to a state storing these values. Then, for $\A_d$ to transition to the correct state in $\SC_{\it Sync}$ from $\SC(\neq v_1)$, item 2 suggests that we need to also consider the value of $v_1$ as essential, and finally, the transition from $\SC_{\it Sync}$ to $\SC(v_1)$ is based also on the values of vertices in $\eta(v_1)$.
\end{remark}

\color{black}
We let $ess(q)\subseteq \mathsf{AP}$ denote the set of atomic propositions that are essential for the state $q$.
Theorem~\ref{HP is hard thm} suggests we can assume that every vertex of $G$ has a neighborhood of constant size; in particular, the number of vertices of distance at most $2$ from any vertex of $G$ is constant as well.
Thus, we have the following:

\begin{lemma}\label{sensing lemm}
    For every state $q\in Q$ of $\A^{G^{v_1}}$, it holds that the set $ess(q)$ of atomic propositions essential to $q$ is of constant size, and can be computed in polynomial time from the input graph $G$.
\end{lemma}

We now describe the encoding of $\A^{G^{v_1}} = \zug{2^{\mathsf{AP}} \cup \{\it{sleep}\}, Q, q_0, \delta, \alpha}$ formally. Let $e(\A^{G^{v_1}}) = \zug{\{ 0, 1, \it {sleep}\} \cup \Sigma_\$, Q', q_0, \delta', \alpha'}$ denote the encoding of $\A^{G^{v_1}}$, to be defined below. 
The set of states $Q'$ in $e(\A^{G^{v_1}})$ includes all the states in $Q$; thus $Q\subseteq Q'$. In particular, the initial state $q_0$ of $\A^{G^{v_1}}$ is in $Q'$. 
For every state $s\in Q$,  we let $s_\epsilon$ be an alias for the state $s$. %
Recall that $\mathsf{AP} = \{ a_1, a_2, \ldots, a_k\}$.
We add to $e(\A^{G^{v_1}})$ new internal states, denoted $s_x$, corresponding to every state $s\in Q$ and nonempty binary word $x$ of length at most $k-1$  (see Example~\ref{exm 2}). 
The internal state $s_x$ is defined iteratively so that it is reachable deterministically from $s$ via a safe path upon reading the binary word $x$. Essentially, internal states corresponding to the state $s$ encode transitions from $s$ in $\A^{G^{v_1}}$. 
For starters, recall that $s_\epsilon = s$.  Assume that an internal state of the form $s_{x'}$ is already defined where $x'$ is a binary number of length at most $k-2$, and we define the $(0|1)$-successors of $s_{x'}$ as follows. If the atomic proposition $a_{|x'|+1}$ is essential for $s$, then we add distinct internal states $s_{x'\cdot 0}$ and $s_{x'\cdot 1}$ that are reachable only from $s_{x'}$ upon reading $0$ and $1$ via safe transitions, respectively. Otherwise, if $a_{|x'|+1}$ is not essential for $s$, then, we do not branch to distinct internal states from $s_{x'}$ upon reading $0$ and $1$, and instead, we proceed via a safe transition from $s_{x'}$ to a new internal state, to which we refer simultaneously as $s_{x'\cdot 0}$ and $s_{x'\cdot 1}$, upon reading $0$ or $1$.
We note that if $s_{x'} = s_{y'}$ for some binary number $y'\neq x'$, then it
is not hard to see that in this case $|x'| = |y'|$.
In addition, the definition of the $(0|1)$-successors of the internal state $s_{x'}$ depends only on the length of $x'$ and not on the actual bits appearing in $x'$. Thus, transitions from $s_{x'}$ are well-defined, and we refer to the state $s_{x'\cdot b}$ also as $s_{y'\cdot b}$ for all bits $b\in \{0, 1\}$.
Finally, from a state $s_x$, where $x$ is a binary number of length $k-1$, and a bit $b\in \{ 0, 1\}$, if $\zug{s, dec(x\cdot b), p}$ is a safe transition of $\A^{G^{v_1}}$, then we define the safe transition $\delta'(s_x, b) =  \{p\}$ in $e(\A^{G^{v_1}})$. Otherwise, if all $dec(x\cdot b)$-transitions from $s$ in $\A^{G^{v_1}}$ are rejecting, then we define $\delta'(s_x, b) = Q$, and all these transitions are rejecting in $e(\A^{G^{v_1}})$.
Note that for two binary words $x$ and $y$ of length $k-1$ with $s_x = s_y$, we have the following: for every bit $b\in \{ 0, 1\}$, it holds that $\zug{s, dec(x\cdot b), p}$ is a safe transition of $\A^{G^{v_1}}$ iff $\zug{s, dec(y\cdot b), p}$ is a safe transition of $\A^{G^{v_1}}$. Thus, transitions from internal states $s_x$ with  $|x| = k-1$ are well-defined as well. Indeed, $e(\A^{G^{v_1}})$ is defined so that distinct binary words of the same length $x$ and $y$ lead from $s$ to the same internal state only when $x$ and $y$ disagree on bits that correspond to atomic propositions that are not essential for $s$, and safe $\sigma$-transitions from $s$ in $\A^{G^{v_1}}$ are determined based only on $s$ and proposition values essential to it as specified by $\sigma$.

Observe that not only the size of $e(\A^{G^{v_1}})$ is polynomial in $G$,  thanks to Lemma~\ref{sensing lemm}, but also the above description can be implemented iteratively from $Q$ and $G$ (without ever constructing the transition function of $\A^{G^{v_1}}$ explicitly) to compute, in polynomial time, the internal states corresponding to every state $s\in Q$, and transitions from them.  
For example, for a state $s = [p_1(u_1), p_1(u_2), \ldots, p_1(u_m), p_1(v), p_2(v), v]\in \SC(v)$ for some vertex $v\in V$, and a letter $\sigma\in 2^{\mathsf{AP}}$, to check whether $\sigma$ is consistent with the state $s$ and determine the $\sigma$-successor $p$ of $s$ when $\zug{s, \sigma, p}$ is a safe transition in $\A^{G^{v_1}}$, it is sufficient to consider the values of atomic propositions that are essential to $s$ as specified by the encoding $\zug{\sigma}$. Specifically, checking consistency with $s$ amounts to checking whether for all $j\in [m]$, $\sigma(\hat u_j) = p_1(u_j)$, and  whether $\sigma(\check v) = p_2(v)$, and when the consistency check passes, we move to the state $p = [\sigma(u_1), \sigma(u_2), \ldots, \sigma(u_m), \sigma(v), p_1(v), v ]$. Note that indeed both the consistency check and the state $p$ are based only $s$ and values of atomic propositions in $ess(s)$.

Next, $\it {sleep}$-transitions of $e(\A^{G^{v_1}})$ are inherited from $\A^{G^{v_1}}$ in the expected way: for every state $s\in Q$, we define $\delta'(s, \it {sleep}) = \delta(s, \it {sleep})$, and
a transition $\zug{s, \it {sleep}, p}$ of $e(\A^{G^{v_1}})$ is rejecting only when it is rejecting in $\A^{G^{v_1}}$.
From states in $Q'\setminus Q$, reading $\it {sleep}$ leads via a rejecting transition to a rejecting sink $q_{rej}$, which we also add to $Q'$. 
We end the construction of $e(\A^{G^{v_1}})$ by defining safe-cycles that loop back to every state in $Q$ without increasing the state-space of $e(\A^{G^{v_1}})$ and while adding only polynomially many special letters to $\Sigma_\$$:
we first initialize $\Sigma_\$ = \emptyset$. Then, 
for every state $d\in Q'\setminus \{ q_{rej}\}$, we add a special letter $\$_d$ to $\Sigma_\$$, and define $\$_d$-labeled transitions  as follows. 
Let $s\in Q$ be the state that $d$ corresponds to, that is, $d = s_x$, for some binary word $x$ of length at most $k-1$.
We add the safe transition $\zug{d, \$_d, s}$, and for all states $d'\neq d$ in $Q'\setminus\{q_{rej}\}$, we add the rejecting transitions
$\delta'(d', \$_d) = Q$.
Finally, all the $\$_d$-transitions from $q_{rej}$ are rejecting transitions leading to $q_{rej}$.
In words, $q_{rej}$ remains a rejecting sink,  
 the only safe $\$_d$-transition is from the state $d = s_x$ to the state $s$ that $d$ corresponds to, and  all other $\$_d$-transitions lead nondeterministically via rejecting transitions to every state in $Q$.

\begin{example}\label{exm 2}
{\rm  Figure~\ref{B example} below describes transitions from a state $s\in Q$ of $e(\A^{G^{v_1}})$.
The state $p$ is also in $Q$, and we drew a node %
labeled with $Q$ representing all states in $Q$; thus transitions to it read as transitions to all states in $Q$.
Recall that dashed transitions are rejecting. In this example, we assume for simplicity that $\mathsf{AP} = \{a_1, a_2, a_3, a_4 \}$. Then, the set of atomic propositions essential to $s$ is given by $ess(s) = \{ 1, 3\}$. Accordingly, we branch to distinct internal states only from states of the form $s_x$ where $|x| \in \{ 0, 2\}$; thus, from the states $s_{\epsilon}, s_{00}$, and $s_{10}$. 
Transitions from $p$ and $Q$ as well as $\it {sleep}$-transitions are omitted.
Also, missing $\Sigma_\$$-transitions are rejecting and lead to $Q$. Finally, for simplicity, we assume that only letters $\sigma$ with  $\zug{\sigma}\in \{ 000, 010\} \cdot \{ 0, 1\}$ are such that
$s$ has an outgoing safe $\sigma$-transition in $\A^{G^{v_1}}$, and  we assume that all safe transitions from $s$ lead to $p$ in $\A^{G^{v_1}}$. Hence, the only internal state $s_x$ with $|x|=3$ and an outgoing safe transition labeled with a bit is the state $s_{000}$, and $(0|1)$-transitions from it lead to $p$.
\hfill \qed}

\begin{figure}[htb] 
\begin{center}
\includegraphics[width=.8\textwidth]{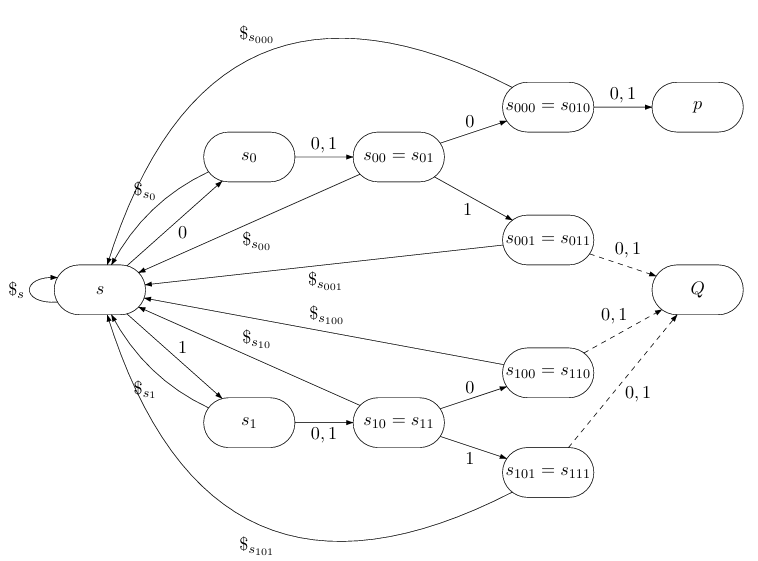}
\caption{The tNCW $e(\A^{G^{v_1}})$.}
\label{B example}
\end{center}
\end{figure}
\end{example}

Essentially, $e(\A^{G^{v_1}})$ reads $\zug{\sigma}$, for some $\sigma \in 2^{\mathsf{AP}}$, from $s\in Q$ by deterministically following a safe path to the state $s_x$ after reading a proper prefix $x$ of $\zug{\sigma}$. Then, upon reading a bit $b$ from the state $s_x$, for some $x\in \{ 0, 1\}^{k-1}$, i.e., when reading the last bit $b$ of the encoding of some letter $\zug{\sigma}$, $e(\A^{G^{v_1}})$ either proceeds via a safe transition to the only $\sigma$-successor $p$ of $s$ when $\zug{s, \sigma, p}$ is a safe transition in $\A^{G^{v_1}}$, or proceeds nondeterministically to all states $p$ in $Q$ via rejecting transitions otherwise. Note that the only nondeterminsitic choices in $e(\A^{G^{v_1}})$ are (a) upon reading the last bit from a state in $\{s_x\}_{x\in \{ 0, 1\}^{k-1}}$, (b) when reading $\it {sleep}$ from a state in $Q$ that has a nondeterministic choice upon reading $\it {sleep}$ in $\A^{G^{v_1}}$, or (c) when reading a $\$_d$-transition from an internal state distinct from $d$. In addition, the fact that $\A^{G^{v_1}}$ is $\alpha$-homogeneous implies that $e(\A^{G^{v_1}})$ is $\alpha'$-homogeneous as well. Indeed, upon reading the last letter of an encoding, we get that the induced encoded letter $\sigma$ either has a safe transition from $s$ in $\A^{G^{v_1}}$ or  rejecting transitions that lead to all states in $\A^{G^{v_1}}$, but not both. 
Also, there are no $\$_d$-labeled rejecting and safe transitions going out from the same state.

Recall that $\Sigma_\$ = \{ \$_d: d\in Q'\setminus \{ q_{rej}\}\}$ denotes the set of special letters in $e(\A^{G^{v_1}})$.
We say that a word
$w \in (\{ 0, 1, \it {sleep}\} \cup \Sigma_{\$})^\omega$
{\em respects the encoding} when $w\in ( (  (0+1)^{\leq k-1} \cdot \Sigma_{\$})^* \cdot ((0+1)^k + \it {sleep})^*)^\omega$.
Otherwise, we say that $w$ {\em violates the encoding}. For a word $w$ that respects the encoding, we define $w_{proj}$ as the word obtained from $w$ by removing all infixes of the form $(0+1)^{\leq k-1}\cdot \Sigma_{\$}$: we first write $w = z_1 \cdot z_2 \cdot z_3 \cdots$, where $z_i \in ( (0+1)^{\leq k-1} \cdot \Sigma_{\$} + (0+1)^k + \it{sleep}  )$  
for all $i\geq 1$. Then, $w_{proj}$ is obtained from $w$ by removing all infixes $z_i$ that contain a special letter from $\Sigma_\$$. 
The fact that $\it {sleep}$ transitions from internal states not in  $Q$ lead to the rejecting sink implies that
a run from a state in $Q$ on an infinite word $w$ does not reach the rejecting sink if and only if $w$ respects the encoding.

Consider a state $s\in Q$ and a letter $\sigma\in 2^{\mathsf{AP}}$.
We proceed with the following Lemma suggesting that there is a correspondence between $\sigma$-labeled transitions from $s$ in $\A^{G^{v_1}}$ and runs on $\zug{\sigma}$ from $s$ in $e(\A^{G^{v_1}})$:

\begin{lemma}\label{cr lem}
Consider the automata $\A^{G^{v_1}}= \zug{2^{\mathsf{AP}} \cup \{\it {sleep}\}, Q, q_0, \delta, \alpha}$, and $e(\A^{G^{v_1}}) = \zug{\{ 0, 1, \it{sleep}\} \cup \Sigma_\$, Q', q_0, \delta', \alpha'}$,  states $s, p\in Q$, and a letter $\sigma \in 2^{\mathsf{AP}}$. Then, 
\begin{enumerate}
    \item All runs of $(e(\A^{G^{v_1}}))^s$ on $\zug{\sigma}$ end in a state in $Q$. In addition, if a run $r = s \xrightarrow{\zug{\sigma} }p$ exists in $e(\A^{G^{v_1}})$, then it is  unique, and is of the form $r = s, s_{\zug{\sigma}[1, 1]}, s_{\zug{\sigma}[1, 2]}, \ldots, s_{\zug{\sigma}[1, k-1]}, p$.

    \item The triple $\zug{s, \sigma, p}$ is a transition of $\A^{G^{v_1}}$ iff there is a run $r = s \xrightarrow{\zug{\sigma} }p$ in $e(\A^{G^{v_1}})$. In addition, the run $r$ is safe only when the transition $\zug{s, \sigma, p}$ is safe.
    
\end{enumerate}
\end{lemma}

\begin{proof}
    \begin{enumerate}
        \item 
         Is immediate from the definition of $e(\A^{G^{v_1}})$. Indeed, upon reading $\zug{\sigma}$ from $s$, we move deterministically via safe transitions to the state $s_{\zug{\sigma}[1, i]}$ after reading the $i$'th bit of $\zug{\sigma}$ for all $i\in [k-1]$. 
        Then,  upon reading the last bit of $\zug{\sigma}$, we move to a state $q\in Q$ either deterministically via a safe transition, when $\zug{s, \sigma, q}$ is a safe transition of $\A^{G^{v_1}}$, or we move to all states in $Q$ via  rejecting transitions, otherwise. Either way, as $e(\A^{G^{v_1}})$ reads the first $k-1$ bits of $\zug{\sigma}$ deterministically, we get that there is at most one run  from $s$  on $\zug{\sigma}$ that ends in $p$, and that run is of the required form.

        \item 
         Consider a triple $t = \zug{s, \sigma, p} \in Q\times \Sigma \times Q$. 
         By the definition of $e(\A^{G^{v_1}})$, we know that there is a safe run of the form $ s \xrightarrow {\zug{\sigma}[1, k-1]} s_{\zug{\sigma}[1, k-1]}$ in $e(\A^{G^{v_1}})$.
        In addition, the above run can be extended to a run $r$ from $s$ on $\zug{\sigma}$ to $p$ only when the triple $t = \zug{s, \sigma, p}$ is a transition of $\A^{G^{v_1}}$.
         The definition of the acceptance condition of $e(\A^{G^{v_1}})$ implies that the run $r$ is safe only when $\zug{s, \sigma, p}$ is a safe transition of $\A^{G^{v_1}}$, and we are done.   \qedhere 
    \end{enumerate}
\end{proof}

As $\it{sleep}$-transitions in $e(\A^{G^{v_1}})$ are inherited from $\A^{G^{v_1}}$, we conclude the following by Lemma~\ref{cr lem}:

\begin{proposition}\label{B sharp-free prop}
    Consider a state $s\in Q$. It holds that 
\begin{multline*}
L(e(\A^{G^{v_1}})^s) \cap \{0, 1, \it {sleep}\}^\omega = 
    \{ w\in \{ 0, 1, \it{sleep}\}^\omega: \\\text{ $w$ respects the encoding and $dec(w) \in L((\A^{G^{v_1}})^s)$}\}
    \end{multline*}
\end{proposition}

\begin{proof}
\sloppypar    We first show that $ L(e(\A^{G^{v_1}})^s) \cap \{0, 1, \it {sleep}\}^\omega \subseteq  \{ w\in \{ 0, 1, \it {sleep}\}^\omega: \text{ $w$ respects the encoding and $dec(w) \in L((\A^{G^{v_1}})^s)$}\}$. 
    Consider a word $w = \sigma_1 \cdot \sigma_2 \cdot \sigma_3 \cdots \in L(e(\A^{G^{v_1}})^s)$ that does not contain special letters from $\Sigma_{\$}$, and let $r$ be an accepting run of $e(\A^{G^{v_1}})^s$ on $w$.
    As $r$ is accepting, it does not visit the rejecting sink $q_{rej}$; in particular, $w$ respects the encoding. Therefore, as $w$ does not contain special letters, we can write $w = z_1\cdot z_2 \cdot z_3 \cdots $, where $z_i \in ((0+1)^k + \it {sleep})$ for all $i \geq 1$. 
    Then, by Lemma~\ref{cr lem} and  the fact that $\A^{G^{v_1}}$ and $e(\A^{G^{v_1}})$ agree on $\it{sleep}$-transitions, the run $r$ is of the form $r = s_0 \xrightarrow{z_1} s_1 \xrightarrow{z_2} s_2 \cdots$, where $s_0 = s$, and all for all $i\geq 0$, the state $s_i$ is in  $Q$.
    Now an iterative application of Lemma~\ref{cr lem} implies that $r$ induces,
    by removing states not in $Q$, a  run $r'$ of the form $r' = s_0 \xrightarrow{dec(z_1)} s_1 \xrightarrow{dec(z_2)} s_2 \cdots$ of $(\A^{G^{v_1}})^s$ on $dec(w)$ that agrees with $r$ on acceptance. In particular, $r'$ is an accepting run, and thus $dec(w) \in L((\A^{G^{v_1}})^s)$. 

    Showing that $\{ w\in \{ 0, 1, \it{sleep}\}^\omega: \text{ $w$ respects the encoding and $dec(w) \in L((\A^{G^{v_1}})^s)$}\} \subseteq L(e(\A^{G^{v_1}})^s) \cap \{0, 1, \it{sleep}\}^\omega$ follows  similar considerations. Indeed, if $w = z_1\cdot z_2 \cdot z_3 \cdots $, where  $z_i\in ((0+1)^k + \it{sleep})$ for all $i\geq 1$, and $dec(w) = dec(z_1) \cdot dec(z_2) \cdot dec(z_3) \cdots \in L((\A^{G^{v_1}})^s)$, then
    if $r = s_0, s_1, \ldots$ is an accepting run of $(\A^{G^{v_1}})^s$ on $dec(w)$, then 
    an iterative application of Lemma~\ref{cr lem} implies that there is a run $r'$ of $e(\A^{G^{v_1}})^s$ on $w$ that agrees with $r$ on acceptance, and is of the form $r' = s_0 \xrightarrow{z_1} s_1 \xrightarrow{z_2} s_2 \cdots$, and we're done.
\end{proof}

We proceed with the following proposition suggesting that $e(\A^{G^{v_1}})$ is already a nice minimal HD-tNCW: 

\begin{proposition}
The tNCW  $e(\A^{G^{v_1}}) = \zug{\{ 0, 1, \it{sleep}\} \cup \Sigma_{\$},  Q', q_0, \delta', \alpha'}$  is a nice minimal HD-tNCW.    
\end{proposition}
    
\begin{proof}
  We first show that $e(\A^{G^{v_1}})$ is  semantically deterministic.
  By definition, all nondeterministic choices in $e(\A^{G^{v_1}})$ lead to all the states in $Q$. Thus, it is sufficient to show that all the states in $Q$ are $e(\A^{G^{v_1}})$-equivalent, and this follows by  Proposition~\ref{B sharp-free prop} and  the fact that all the states in $\A^{G^{v_1}}$ are equivalent. 
  Consider states $s_1, s_2\in Q$,
  and let  $r = r_0, r_1, \ldots$ be an accepting run of $e(\A^{G^{v_1}})^{s_1}$ on an infinite word $w$.
  We show next that $s_2$ has an accepting run %
  on $w$ as well, and thus $L(e(\A^{G^{v_1}})^{s_1})\subseteq L(e(\A^{G^{v_1}})^{s_2})$. As the latter holds for all states $s_1$ and $s_2$ in $Q$, we conclude that $s_1\sim_{e(\A^{G^{v_1}})} s_2$.
  As $r$ is an accepting run of $e(\A^{G^{v_1}})^{s_1}$ on $w$, then $w$ respects the encoding. %
  In particular, we can write $w = z_1\cdot z_2 \cdot z_3 \cdots $, where  $z_i \in ((0+1)^{\leq k-1}\cdot \Sigma_{\$} + (0+1)^k + \it{sleep})$ for all $i\geq 1$.
 Also, the run $r$ is of the form $r = r_0 \xrightarrow{z_1} r_1 \xrightarrow{r_2} r_2 \cdots $, where $r_0 = s_1$, and the state $r_i$ is in $Q$ for all $i\geq 0$.
  As $r$ is accepting, there is some $j\geq 0$ such that $r[j, \infty] =r_j \xrightarrow{z_{j+1}} r_{j+1} \xrightarrow{r_{j+2}} r_{j+2} \cdots$ is a safe run from the state in $r_j\in Q$ on the suffix $z_{j+1}\cdot z_{j+2} \cdots $. 
  Now consider an arbitrary finite run $r'[0, j]$ of $e(\A^{G^{v_1}})^{s_2}$ on the prefix $z_1\cdot z_2 \cdots z_j$ of $w$. Then, $r'[0, j]$ is of the form 
  $r'[0, j] =  r'_0 \xrightarrow{z_1} r'_1 \xrightarrow{r_2} r'_2 \cdots \xrightarrow{z_j} r'_j$, where $r'_0 = s_2$, and the state $r'_i$ is in $Q$ for all $i$.
  We proceed by showing that we can always extend the finite run $r'[0, j]$ to an accepting run of $e(\A^{G^{v_1}})^{s_2}$ on $w$. %
  We distinguish between two cases. If the suffix $z_{j+1}\cdot z_{j+2} \cdots$ has no special letters, then the finite run $r'[0, j]$ can be extended to an accepting run of $e(\A^{G^{v_1}})^{r'_0} = e(\A^{G^{v_1}})^{s_2}$ on $w$ as by Proposition~\ref{B sharp-free prop} and the fact that $r_j\sim_{\A^{G^{v_1}}} r'_j$, we have that $e(\A^{G^{v_1}})^{r_j}$ and $e(\A^{G^{v_1}})^{r'_j}$ agree on words that have no special letters from $\Sigma_\$$. Otherwise, let $\$_d$ be  the first special letter appearing in the suffix  $z_{j+1}\cdot z_{j+2} \cdots$. As the run $r[j, \infty]$ on the suffix  $z_{j+1}\cdot z_{j+2} \cdots$  is safe, then it traverses the safe $\$_d$-transition $\zug{d, \$_d, s}$ where $s\in Q$ is the state that $d$ corresponds to. %
  Hence, if an arbitrary extension of the finite run $r'[0,j]$ to a run $r'$ on $w$  does not visit the state $s$ right after reading the first special letter $\$_d$ in the suffix $z_{j+1}\cdot z_{j+2} \cdots$, then we can modify the extension $r'$ to eventually follow the accepting run $r[j, \infty]$ after reading the first special letter $\$_d$. Indeed, there is a single safe $\$_d$-transition and it leads to the state $s$, and all internal states that have no outgoing safe $\$_d$-transition, have rejecting $\$_d$-transitions leading to all states in $Q$. Thus, if by reading $\$_d$ the run $r'$ did not move to the state $s$, we can modify it by letting it move to $s$ and then follow the accepting run $r$ on the suffix to be read.

  We show next that all the states of $e(\A^{G^{v_1}})$ are reachable. As $\A^{G^{v_1}}$ is nice, we get that from the initial state $q_0$, we can reach every other state $s\in Q$ in $\A^{G^{v_1}}$. Hence, an iterative application of Lemma~\ref{cr lem} and the fact that $\A^{G^{v_1}}$ and $e(\A^{G^{v_1}})$ agree on $\it{sleep}$-transitions, imply that $s$ is also reachable from $q_0$ in $e(\A^{G^{v_1}})$. Then, we can also reach, from $s$, all internal states corresponding to $s$. Finally, the rejecting sink $q_{rej}$ is also reachable upon reading a word that violates the encoding. Next, note that $e(\A^{G^{v_1}})$ is safe-deterministic since, as by definition, it encodes safe transitions of $\A^{G^{v_1}}$ deterministically via the internal states,  safe $\Sigma_{\$}$-transitions are defined deterministically as well, and $\it {sleep}$-transitions are inherited from $\A^{G^{v_1}}$ which is nice, in particular, safe-deterministic to begin with. 
  We proceed by showing that $e(\A^{G^{v_1}})$ is normal.
  For a state $s\in Q$, let $I_s = \{ s_x\}_{|x|\leq k-1}$ denote the set of internal states corresponding to $s$.
    By definition, for a state $s\in Q$, all internal states in $I_s$ belong to the same safe component in $e(\A^{G^{v_1}})$. Indeed, for all $x$, $s \xrightarrow{x} s_x \xrightarrow{ \$_{s_x}} s$ is a safe cycle in $e(\A^{G^{v_1}})$. 
    Next, we show that if there is a safe run $r$ from a state $s\in Q$ to a state $p\in Q$ labeled with the encoding of some letter in $\sigma \in 2^{\mathsf{AP}} \cup \{ \it sleep\}$, then $s$ and $p$ belong to the same safe component in $e(\A^{G^{v_1}})$.
    Indeed, either $r = \zug{s, \sigma, p}= \zug{s, \it sleep, p}$, or
    $\sigma$ is a letter in $2^{\mathsf{AP}}$ and
    $r$ is a run of the form  $r = s, s_{\zug{\sigma}[1, 1]}, s_{\zug{\sigma}[1, 2]}, \ldots, s_{\zug{\sigma}[1, k-1]}, p$ over the encoding of  $\sigma$.   The fact that $\A^{G^{v_1}}$ and $e(\A^{G^{v_1}})$ agree on $\it sleep$ transitions, and Lemma~\ref{cr lem}, imply that in both cases, $\zug{s, \sigma, p}$ is a safe transition in $\A^{G^{v_1}}$. Now the fact that $\A^{G^{v_1}}$ is normal implies that 
    there is a safe path from $p$ back to $s$ in $\A^{G^{v_1}}$ which induces by Lemma~\ref{cr lem} a safe path from $p$ back to $s$ in $e(\A^{G^{v_1}})$, and thus $p$ and $s$ belong to the same safe component.
    By the above, 
    it follows that for every two states $s, p\in Q$, if there is a safe run $r$ from a state in $I_s$ to a state in $ I_p$ that visits only states in $I_s\cup I_p$, then  all the states in $I_s\cup I_p$ belong to the same safe component. 
    Indeed, since states in $I_s$ can reach only the state $p $ in $ I_p$ upon reading a single letter, then either $r$ traverses the transition $\zug{s, \it sleep, p}$, or 
 $r$ traverses a sub-run of the form  $s_x \xrightarrow{z} s_{y} \xrightarrow{b} p $, where $z\in \{ 0, 1\}^*$, $b$ is a bit, and $s_x \xrightarrow{z} s_{y}$ is a run that does not enter $s$ from a different state in $I_s$. The latter sub-run extends the finite run $s \xrightarrow{x} s_x$ to a safe run from $s$ to $p$ over the encoding of some letter $\sigma\in 2^{\mathsf{AP}}$. In both cases, we get that $s$ and $p$ belong to the same safe component, and thus all the states in $I_s\cup I_p$ belong to the same safe component as well. 
    Therefore, by induction, 
    any internal state $p_x \in I_p$ that is reachable via an arbitrary safe path from an internal state $s_y\in I_s$ is such that $s_y$ and $p_x$ belong to the same safe component, in particular, $p_x$ can close a safe cycle back to $s_y$, and 
    thus $e(\A^{G^{v_1}})$ is normal.
  To conclude that $e(\A^{G^{v_1}})$ is a nice HD-tNCW, we show that $e(\A^{G^{v_1}})^d$ is HD for all $d\in Q'$, and this follows from Theorem~\ref{are GFG thm}. Indeed, $e(\A^{G^{v_1}})$ is $\alpha'$-homogeneous, and we show next that it is $\alpha'$-saturated. 
  Consider a rejecting transition $t = \zug{d, a, d'}$ in $e(\A^{G^{v_1}})$. We distinguish between two cases. (1) either $d' = q_{rej}$, in which case, as every state in $Q'\setminus \{ q_{rej}\}$ belongs to a nontrivial safe component, $d'$ has no $e(\A^{G^{v_1}})$-equivalent states, or (2) $d' \neq q_{rej}$, in which case $d = s_x$ for some state $s\in Q$ and binary word $x$ of length at most $k-1$. In this, case, if  $a\in \Sigma_{\$}$ is a special letter or $a\in \{0, 1\}$ is a bit, then, in both cases, by the definition of the transition function of $e(\A^{G^{v_1}})$, the state $d$ has $a$-labeled rejecting transitions that lead to every state in $Q$. As all the states in $Q$ are $e(\A^{G^{v_1}})$-equivalent, and as $\it {sleep}$-transitions separate between states in $Q$ and states not in $Q$, we get that there is no state in $Q'\setminus Q$ that is equivalent to a state in $Q$, and so $e(\A^{G^{v_1}})$ is already $\alpha'$-saturated.

  To conclude the proof, we show that $e(\A^{G^{v_1}})$ is a minimal HD-tNCW, and this is immediate from the fact that every state distinct from $q_{rej}$ is not equivalent to $q_{rej}$, and all distinct states $d$ and $d'$ in $Q'\setminus \{q_{rej}\}$ have incomparable safe languages, and nice HD-tNCWs that satisfy the latter are already minimal~\cite{DBLP:journals/lmcs/RadiK22}. Indeed, let $d = s_x$ for some $s\in Q$. Then, as only the state $d$ has a safe outgoing $\$_d$-labeled transition and it leads to $s$, it follows that the word $\$_d \cdot w$, for some word $w\in L_{\it safe}(e(\A^{G^{v_1}})^s)$, is in $L_{\it safe} (e(\A^{G^{v_1}})^d)\setminus L_{\it safe}(e(\A^{G^{v_1}})^{d'})$. Indeed, every state in $Q'\setminus \{ q_{rej}\}$ belongs to a nontrivial safe component, and thus $w$ exists.
\end{proof}

We can now prove Proposition~\ref{reduc applies prop}, suggesting that the reduction from Section~\ref{exp red sec} can output $e(\A^{G^{v_1}})$ instead of $\A^{G^{v_1}}$: 

\begin{proposition}\label{prop enc corr}
    There is a Hamiltonian path in $G$ starting at $v_1$ iff the minimal HD-tNCW $e(\A^{G^{v_1}})$
has an equivalent tDCW of the same size.
\end{proposition}

\begin{proof}
    We start with the easy direction, showing that the existence of a tDCW $\D$ of size $|Q'|$ equivalent to $e(\A^{G^{v_1}})$ implies the existence of a Hamiltonian path in $G$ starting at $v_1$. 
    Note that as $e(\A^{G^{v_1}})$ is a minimal HD-tNCW, and as every deterministic automaton is HD, we get that 
    $\D$ is a minimal tDCW. In particular, all the states in $\D$ are reachable, and thus it can be assumed to be nice as we can make it normal by turning safe transitions that connect distinct safe components to be rejecting. 
Now the fact that $e(\A^{G^{v_1}})$ is a nice minimal HD-tNCW and $\D$ is nice imply, by Theorem~\ref{iso thm}, that $\D$ and $e(\A^{G^{v_1}})$ have isomorphic safe components. Hence, as all states in $e(\A^{G^{v_1}})$ have distinct safe languages,
    we can write $\D = \zug{ \{ 0, 1, \it{sleep}\} \cup \Sigma_\$, Q', q_0, \delta_d, \alpha_d}$, where the transition functions $\delta'$ and $\delta_d$ of $e(\A^{G^{v_1}})$ and $\D$, respectively, agree on safe transitions, and for every state $q\in Q'$, the automata $e(\A^{G^{v_1}})^q$ and $\D^q$ are equivalent.
    Consider a letter $\sigma \in 2^{\mathsf{AP}} \cup \{ \it{sleep}\}$, and a state $q\in Q$. We claim that the run of $\D^q$ on $\zug{\sigma}$ ends in a state $s\in Q$. Indeed, as the automata $e(\A^{G^{v_1}})^q$ and $\D^q$ are equivalent and semantically deterministic,
    it follows by an iterative application of Proposition~\ref{SD prop} that the state $\delta_d(q, \zug{\sigma})$ is equivalent to the states in $\delta'(q, \zug{\sigma}) \subseteq Q$. 
    Therefore, as the rejecting sink and internal states in $Q'\setminus (Q\cup \{q_{rej}\})$ are not equivalent to any state in $Q$ in the automaton $e(\A^{G^{v_1}})$, we get that $\delta_d(q, \zug{\sigma})$ must be a state in $Q$. %
    
We can now define a tDCW $\A_d$ of size $|Q|$  equivalent to $\A^{G^{v_1}}$, and thus, by Theorem~\ref{corr thm},  
a Hamiltonian path in $G$  starting at $v_1$ exists.
The tDCW $\A_d = \zug{2^{\mathsf{AP}} \cup \{ \it{sleep}\}, Q, q_0, \eta, \beta}$ is defined as follows. For every letter $\sigma\in 2^{\mathsf{AP}} \cup \{\it{sleep}\}$ and state $q\in Q$, we define $\eta(q, \sigma) = \delta_d(q, \zug{\sigma})$, and that transition is safe only when the run $q \xrightarrow{\zug{\sigma}} \delta_d(q, \zug{\sigma})$ is safe in $\D$. Clearly,  $\A_d$ is equivalent to $\A^{G^{v_1}}$ as, for an infinite word $w\in (2^{\mathsf{AP}} \cup \{ \it{sleep}\})^\omega$, the run of $\A_d$ on $w$ is accepting only when the run of  $\D$ on $\zug{w}$ is accepting.
Hence, as $\zug{w}$ respects the encoding, and does not contain special letters from $\Sigma_\$$, then Proposition~\ref{B sharp-free prop} implies that $\D$ accepts $\zug{w}$ only when $dec(\zug{w}) = w$ is in $L(\A^{G^{v_1}})$, and we're done.

We proceed with the other direction. Assume that $G$ has a Hamiltonian path $p$ starting at $v_1$, and let $\A_d$ be the deterministic pruning of $\A^{G^{v_1}}$ that recognizes $L(\A^{G^{v_1}})$ and is induced by the Hamiltonian path $p$ as described in the proof of Theorem~\ref{corr thm}. We show that $\A_d$ induces a deterministic pruning $\D$ of $e(\A^{G^{v_1}})$ that recognizes $L(e(\A^{G^{v_1}}))$:

    \noindent As $e(\A^{G^{v_1}})$ is $\alpha'$-homogenuous and safe-deterministic, we need only to prune rejecting transitions to obtain $\D$. 
    \noindent We first determinize $\Sigma_\$$-labeled rejecting transitions  by letting a $\$_d$-labeled transition move only to the state $s\in Q$ that $d$ corresponds to, that is, $d = s_x$ for some binary word $x$ of length at most $k-1$.
    Next, we determinize rejecting transitions labeled with letters from $\{ 0, 1, \it {sleep}\}$ based on the pruning $\A_d$ of $\A^{G^{v_1}}$. %
    Consider a rejecting transition $\zug{q, \sigma, s}$ of $\A_d$. We distinguish between two cases. First, if $\sigma \in 2^{\mathsf{AP}}$, then by Lemma~\ref{cr lem}, there is a unique non-safe run of the form $r = q, q_{\zug{\sigma}[1, 1]}, q_{\zug{\sigma}[1, 2]}, \ldots, q_{\zug{\sigma}[1, k-1]}, s$ on $\zug{\sigma}$, where $r$ traverses a rejecting transition only upon reading the last bit $b = \zug{\sigma}[k]$ of the encoding $\zug{\sigma}$. Then, we determinize the rejecting $b$-labeled transitions from the internal state $q_{\zug{\sigma}[1, k-1]}$ by letting it proceed only to the state $s$. 
    Next, if $\sigma = \it{sleep}$, then we determinize  $\it{sleep}$-transitions from $q$ in $e(\A^{G^{v_1}})$ by letting it proceed to the state $s$.
    
    We now show that $\D$ is well-defined by claiming that there are no rejecting transitions $\zug{q, \sigma_1, s_1}$ and $\zug{q, \sigma_2, s_2}$ in $\A_d$ where (1) $q_{\zug{\sigma_1}[1, k-1]} = q_{\zug{\sigma_2} [1, k-1]}$, (2) $\zug{\sigma_1}[k] = \zug{\sigma_2} [k]$, and (3) $s_1\neq s_2$.
    Indeed, %
    the destination of a rejecting $\sigma$-transition from $q$ in $\A_d$
    is determined based on $q$ and the values that the letter $\sigma$ assigns to atomic propositions essential to the state $q$.  
    Therefore, as $s_1 \neq s_2$, it follows that the letters $\sigma_1$ and $\sigma_2$ assign different values to at least one atomic proposition  $a_i\in ess(q)$ that is essential to $q$, and w.l.o.g $i$ is the minimal such index in $[k]$. As $\zug{\sigma_1}[k] = \zug{\sigma_2} [k]$, it holds that $i < k$. 
    Then,
    the definition of $e(\A^{G^{v_1}})$ implies that $q_{\zug{\sigma_1}[1, k-1]} \neq q_{\zug{\sigma_2} [1, k-1]}$ as the deterministic safe runs of $e(\A^{G^{v_1}})^q$ on $\zug{\sigma_1}[1, k-1]$ and $\zug{\sigma_2} [1, k-1]$ diverge after reading the $i$'th bit corresponding to the atomic proposition $a_i$.

    We proceed by showing that if we apply the described pruning of $e(\A^{G^{v_1}})$ for every rejecting transition of $\A_d$, we get that the resulting automaton $\D$ is indeed deterministic.
    Consider an encoding $\zug{\sigma}$ of a letter $\sigma\in 2^{\mathsf{AP}}$, and let $b$ denote the last bit of $\zug{\sigma}$, and consider a rejecting $b$-transition in $\D$ from the internal state  $q_{\zug{\sigma}[1, k-1]}$ to a state in $Q$. The latter transition witness, by Lemma~\ref{cr lem}, that $\A^{G^{v_1}}$ has a rejecting $\sigma$-transition from $q$.
    Then, as $\A^{G^{v_1}}$ is $\alpha$-homogeneous, 
    and $\A_d$ is a deterministic pruning of its rejecting transitions, we have that the state $q$ has a rejecting transition $t = \zug{q, \sigma, s}$  also in $\A_d$, and therefore we must have pruned  $b$-transitions from  $q_{\zug{\sigma}[1, k-1]}$ in $e(\A^{G^{v_1}})$ when considering the transition $t$.

Finally, as $\D$ is embodied in $e(\A^{G^{v_1}})$, then $L(\D) \subseteq L(e(\A^{G^{v_1}}))$. To conclude the proof, we show in Appendix~\ref{prop enc corr app} that $L(e(\A^{G^{v_1}}))\subseteq L(\D)$. Essentially, it follows from the fact that the existence of an accepting run $r$ of $e(\A^{G^{v_1}})$ on a word $w$, and a rejecting run of $r'$ of $\D$ on the same word $w$ is impossible as the latter runs induce a word in $L(\A^{G^{v_1}})\setminus L(\A_d)$.
\end{proof}

\subsection{Missing Details in The Proof of Proposition~\ref{prop enc corr}}\label{prop enc corr app}
  We show next that $L(e(\A^{G^{v_1}}))\subseteq L(\D)$.
    Consider a word $w\in L(e(\A^{G^{v_1}}))$, and note that since $e(\A^{G^{v_1}})$ has a run on $w$ that does not visit $q_{rej}$, then $w$ respects the encoding. Thus, we can write $w = z_1 \cdot z_2 \cdot z_3 \cdots $, where $z_i \in ( (0+1)^{\leq k-1} \cdot \Sigma_{\$} + (0+1)^k + \it{sleep}  )$ for all $i\geq 1$.
    Let $r = r_0 \xrightarrow{z_1} r_1 \xrightarrow{z_2} r_2 \cdots$ be an accepting run of $e(\A^{G^{v_1}})$ on $w$, where $r_0 = q_0$, and the state $r_i$ is in $Q$ for all $i\geq 0$. Then, let $j\geq 0$ be such that $r[j, \infty] = r_j \xrightarrow{z_{j+1}} r_{j+1} \xrightarrow{z_{j+2}} r_{j+2} \cdots$ is a safe run on the suffix $z_{j+1}\cdot z_{j+2} \cdots $.
    First, note that as $\D$ is a deterministic pruning of $e(\A^{G^{v_1}})$, then the run $r'$ of $\D$ on $w$ is also of the form $r' = r'_0 \xrightarrow{z_1} r'_1 \xrightarrow{z_2} r'_2 \cdots$, where for all $i\geq 0$, $r'_i\in Q$.
    We claim that the run $r'[j, \infty] = r'_j \xrightarrow{z_{j+1}} r'_{j+1} \xrightarrow{z_{j+2}} r'_{j+2} \cdots$ is a run of $\D^{r'_j}$ on the suffix $z_{j+1}\cdot z_{j+2} \cdots $  that traverses  rejecting transitions only finitely often, and thus $\D$ accepts $w$.  
    Assume towards contradiction that $r'[j, \infty]$ traverses infinitely many rejecting transitions. 
    First, note that it cannot be the case that $r'_k \xrightarrow{z_{k+1}} r'_{k+1}$ is a non-safe run in $\D$, for some $k\geq j$ with $z_{k+1}\in (0+1)^{\leq k-1} \cdot \Sigma_\$$.
    Indeed, in this case,  $r'_k \xrightarrow{z_{k+1}} r'_{k+1}$ traverses a rejecting transition only in its last transition.  Let $\$_{s_x}$ denote the last letter in $z_{k+1}$ where $s\in Q$. The definition of rejecting $\Sigma_\$$ transitions in $\D$ implies that $r'_{k+1} = s$. Also, as $r_k \xrightarrow{z_{k+1}} r_{k+1}$ is a safe sub-run of the run $r[j, \infty]$, we get by the definition of safe $\Sigma_\$$-transitions in $e(\A^{G^{v_1}})$ that $r_{k+1} = s = r'_{k+1}$. In particular, as $r[k, \infty]$ is safe, and $e(\A^{G^{v_1}})$ is safe-deterministic and $\alpha'$-homogeneous, we get that $r'[k+1, \infty] = r[k+1, \infty]$; Indeed, a safe run from a state $q$ on a word $x$ in any pruning of $e(\A^{G^{v_1}})$ (including $e(\A^{G^{v_1}})$ itself) is the only run of $q$ on $x$. So we got in total that $r'[j, \infty]$ eventually becomes safe, and we've reached a contradiction. 
    Hence, it must be  the case that for all $k\geq j$,  if $r'_k \xrightarrow{z_{k+1}} r'_{k+1}$ traverses a rejecting transition, then    $z_{k+1}\in ((0+1)^k + \it sleep)$.
    Consider the runs $r_{proj}[j, \infty]$ and $r'_{proj}[j, \infty]$ on $(z_{j+1}\cdot z_{j+2} \cdots )_{proj}$ that are obtained from  $r[j, \infty]$ and $r'[j, \infty]$, respectively, by removing sub-runs on infixes $z_{i}$ that contain a special letter from $\Sigma_\$$. By what we've argued above, the removed sub-runs are safe, in particular, by the assumption, $r'_{proj}[j, \infty]$ traverses infinitely many rejecting transitions. 
    Note that the latter  implies also that the runs $r_{proj}[j, \infty]$ and $r'_{proj}[j, \infty]$ are infinite.
    In addition, note that the definition of safe $\Sigma_\$$-transitions implies that a removed sub-run  %
    is a run from a state $q$ in $Q$ to the same state $q$. Thus, $r_{proj}[j, \infty]$ and $r'_{proj}[j, \infty]$ are legal runs of $e(\A^{G^{v_1}})$ and $\D$, respectively. 
    Therefore, 
    by the definition of $\D$, Lemma~\ref{cr lem}, and the fact that $\A^{G^{v_1}}$  and $e(\A^{G^{v_1}})$ agree on $\text{sleep}$-transitions, we get that the runs $r_{proj}[j, \infty]$ and $r'_{proj}[j, \infty]$ on $(z_{j+1}\cdot z_{j+2} \cdots )_{proj}$ induce runs of $(\A^{G^{v_1}})^{r_j}$ and $\A^{r'_j}_d$ on $dec((z_{j+1}\cdot z_{j+2} \cdots )_{proj})$, where the first run is safe, and the other is rejecting.  Therefore, as $\A_d$ is deterministic, and all the states of $\A_d$ are equivalent to all the states of $\A^{G^{v_1}}$, we have reached a contradiction, and so $L(e(\A^{G^{v_1}})) \subseteq L(\D)$.
    
    To conclude the proof, it remains to explain  how the run $r'_{proj}[j, \infty]$  induces a rejecting run of  $\A^{r'_j}_d$ on $dec((z_{j+1}\cdot z_{j+2} \cdots )_{proj})$, and this follows from the fact that the definition of $\D$ is such that a run $r_{q\to s} = q\xrightarrow{\zug{\sigma}} s$ in $\D$, for states $q, s\in Q$ and a letter $\sigma\in 2^{\mathsf{AP}} \cup \{ \it sleep\}$, witness the existence of a transition $\zug{q, \sigma, s}$ in $\A_d$ that agrees with $r_{q\to s}$ on acceptance. 
    To see why,  we distinguish between two cases. If $\sigma = \it sleep$, then we're done as $\it sleep$-transitions in $\D$ are inherited from $\A_d$; in particular $q \xrightarrow{\zug{\it sleep}} s = q\xrightarrow{\it sleep} s$ is a transition in $\A_d$.
    Indeed, if $\zug{q, \it sleep, s}$ is rejecting in $\D$, then by definition of rejecting transitions of $\D$, $\zug{q, \it sleep, s}$ is also a rejecting transition of $\A_d$. Then, if $\zug{q, \it sleep, s}$ is safe in $\D$, then it is safe in $e(\A^{G^{v_1}})$, and thus it must be safe in $\A^{G^{v_1}}$ as well. Then, the fact that $\A_d$ is a deterministic pruning of $\A^{G^{v_1}}$'s rejecting transitions implies that $\zug{q, \it sleep, s}$
    is also a safe transition in $\A_d$. 
    We proceed to the case where  $\sigma$ is some letter in $2^{\mathsf{AP}}$, and note that  
    since $\D$ is a deterministic pruning of $e(\A^{G^{v_1}})$, then the run $r_{q \to s} = q \xrightarrow{\zug{\sigma}} s$ exists also in $e(\A^{G^{v_1}})$. Therefore, we get by Lemma~\ref{cr lem} that $r_{q \to s}= q\xrightarrow{\zug{\sigma}} s$ induces a transition $t = \zug{q, \sigma, s}$ in $\A^{G^{v_1}}$ that is safe only when $r_{q \to s}$ is safe. We need to show that  $t = \zug{q, \sigma, s}$ exists in $\A_d$.
    If $t$ is safe,  then as $\A_d$ is a deterministic pruning of $\A^{G^{v_1}}$'s rejecting transitions, we're done. Otherwise, if the transition $t = \zug{q, \sigma, s}$ is rejecting, then as $\A^{G^{v_1}}$ is $\alpha$-homogenous, we have that $\sigma$-transitions from $q$ in $\A_d$ are rejecting.
    Assume towards contradiction that  $\zug{q, \sigma, s'}$ is a rejecting transition of $\A_d$, yet $s'\neq s$. Then, on the one hand, the definition of $\D$ is such that $q_{\zug{\sigma}[1, k-1]}$ has a rejecting $\zug{\sigma}[k]$-transition to $s'$, yet on the other hand, as $t$ is not safe,  we get that $r_{q\to s}$ traverses a rejecting transition upon reading the last letter of $\zug{\sigma}$, and so $q_{\zug{\sigma}[1, k-1]}$ has also a rejecting $\zug{\sigma}[k]$-transition to $s$, and we have reached a contradiction to the fact that $\D$ is deterministic.

\section{Adapting The Encoding to The tDCW $\A^{G^{v_1}_p}$}
\label{adapt enc app}

Recall that $\A^{G^{v_1}} = \zug{2^{\mathsf{AP}} \cup \{ \it sleep\}, Q, q_0, \delta, \alpha}$, $e(\A^{G^{v_1}}) = \zug{ \{0, 1 \it sleep\} \cup \Sigma_\$, Q', q_0, \delta', \alpha'}$, and $k = |\mathsf{AP}|$.
Before adapting the encoding scheme to the tDCW $\A^{G^{v_1}_p}$, note that by Lemma~\ref{cr lem}, the encoding $e(\A^{G^{v_1}})$ preserves the safe components of $\A^{G^{v_1}}$ in the sense that two states $q$ and $s$ in $Q$ belong to the same safe component in $\A^{G^{v_1}}$ if and only if they belong to the same safe component in $e(\A^{G^{v_1}})$.
Indeed, if there is a safe run $r$ from $q$ to $s$ in $e(\A^{G^{v_1}})$ that visits only internal states corresponding to $q$ and $s$, then as safe $\Sigma_\$$ transitions from internal states corresponding to $q$ lead back to $q$, and as the first state corresponding to $s$ that $r$ visits is $s$,  it follows that $r$ eventually follows a safe run from $q$ to $s$ labeled with some word in $((0+1)^k + \it sleep)$; thus inducing a safe transition from $q$ to $s$ in $\A^{G^{v_1}}$.
As the encoding $e(\A^{G^{v_1}})$ preserves safe components, then we can think of it as an encoding of safe components -- we first encode every state $q$ in $\SC(q) \in \SC(\A^{G^{v_1}})$ and add its internal states and safe transitions from them, then we add rejecting transitions among encoded safe components, add the rejecting sink $q_{rej}$ and transitions that lead to it, and  finally we add  transitions labeled with special letters.

We show next how to adapt the encoding scheme to $\A^{G^{v_1}_p}$ to get an encoding $e(\A^{G^{v_1}_p})$ with the following properties: (1) $L(e(\A^{G^{v_1}})) = L(e(\A^{G^{v_1}_p}))$; thus, the encodings of the two automata are equivalent, and (2) $e(\A^{G^{v_1}_p})$ is a tDCW; thus, the encoding preserves determinism.  
Recall that the tDCW $\A^{G^{v_1}_p}$ is obtained from the HD-tNCW $\A^{G^{v_1}}$ by duplicating its safe components in a way that simulates the polynomial path $p = v_1, v_2, \ldots, v_l$.
Therefore, an encoding of $\A^{G^{v_1}_p}$ can be obtained from $e(\A^{G^{v_1}})$ by duplicating its safe components, connecting them (deterministically) via rejecting transitions by following the path $p$, and finally letting a rejecting transition labeled with a special letter $\$_{s_x}\in \Sigma_\$$ move deterministically only to the copy of the state $s\in Q$ that lies in the minimal safe component that contains a copy of $s$. 
Formally, for a safe component $S\in \SC(\A^{G^{v_1}})$, let $e(S)$ denote its (deterministic) encoding in $e(\A^{G^{v_1}})$. Thus, $e(S)$ contains all the states in $S$, their corresponding internal states, and all safe transitions between them (including safe $(\Sigma_\$|\it{sleep})$-labeled transitions).
Note that as $e(\A^{G^{v_1}})$ preserves safe components of $\A^{G^{v_1}}$, the safe transitions from a state of the form $q_x$ for some $q\in S$ lead only to states in $e(S)$. 
Then, the encoding  $e(\A^{G^{v_1}_p})$ is defined over the alphabet $\{ 0, 1, \it sleep\} \cup \Sigma_\$$, and on top of the safe components $\{e(\SC(v_i))\times \{i\}\}_{i\in [l]} \cup \{ e(\SC_{\it Sync}) \times \{l+1\}\} \cup \{ \{q_{rej}\}\}$, where $e(S(v_i))\times \{i\}$ is a copy of the encoding of the safe component corresponding to the $i$'th vertex $v_i$ of the path $p$.
Note that for convenience, also here%
, we label the states in $e(\SC_{\it Sync})$ with the index $l+1$.
So far, we have defined the safe components on which the encoding $e(\A^{G^{v_1}_p})$ is defined on top.
Before defining rejecting transitions between these safe components, let us first show that they already encode
safe transitions of $\A^{G^{v_1}_p}$.
Consider a state $\zug{q, i}$ of $\A^{G^{v_1}_p}$, and  a letter $\sigma \in 2^{\mathsf{AP}}\cup \{ \it sleep\}$, where $t = \zug{q, i} \xrightarrow{\sigma} \zug{q', i}$ is a safe transition of $\A^{G^{v_1}_p}$. Then, let $S$ denote the safe component of $q$ in $\A^{G^{v_1}}$, and note that
by the definition of $\A^{G^{v_1}_p}$, we have that $\zug{q, \sigma, q'}$ is a safe transition of $\A^{G^{v_1}}$. Therefore,  
as the safe component $e(S)\times \{ i\}$ is  a copy of $e(S)$, then  if $\sigma \in 2^{\mathsf{AP}}$, we get by Lemma~\ref{cr lem}  that $e(S)\times \{ i\}$ contains the safe run $\zug{q, i}, \zug{q_{\zug{\sigma}[1, 1]}, i}, \zug{q_{\zug{\sigma}[1, 2]}, i}, \ldots, \zug{q_{\zug{\sigma}[1, k-1]}, i} , \zug{q', i}$ over $\zug{\sigma}$, and if $\sigma = \it sleep$, we have that the transition $t = \zug{q, i}  \xrightarrow{\it sleep} \zug{q', i}$ belongs to $e(S)\times \{ i\}$ as $\it sleep$ transitions are inherited as is. %
We define next rejecting transitions between the safe components of $e(\A^{G^{v_1}_p})$  by simulating rejecting transitions of $\A^{G^{v_1}_p}$. Recall that $\A^{G^{v_1}_p}$ has no rejecting $\it sleep$  transitions. Then, for  every rejecting transition 
$t = \zug{q, i} \xrightarrow{\sigma} \zug{q', (i \ mod (l+1))+1}$ of $\A^{G^{v_1}_p}$, where $\sigma\in 2^{\mathsf{AP}}$,  we add the rejecting transition $\zug{q_{\zug{\sigma}[1, k-1]}, i} \xrightarrow{ \sigma[k]} \zug{q',  (i \ mod (l+1))+1}$ in $e(\A^{G^{v_1}_p})$. Thus, rejecting transitions in $e(\A^{G^{v_1}_p})$ follow those in $\A^{G^{v_1}_p}$ by moving to the appropriate state in the encoding of the next safe component as induced by $\A^{G^{v_1}_p}$'s structure. As we argue in Lemma~\ref{equiv enc lemm}, the encoding of rejecting transitions of $\A^{G^{v_1}_p}$ does not introduce nondeterministic choices in $e(\A^{G^{v_1}_p})$.
We proceed with defining rejecting $\it sleep$ and $\Sigma_\$$ transitions in $e(\A^{G^{v_1}_p})$.
Let $Q_{v_1, p}$ denote the state-space of $\A^{G^{v_1}_p}$.
Recall that  $\it{sleep}$-labeled transitions in $\A^{G^{v_1}_p}$ are safe transitions, and so they are already encoded (from states in $Q_{v_1, p}$) and inherited from the original automaton as is. Then, similarly to the encoding $e(\A^{G^{v_1}})$, we let $\it{sleep}$-labeled transitions from states not in $Q_{v_1, p}$ be rejecting transitions that
lead to the rejecting sink $q_{rej}$.
Finally, recall that only the  special letters from $\Sigma_\$ = \{\$_d: d \in Q'\setminus \{q_{rej}\} \}$ are added to the alphabet of the encoding $e(\A^{G^{v_1}_p})$; in particular, both encodings $e(\A^{G^{v_1}})$ and $e(\A^{G^{v_1}_p})$ are defined over the same alphabet.
First, as stated previously, note that safe $\Sigma_\$$-labeled transitions are already defined, and so we only need to define rejecting $\Sigma_\$$-transitions.  
Consider a state $\zug{q_y, i}$ that belongs to a safe component of the form $e(\SC(v_i)) \times \{ i\}$ or belongs to $e(\SC_{\it Sync}) \times \{ l+1\}$, and consider a special letter $\$_{s_x}$ for some $s\in Q$ and a binary word $x$ of length at most $k-1$. If $\zug{q_y, i}$ has no safe outgoing-transition labeled with $\$_{s_x}$ in $e(\A^{G^{v_1}_p})$, then we add (deterministically) the rejecting transition $ \zug{\zug{ q_y, i }, \$_{s_x},\zug{s, i_s}  }$ in $e(\A^{G^{v_1}_p})$, where $i_s$ is the minimal index $j$ in $[l+1]$ such that $\zug{s, j}$ belongs to the encoded safe component whose states are labeled with the index $j$. Also, as expected, $\Sigma_\$$-labeled transitions from the rejecting sink $q_{rej}$ are rejecting and lead to $q_{rej}$. To conclude the construction of $e(\A^{G^{v_1}_p})$ we choose its initial state to be identical to the initial state of $\A^{G^{v_1}_p}$; in particular its initial state is in $Q_{v_1, p}$.

Thus, the encoding $e(\A^{G^{v_1}_p})$ is induced naturally from the encoding $e(\A^{G^{v_1}})$ and the path $p$, it is identical to $e(\A^{G^{v_1}})$ in the way it encodes safe components. The exception is how it defines rejecting transitions among them: for non-special letters, $e(\A^{G^{v_1}_p})$ follows the transition function of $\A^{G^{v_1}_p}$, and for rejecting transitions labeled with a special letter $\$_{s_x}$, it always proceeds
to the same state $\zug{s, i_s}$ in the same safe component to preserve determinism.

\subsection{Proof of Lemma~\ref{equiv enc lemm}}
\label{equiv enc lemm app}
We start with showing that $e(\A^{G^{v_1}_p})$ is deterministic.
First, the fact that the encoding $e(\A^{G^{v_1}})$ is safe-deterministic implies that the encoding $e(\A^{G^{v_1}_p})$ is defined on top of deterministic safe components. In addition, rejecting $\it sleep$
 and rejecting $\Sigma_\$$ transitions are defined deterministically as well.
Therefore, to conclude that $e(\A^{G^{v_1}_p})$ is deterministic,
 it is sufficient to show that encoding rejecting transitions of $\A^{G^{v_1}_p}$ does not introduce nondeterministic choices in $e(\A^{G^{v_1}_p})$, namely, we show next that there are no two  transitions labeled with the same bit from a state of the form $\zug{q_{\zug{\sigma}[1, k-1]}, i}$ for some $q\in Q, \sigma\in 2^{\mathsf{AP}}$ and $i\in [l+1]$, yet lead to distinct states.
 First, %
 the definition of rejecting $\sigma$-transitions from $\zug{q, i}$ in $\A^{G^{v_1}_p}$ is  based only on $q$ and the values $\sigma$ assigns to atomic propositions essential to $q$.
Then, the fact that $e(\A^{G^{v_1}})$ is well-defined implies that for every two letters $\sigma_1, \sigma_2\in 2^{\mathsf{AP}}$ with $\zug{\sigma_1}[k] = \zug{\sigma_2}[k]$, we have  that if $q_{\zug{\sigma_1}[1, k-1]} = q_{\zug{\sigma_2}[1, k-1]}$, then  the letters $\sigma_1$ and $\sigma_2$ disagree only on bits corresponding to atomic propositions not essential to $q$. Therefore, as safe $\sigma$-transitions from $q$ in $\A^{G^{v_1}}$ are defined also based only on $q$ and atomic propositions essential to $q$, it follows that %
$\zug{\zug{q, i}, \sigma_1, \zug{s, i}}$ is a safe(rejecting) transition in $\A^{G^{v_1}_p}$ if and only if $\zug{\zug{q, i}, \sigma_2, \zug{s, i}}$ is a safe(rejecting, respectively) transition in $\A^{G^{v_1}_p}$. %
Therefore, for every distinct letters $\sigma_1$ and $\sigma_2$ in $2^{\mathsf{AP}}$  with  $\zug{\sigma_1}[k] = \zug{\sigma_2}[k]$,  and
$q_{\zug{\sigma_1}[1, k-1]} = q_{\zug{\sigma_2}[1, k-1]}$, it holds that the transitions 
$t_1 = \zug{q, i} \xrightarrow{\sigma_1} \zug{p_1, i_1}$ and $t_2 = \zug{q, i} \xrightarrow{\sigma_2} \zug{p_2, i_2}$ are transitions of $\A^{G^{v_1}_p}$ that lead to the same state and agree on acceptance, in particular,  $t_1$ and $t_2$ do not induce nondeterministic choices from $\zug{q_{\zug{\sigma_1}[1, k-1]}, i}$ upon reading $\zug{\sigma_1}[k]$ in $e(\A^{G^{v_1}_p})$.

We proceed with showing that $e(\A^{G^{v_1}_p})$ is equivalent to $e(\A^{G^{v_1}})$.
Recall that $Q$ denotes the state-space of $\A^{G^{v_1}}$ and $Q_{v_1, p}$ denotes the state-space of $\A^{G^{v_1}_p}$.
Recall also that a word
$w \in (\{ 0, 1, \it{sleep}\} \cup \Sigma_{\$})^\omega$
respects the encoding when $w\in ( (  (0+1)^{\leq k-1} \cdot \Sigma_{\$})^* \cdot ((0+1)^k + \it{sleep})^*)^\omega$.
Then, note that also in the encoding $e(\A^{G^{v_1}_p})$, since only $\it{sleep}$-transitions from internal states not in $Q_{v_1, p}$  lead to the rejecting sink $q_{rej}$, 
since transitions labeled with special letters from internal states lead to $Q_{v_1, p}$, and since reading the last letter of an encoding $\zug{\sigma}[k]$ from an internal state of the form $\zug{q_{\zug{\sigma}[1, k-1]}, i}$ leads to a state in $Q_{v_1, p}$, we get that 
that a run from a state in $Q_{v_1, p}$ on an infinite word $w$ does not reach the rejecting sink if and only if $w$ respects the encoding. 
Now as $e(\A^{G^{v_1}_p})$ encodes transitions of $\A^{G^{v_1}_p}$ similarly to the way $e(\A^{G^{v_1}})$ encodes transitions of $\A^{G^{v_1}}$, then it is easy to verify that Lemma~\ref{cr lem} and Proposition~\ref{B sharp-free prop} extend to the setting of $e(\A^{G^{v_1}_p})$. In particular, we have the following:

\begin{lemma}\label{cr lem 2}
Consider the tDCW $\A^{G^{v_1}_p}$, its encoding $e(\A^{G^{v_1}_p})$,  states  $\zug{s, i_1}, \zug{q, i_2} \in Q_{v_1, p}$, and a letter $\sigma \in 2^{\mathsf{AP}}$. Then, 
\begin{enumerate}
    \item \sloppypar The run of $e(\A^{G^{v_1}_p})^{\zug{s, i_1}}$ on $\zug{\sigma}$ ends in a state in $Q_{v_1, p}$. In addition, if a run $r = \zug{s, i_1} \xrightarrow{\zug{\sigma} } \zug{q, i_2}$ exists in $e(\A^{G^{v_1}_p})$, then it is  unique, and is of the form $r = \zug{s, i_1}, \zug{s_{\zug{\sigma}[1, 1]}, i_1}, \zug{s_{\zug{\sigma}[1, 2]}, i_1}, \ldots, \zug{s_{\zug{\sigma}[1, k-1]}, i_1}, \zug{q, i_2}$.

    \item The triple $\zug{\zug{s, i_1}, \sigma, \zug{q, i_2}}$ is a transition of $\A^{G^{v_1}_p}$ iff there is a run $r = \zug{s, i_1} \xrightarrow{\zug{\sigma} } \zug{q, i_2}$ in $e(\A^{G^{v_1}_p})$. In addition, the run $r$ is safe only when the transition $\zug{\zug{s, i_1}, \sigma, \zug{q, i_2}}$ is safe.
    
\end{enumerate}
\end{lemma}

\begin{proof}
    The proof is similar to that of Lemma~\ref{cr lem} and is immediate from  the definitions. 
    Specifically,
    when reading $\zug{\sigma}$ from $\zug{s, i_1}$,  we move deterministically to the state $\zug{s_{\zug{\sigma}[1, j]}, i_1}$ via a safe transition after reading the $j$'th letter of $\zug{\sigma}$ for all $j\in [k-1]$.
    Then, by reading $\zug{\sigma}[k]$ from the state $\zug{s_{\zug{\sigma}[1, k-1]}, i_1}$, we either move to the state  $\zug{q, i_1}\in Q_{v_1, p}$  via a safe transition, when $\zug{\zug{s, i_1}, \sigma, \zug{q, i_1}}$ is a safe transition of $\A^{G^{v_1}_p}$, or traverse a rejecting transition that leads to the state $\zug{q, (i_1 \ mod (l+1)) + 1}$, when $\zug{\zug{s, i_1} , \sigma, \zug{q, (i_1 \ mod (l+1)) + 1}}$ is a rejecting transition of $\A^{G^{v_1}_p}$.  The fact that the run of $\zug{s, i_1}$ on $\zug{\sigma}$ is unique follows from the fact that $e(\A^{G^{v_1}_p})$ is deterministic. 
    Finally, note that the existence of the run $r = \zug{s, i_1} \xrightarrow{\zug{\sigma} } \zug{q, i_2}$ in $e(\A^{G^{v_1}_p})$ implies the existence of the transition $ \zug{\zug{s, i_1}, \sigma, \zug{q, i_2}}$ in $\A^{G^{v_1}_p}$ that agrees with $r$ on acceptance. Indeed, if reading $\sigma$ from $\zug{s, i_1}$ in $\A^{G^{v_1}_p}$ does not lead to $\zug{q, i_2}$ or does not agree with $r$ on acceptance, then the definition of $e(\A^{G^{v_1}_p})$ implies that there is a nondeterministic choice upon reading $\zug{\sigma}[k]$ from $\zug{s_{\zug{\sigma}[1, k-1]}, i_1}$; one choice is induced by the $\sigma$-transition of $\zug{s, i_1}$ in $\A^{G^{v_1}_p}$, and the other choice is traversed by the run $r$.
\end{proof}

\begin{proposition}\label{B sharp-free prop 2}
    Consider a state $\zug{s, i}\in Q_{v_1, p}$. It holds that 
    \begin{multline*}
     L(e(\A^{G^{v_1}_p})^{\zug{s, i}}) \cap \{0, 1, \it{sleep}\}^\omega = 
    \{ w\in \{ 0, 1, \it{sleep}\}^\omega: \\ \text{ $w$ respects the encoding and $dec(w) \in L((\A^{G^{v_1}_p}) ^{\zug{s, i}})$}\}
    \end{multline*}
\end{proposition}

\begin{proof}
The proof is very similar to that of Proposition~\ref{B sharp-free prop} and is immediate from  the definitions and the fact that $\A^{G^{v_1}_p}$ and $e(\A^{G^{v_1}_p})$ agree on $\it sleep$ transitions. Specifically, the latter proof applies also here if we replace Lemma~\ref{cr lem} with Lemma~\ref{cr lem 2}, and replace $Q$ with $Q_{v_1, p}$. 
\end{proof}

Recall that $Q$ denotes the state-space of the HD-tNCW $\A^{G^{v_1}}$, and $Q_{v_1, p}$ denotes the state-space of the tDCW $\A^{G^{v_1}_p}$.
We show next that  $L(e(\A^{G^{v_1}})) = L(e(\A^{G^{v_1}_p}))$. 
We show first that $L(e(\A^{G^{v_1}_p})) \subseteq L(e(\A^{G^{v_1}}))$.
    Consider a word $w\in L(e(\A^{G^{v_1}_p}))$. As there is an accepting run of $e(\A^{G^{v_1}_p})$ on $w$,   we have that $w$ respects the encoding. In particular, we can write $w = z_1\cdot z_2 \cdot z_3 \cdots$, where  for all $j\geq 1$, it holds that $z_j \in ((0+1)^{\leq k-1} \cdot \Sigma_\$ + (0+1)^k + \it{sleep})$.
    Let $r = \zug{r_0, i_0} \xrightarrow{z_1} \zug{r_1, i_1} \xrightarrow{z_2} \zug{r_2, i_2} \cdots$ be an accepting run of $e(\A^{G^{v_1}_p})$ on $w = z_1\cdot z_2 \cdot z_3 \cdots $. 
    By Lemma~\ref{cr lem 2}, the fact that $\A^{G^{v_1}_p}$ and $e(\A^{G^{v_1}_p})$ agree on $\it{sleep}$-transitions, and the fact that the initial state of $e(\A^{G^{v_1}_p})$ is in $Q_{v_1, p}$, we get that for all $j\geq 1$, it holds that $\zug{r_{j-1}, i_{j-1}}$ is a state in $Q_{v_1, p}$; in particular, $r_{j-1} $ is a state in $Q$.
    As $r$ is accepting, there is some $t\geq 0$ such that $r[t, \infty] = \zug{r_t, i_t} \xrightarrow{z_{t+1}} \zug{r_{t+1}, i_{t+1}} \xrightarrow{z_{t+2}} \zug{r_{t+2}, i_{t+2}} \cdots$ is a safe run on the suffix $z_{t+1}\cdot z_{t+2}\cdot z_{t+3}\cdots$ of $w$; in particular, $r[t, \infty]$ is stuck in some safe component $S \times \{ i\}$ of $e(\A^{G^{v_1}_p})$, and so $\{i \}= \{i_t, i_{t+1},i_{t+2}, \ldots \} $. Note that by definition, by projecting the run $r[t, \infty]$ on the left coordinate, we get a run of $r_t$ that is stuck in the safe component $S\in \SC(e(\A^{G^{v_1}}))$ on $z_{t+1}\cdot z_{t+2}\cdot z_{t+3}\cdots$. 
    Now consider an arbitrary run $r' = r'_0 \xrightarrow{z_1} r'_1 \xrightarrow{z_2} r'_2 \cdots \xrightarrow{z_t} r'_t$ of $e(\A^{G^{v_1}})$ on the prefix $z_1\cdot z_2\cdots z_t$ of $w$. We claim that $r'$ can be extended to an accepting run of $e(\A^{G^{v_1}})$ on $w$, and thus $w\in L(e(\A^{G^{v_1}}))$.
    Indeed, consider an arbitrary extension $r'_w$ of the finite run $r'$ to an infinite run on $w$: $r'_w = r'_0 \xrightarrow{z_1} r'_1 \xrightarrow{z_2} r'_2 \cdots \xrightarrow{z_t} r'_t \xrightarrow{z_{t+1}} r'_{t+1} \xrightarrow{z_{t+2}} r'_{t+2} \cdots$. 
    Note that by Lemma~\ref{cr lem} and the fact $\A^{G^{v_1}}$ and its encoding agree on $\it{sleep}$-transitions, we have that $r'_{j-1}$ is a state in $Q$ for all $j\geq 1$.
    Now if $r'_w$ is accepting, then we're done. Otherwise, let $j\geq t$ be such that $r'_j \xrightarrow{z_{j+1}} r'_{j+1}$ traverses a rejecting transition of $e(\A^{G^{v_1}})$. By the definition of $e(\A^{G^{v_1}})$, the run $r'_w$ can be modified to move to any state in $Q$ after reading the infix $z_{j+1}$ from the state $r'_j$. Indeed, while reading the infix  $z_{j+1}$ from $r'_j$, the encoding $e(\A^{G^{v_1}})$ traverses a rejecting transition upon reading the last letter of $z_{j+1}$,  and by definition, we can modify the run by proceeding nondeterministically to any  state in $Q$ upon reading the last letter of $z_{j+1}$.  In particular, the run $r'_w$ can be modified to move to the state $r_{j+1} \in S\cap Q$  instead of moving to the state $r'_{j+1}$. 
    From this point, the run $r'_w$ can stay in the safe component $S$ while reading the suffix $z_{j+2}\cdot z_{j+3} \cdots$ by simply following the projection of the safe run  $r[j+1, \infty ] = \zug{r_{j+1}, i_{j+1}} \xrightarrow{z_{j+2}} \zug{r_{j+2}, i_{j+2}} \xrightarrow{z_{j+3}} \zug{r_{j+3}, i_{j+3}} \cdots$ on the left coordinate. 
    In particular, $r'_w$ is accepting.

    We show next that $L(e(\A^{G^{v_1}}))\subseteq L(e(\A^{G^{v_1}_p}))$. Consider a word $w\in L(e(\A^{G^{v_1}}))$,  note that $w$ respects the encoding, and
    let $r' = r'_0 \xrightarrow{z_1} r'_1 \xrightarrow{z_2} r'_2, \ldots$ be an accepting run of $e(\A^{G^{v_1}})$ on $w$, where for all $j\geq 1$, it holds that $z_j \in ((0+1)^{\leq k-1} \cdot \Sigma_\$ + (0+1)^k + \it{sleep})$, and $r'_{j-1}$ is a state in $Q$.
    As $r'$ is accepting, there is some $t\geq 0$ such that $r'[t, \infty] = r'_t \xrightarrow{z_{t+1}} r'_{t+1} \xrightarrow{z_{t+2}} r'_{t+2}, \ldots$ is a safe run on the suffix $z_{t+1}\cdot z_{t+2}\cdots$ of $w$.
    Consider the run $r = r_0 \xrightarrow{z_1} r_1 \xrightarrow{z_2} r_2, \ldots $ of the tDCW $e(\A^{G^{v_1}_p})$ on $w$, and note that $r_{j-1}\in Q_{v_1, p}$ for all $j\geq 1$. We need to show that $r$ is accepting. 
    We distinguish between two cases. The first case is when the suffix $z_{t+1}\cdot z_{t+2}\cdots$ of $w$ is such that there is $j\geq t+1$ with  $z_j \in (0+1)^{\leq k} \cdot \$_{s_x}$, where $s\in Q$ and $x$ is a binary word of length at most $k-1$. 
   By the definition of $e(\A^{G^{v_1}_p})$, when the run $r$ reads the infix $z_j$, it moves a state of the form $r_j = \zug{s, i}$. Thus, $r_j$ is a copy of the state $s\in Q$. %
   Now as $r'[t, \infty]$ is safe, it follows by the definition of safe transitions that are labeled with special letters in $e(\A^{G^{v_1}})$ that $r'_j = s$. Hence,
   as the safe component $\SC(\zug{s, i})$ is a copy of the safe component $\SC(s)$, and as $r'_j = s$ and $r_j = \zug{s, i}$, it follows that from this point, 
   the run $r$ stays in the safe component $\SC(\zug{s, i})$, and thus is accepting.
   We proceed to the other case, namely the case where the suffix $z[t+1, \infty] = z_{t+1}\cdot z_{t+2}\cdots$ of $w$ has no special letters from $\Sigma_\$$. 
   Recall that $r'_t$ and $r_t$ are  states in $Q$ and $Q_{v_1, p}$, respectively. As the suffix $z[t+1, \infty]$ has no special letters,  $z[t+1, \infty ]\in L(e(\A^{G^{v_1}})^{r'_t})$, and as all states in $\A^{G^{v_1}}$ and all the states in $\A^{G^{v_1}_p}$ are equivalent,  we get by  Propositions~\ref{tDCW is equiv pro}, \ref{B sharp-free prop}, and \ref{B sharp-free prop 2} that 
   
    $$z[t+1, \infty]\in L(e(\A^{G^{v_1}})^{r'_t}) \cap \{0, 1, \it{sleep}\}^\omega $$ 
    $$= 
    \{ w\in \{ 0, 1, \it{sleep}\}^\omega: \text{ $w$ respects the encoding and $dec(w) \in L((\A^{G^{v_1}})^{r'_t})$}\}
    $$ 
    $$ =
    \{ w\in \{ 0, 1, \it{sleep}\}^\omega: \text{ $w$ respects the encoding and $dec(w) \in L((\A^{G^{v_1}_p})^{r_t})$}\} 
    $$ 
    $$=  L(e(\A^{G^{v_1}_p})^{r_t}) \cap \{0, 1, \it{sleep}\}^\omega$$

    In particular, the run of $e(\A^{G^{v_1}_p})^{r_t}$ on the suffix $z[t+1, \infty] = z_{t+1}\cdot z_{t+2} \cdots$ is accepting, and we're done.

\end{document}